%% file: main_elsarticle.tex
\documentclass[preprint,12pt]{elsarticle}

\usepackage[
    letterpaper,
    left=1in,
    right=1in,
    top=1in,
    bottom=1in
]{geometry}

\usepackage{amsmath,amssymb,amsfonts,mathtools,mathrsfs,amsthm}
\usepackage{amsopn}
\usepackage{dsfont}
\usepackage{bm}
\usepackage{xcolor}

\usepackage{graphicx}
\usepackage{algorithm,algorithmic}
\usepackage{tikz,subfig}
\usepackage{enumitem}
\usepackage{float}
\usepackage{euscript}
\usepackage{blkarray,bigstrut}
\usetikzlibrary{calc}
\usetikzlibrary{arrows.meta}
\setlist[itemize]{parsep=0pt, partopsep=0pt, topsep=5pt, itemsep=3pt}

\usepackage[colorlinks=true,linkcolor=blue,citecolor=blue,urlcolor=blue]{hyperref}
\usepackage[nameinlink,noabbrev]{cleveref}

\DeclareGraphicsExtensions{.eps,.pdf,.png,.jpg}

\newtheorem{theorem}{Theorem}[section]
\newtheorem{lemma}[theorem]{Lemma}
\newtheorem{proposition}[theorem]{Proposition}
\newtheorem{corollary}[theorem]{Corollary}

\theoremstyle{definition}
\newtheorem{definition}[theorem]{Definition}
\newtheorem{remark}[theorem]{Remark}

\crefname{hypothesis}{Hypothesis}{Hypotheses}

\newcommand{\red}[1] {\textcolor{red}{#1}}
\newcommand{\blue}[1] {\textcolor{blue}{#1}}
\newcommand{\green}[1] {\textcolor{green}{#1}}
\newcommand{\node}{\mathsf{v}}
\newcommand{\VerticesD}{V_D}
\newcommand{\EdgesD}{E_D}
\newcommand{\TreeD}{\mathscr{T}_{D}}

\newcommand{\PathsD}{\mathscr{P}_{D}}
\newcommand{\PathsDk}[1]{\mathscr{P}_{D,#1}}
\newcommand{\PosPathsDk}[1]{\mathscr{P}^{+}_{D,#1}}
\newcommand{\ppath}{\mathsf{p}}
\newcommand{\Def}{\overset{\text{def}}{=}}
\newcommand{\Order}{\mathscr{O}}
\newcommand{\Flip}{\mathscr{F}}
\newcommand{\pathprec}{\prec_{\mathrm{path}}}
\newcommand{\pathsucc}{\succ_{\mathrm{path}}}
\newcommand{\cardprec}{\prec_{\mathrm{card}}}
\newcommand{\cardsucc}{\succ_{\mathrm{card}}}
\newcommand{\bk}{\mathbf{k}}
\newcommand{\bOne}{\mathbf{1}}

\newcommand{\vvalue}{\mathfrak{v}}
\newcommand{\Z}{\mathbb{Z}}

\newcommand{\fitlabel}[1]{%
  \resizebox{1cm}{!}{\centering #1}%
}

\newcommand{\FUNCTION}[2]{\STATE \textbf{function} \textsc{#1}(#2)}
\newcommand{\ENDFUNCTION}{\STATE \textbf{end function}}
\newcommand{\CALL}[2]{\textsc{#1}(#2)}

\journal{}

\begin{document}

\begin{frontmatter}

\title{Path Enumeration by Position-Visit Counts in Recombining Trinomial Trees}

\author[uiuc-ie]{Ethan Torres\corref{cor1}}
\ead{ethanjt2@illinois.edu}

\author[uiuc-ie]{Ramavarapu Sreenivas}
\ead{rsree@illinois.edu}

\author[uiuc-ie,uiuc-math]{Richard Sowers}
\ead{r-sowers@illinois.edu}

\cortext[cor1]{Corresponding author}

\affiliation[uiuc-ie]{
  organization={Department of Industrial and Enterprise Systems Engineering, University of Illinois Urbana-Champaign},
  city={Urbana},
  state={IL},
  country={USA}
}

\affiliation[uiuc-math]{
  organization={Department of Mathematics, University of Illinois Urbana-Champaign},
  city={Urbana},
  state={IL},
  country={USA}
}

\begin{abstract}
Recombining trinomial trees are a workhorse for modeling discrete-event systems in option pricing, logistics, and feedback control. Because each node stores a state-dependent quantity, a depth-$D$ tree contains $3^D$ raw trajectories,
making exhaustive enumeration rapidly infeasible. However, when each node’s value depends only on its position, a raw trajectory's aggregate is determined by its position-visit counts. We call these count vectors cardinality tuples and decompose the admissible tuples into weak-composition mass layers. Leveraging these structures, we introduce a \textbf{mass-shifting enumeration algorithm} that slides integer ``masses'' through cardinality tuples to
generate exactly one representative of each path-equivalence class, while the accompanying weak-composition bijections yield exact counting formulas for the generated families. This suppresses redundant raw-path orderings a priori rather than enumerating and deduplicating them afterward. For the full-tuple implementation, we prove an output-sensitive running-time bound at each fixed endpoint, together with a uniform worst-case upper bound $\mathscr{O}(D2^D)$ and an exact worst-case exponential growth base of $2$, compared with base $3$ for exhaustive raw-path enumeration. Thus the construction achieves a provable exponential reduction in the enumeration space, up to polynomial factors. The same framework also recovers the information compressed by the equivalence classes: we derive an exact degeneracy formula for the number of raw paths represented by every cardinality tuple. We further prove that the nonnegative return specialization is exactly the classical Motzkin family, recover its recursive and generating-function structure and the Dyck specialization, and derive a multivariate occupation-profile $J$-fraction whose coefficients recover the corresponding cardinality-tuple degeneracies.
\end{abstract}


\begin{highlights}
\item Develops a combinatorial framework for enumerating unique paths.
\item Derives explicit counting formulas for constrained path families.
\item Gives a recursion-free construction that generates each cardinality tuple exactly once.
\item Relates the enumeration framework to classical discrete structures.
\end{highlights}

\end{frontmatter}

\section*{AI Usage Disclosure}
AI tools were used throughout this work to assist with developing and computationally verifying code implementing the algorithms, carrying out mathematical computations and checking derivations, generating the code for Figures 7 and 8, and revising or developing manuscript text. The authors conceived the ideas underlying the work. In some cases, AI was used to revise material already drafted by the authors. In others, the authors first formulated and clarified an idea, then used AI to help develop it into a proposed statement or exposition through an iterative process of review and refinement. No AI-generated material was adopted without review, and the authors independently checked and verified the mathematical statements, code, computations, figures, and text. The authors take responsibility for and stand behind all work and results reported in this paper. The authors independently conceived the bijection, its connection to weak compositions, and the mathematical algorithms underlying the construction; AI assistance supported their implementation and presentation, as described above, but did not originate these core ideas.

\section{Introduction}
\label{introduction}
Recombining trinomial trees themselves exist as an approximation of a branching process. This is apparent since it is not always true that a process will return to a previous value after evolving to one in the future. However, these structures are well researched as approximations in the context of discrete event control problems or option pricing where approximations such as these are more strongly held \cite{clifford2008trinomial,crack2024trinomial,dai2000asian}. Although recombination keeps the underlying graph compact, it does not prevent exponential growth in the number of raw paths. With that aside, even when a tree recombines, the number of distinct path histories can grow extremely quickly with depth. In a recombining trinomial tree, many of these paths pass through the same states while differing in the order in which those states are visited, so the principal combinatorial difficulty lies in the proliferation of path histories rather than in the number of distinct states themselves.
General algorithms for combinatorial generation provide relevant background for this particular problem; see, for example, Knuth’s treatments \cite{knuth2011taocp4a,knuth2022taocp4b}.

In this paper, we aim to exploit the graphical structure of the recombining tree in order to reduce the combinatorial explosion associated with full path enumeration. The key property is that node values depend only on position when considering their aggregate sum: any two paths with the same position-visit counts have the same path sum, regardless of the order of visits. Consider the example shown in \Cref{fig:invariance}. The recombining tree consists of 25 nodes, each identified by a pair of integers: \emph{depth} (ranging from 0 at the root to 4 at the terminal level) and \emph{position} (indicating horizontal placement). Values are assigned purely by position, so all nodes at the same position share the same value. For instance, the root at position 0 has value 20, while its children at positions -1, 0, and 1 take values 18, 20, and 22. Paths then accumulate these position-based values from root to terminal nodes. In effect this is just a histogram representation—the node values themselves are illustrative and not essential.

Let's take a look at \Cref{fig:invariance} and consider the terminal node ending at position 0.  The blue, green, and red paths in \Cref{fig:invariance} each sum to $102$.  Each path visits the nodes with value 22 once, and visits the node with values 20 4 times; the sum is 
\begin{equation}\label{E:LebesgueSum} 4\times 20 + 1\times 22=102. \end{equation}
Generally, our setup allows for multiple paths to have the same accumulated value. Trees are often used to discretize path integrals. In our case, we want to average the path--sum of values (e.g., terms like \(102\)) over all paths terminating at position~0. A na\"{\i}ve method would enumerate every such path, compute its path--sum, and then average---quickly becoming infeasible as the tree depth grows. Instead, paths that visit each position the same number of times can be grouped by their common position-visit profile. Because
node values depend only on position, every path in such a group
has the same path--sum. If several different profiles produce the
same numerical value \(p\), their multiplicities can subsequently
be combined into a single count \(N_p\). The average is then
\[
    \text{Average}
    =
    \frac{\sum_p p\,N_p}{\sum_p N_p}.
\]
This has the same ``values times measures'' interpretation as the
Lebesgue-style viewpoint above, but the enumeration is carried out
over structural visit profiles rather than over the full collection
of raw path histories.

\begin{figure}[H]
\centering
\begin{minipage}{0.7\textwidth}
\centering
\begin{tikzpicture}[
    scale=0.3, transform shape, >=Stealth,
    every node/.style={
        circle, draw,
        minimum size=1.5cm,
        inner sep=0pt,
        align=center,
        font=\large        
    },
    node distance = 2cm and 1.2cm
    ]
    \node (n0) at (0,0) {\fitlabel{20}};
    
    \node (n1u) at (3,  2) {\fitlabel{22}};
    \node (n1m) at (3,  0) {\fitlabel{20}};
    \node (n1d) at (3, -2) {\fitlabel{18}};
    
    \node (n2uu) at (6,  4) {\fitlabel{24}};
    \node (n2um) at (6,  2) {\fitlabel{22}};
    \node (n2mm) at (6,  0) {\fitlabel{20}};
    \node (n2dm) at (6, -2) {\fitlabel{18}};
    \node (n2dd) at (6, -4) {\fitlabel{16}};
    
    \node (n3uuu) at (9,  6) {\fitlabel{26}};
    \node (n3uum) at (9,  4) {\fitlabel{24}};
    \node (n3umm) at (9,  2) {\fitlabel{22}};
    \node (n3mmm) at (9,  0) {\fitlabel{20}};
    \node (n3dmm) at (9, -2) {\fitlabel{18}};
    \node (n3ddm) at (9, -4) {\fitlabel{16}};
    \node (n3ddd) at (9, -6) {\fitlabel{14}};
    
    \node (n4uuuu) at (12,  8) {\fitlabel{28}};
    \node (n4uuum) at (12,  6) {\fitlabel{26}};
    \node (n4uumm) at (12,  4) {\fitlabel{24}};
    \node (n4ummm) at (12,  2) {\fitlabel{22}};
    \node (n4mmmm) at (12,  0) {\fitlabel{20}};
    \node (n4dmmm) at (12, -2) {\fitlabel{18}};
    \node (n4ddmm) at (12, -4) {\fitlabel{16}};
    \node (n4dddm) at (12, -6) {\fitlabel{14}};
    \node (n4dddd) at (12, -8) {\fitlabel{12}};
    
    \draw[red, thick] (n0) -- (n1u);
    \draw[cyan, thick] (n0) -- ($(n0)!0.5!(n1m)$);
    \draw[green, thick] ($(n0)!0.5!(n1m)$) -- (n1m);
    \draw (n0) -- (n1d);
    
    \draw (n1u) -- (n2uu);
    \draw (n1u) -- (n2um);
    \draw[red, thick] (n1u) -- (n2mm);
    
    \draw[green, thick] (n1m) -- (n2um);
    \draw[cyan, thick] (n1m) -- (n2mm);
    \draw (n1m) -- (n2dm);
    
    \draw (n1d) -- (n2mm);
    \draw (n1d) -- (n2dm);
    \draw (n1d) -- (n2dd);
    
    \draw (n2uu) -- (n3uuu);
    \draw (n2uu) -- (n3uum);
    \draw (n2uu) -- (n3umm);
    
    \draw (n2um) -- (n3uum);
    \draw (n2um) -- (n3umm);
    \draw[green, thick] (n2um) -- (n3mmm);
    
    \draw[cyan, thick] (n2mm) -- (n3umm);
    \draw[red, thick] (n2mm) -- (n3mmm);
    \draw (n2mm) -- (n3dmm);
    
    \draw (n2dm) -- (n3mmm);
    \draw (n2dm) -- (n3dmm);
    \draw (n2dm) -- (n3ddm);
    
    \draw (n2dd) -- (n3dmm);
    \draw (n2dd) -- (n3ddm);
    \draw (n2dd) -- (n3ddd);
    
    \draw (n3uuu) -- (n4uuuu);
    \draw (n3uuu) -- (n4uuum);
    \draw (n3uuu) -- (n4uumm);
    
    \draw (n3uum) -- (n4uuum);
    \draw (n3uum) -- (n4uumm);
    \draw (n3uum) -- (n4ummm);
    
    \draw (n3umm) -- (n4uumm);
    \draw (n3umm) -- (n4ummm);
    \draw[cyan, thick] (n3umm) -- (n4mmmm);
    
    \draw (n3mmm) -- (n4ummm);
    \draw[red,   thick] (n3mmm) -- ($(n3mmm)!0.5!(n4mmmm)$);   
    \draw[green, thick] ($(n3mmm)!0.5!(n4mmmm)$) -- (n4mmmm);
    \draw (n3mmm) -- (n4dmmm);
    
    \draw (n3dmm) -- (n4mmmm);
    \draw (n3dmm) -- (n4dmmm);
    \draw (n3dmm) -- (n4ddmm);
    
    \draw (n3ddm) -- (n4dmmm);
    \draw (n3ddm) -- (n4ddmm);
    \draw (n3ddm) -- (n4dddm);
    
    \draw (n3ddd) -- (n4ddmm);
    \draw (n3ddd) -- (n4dddm);
    \draw (n3ddd) -- (n4dddd);

    \draw[<-, thick] (-2.5, -9.5) -- (14, -9.5); 
    \draw[->, thick] (14, -9.5) -- (14, 9);      

    \pgftext[left, at={\pgfpoint{6cm}{-10.2cm}}]{\large Depth}
    
    \pgftext[bottom, at={\pgfpoint{14.7cm}{0cm}}, rotate=90]{\large Position}
    
\end{tikzpicture}
\caption{Recombining Tree with Values}
\label{fig:invariance}
\end{minipage}
\end{figure}
Note that in this figure we explicitly display the value stored at each position. In every other figure, we omit these values as they are implicitly present, and—under the labeling scheme introduced shortly—each node retains the same value throughout. Suppressing them in later diagrams keeps the visual presentation clear and uncluttered. We now introduce our notation for the recombining tree truncated at depth \(D\) as \(\TreeD\), and \(\PathsD\) for the set of root-to-depth-\(D\) paths in \(\TreeD\). We also define an informal reconstruction function $\EuScript{C}(\PathsD)$, where $\EuScript{C}(\PathsD)$ denotes all the possible combinations of paths, which constructs the tree by taking the union of all paths---i.e., $\TreeD = \EuScript{C}(\PathsD)$. This is formally addressed later in \Cref{setup}.

This relation highlights that the full tree can be viewed as arising from the collection of all possible path histories. Due to the recombining nature of the tree, many of these paths are structurally redundant for the quantities considered here. In particular, different paths may visit the same positions with the same multiplicities, but in different orders, without altering the resulting path-dependent quantity. We refer informally to these redundant reorderings as permutations of paths and denote them by $\EuScript{P}(\PathsD)$.
This motivates the following schematic view of the unique path structures:
\begin{align}
\label{unique_formula}
\EuScript{U}(\PathsD)
=
\EuScript{C}(\PathsD)
-
\EuScript{P}(\PathsD).
\end{align}
Here the subtraction is intentionally informal: it represents removing path histories that differ only through such redundant reorderings. Although one could first generate all paths and then remove these redundancies, doing so would still incur the full cost of raw-path enumeration. Instead, our construction is designed around the opposite viewpoint: identify the structure responsible for the redundancy first, and then generate only one representative cardinality profile from each resulting class. In this sense, the algorithm \emph{pre-prunes} the redundant path orderings rather than discovering them after exhaustive enumeration.

Although the argument here is informal, it illustrates the core contribution of our method: a principled way to generate only the distinct path structures in a recombining tree, without redundancy. This offers substantial advantages in discrete event control problems and other applications involving path-dependent dynamics where recombining trees are used to model system evolution. Our approach draws on decades of work aimed at accelerating path enumeration, including the algorithms of Eades and McKay \cite{eades1984algorithm}, Ehrlich’s loop‑less generation technique \cite{ehrlich1973loopless}, and the Steinhaus–Johnson–Trotter algorithm originally presented by Johnson \cite{Johnson1963}. We also refer to work by Ruskey and Williams \cite{ruskey2005prefix} and Cheng \cite{cheng2007generating} for more recent work in this particular field.

The question then arises, why specifically use a trinomial tree? A recombining trinomial lattice improves on the classic binomial grid by adding a "stay-put" or "no-change" branch, delivering countless benefits. With two free probabilities per step, it can match both the drift and the variance of the underlying process, giving second-order weak accuracy instead of the binomial's first-order, so far fewer time steps - and hence far fewer total nodes - are needed to hit a given error tolerance \cite{MaZhu2015}. The trinomial structure is also the discrete analog of a central difference scheme, which is unconditionally stable for many payoffs, in deference to the application in option pricing. This stability lets users of any algorithm that utilizes this underlying structure to stretch time increments two-to-four-fold without inducing oscillations \cite{AhnSong2007}. In control problems, the three branches also line up perfectly with the canonical actions "increase, hold, decrease" eliminating the dummy states a binomial grid must invent. In summary, a third branch provides just enough freedom to hit both drift and volatility, yielding higher accuracy, greater numerical stability, cleaner boundary handling, and better scalability - all while keeping the lattice recombining. 

\section{Setup}
\label{setup}
We model a recombining rooted trinomial tree as a directed graph embedded in $\Z\times\Z_+$, where
\begin{equation*}
\Z_+ \Def \{0,1,2,\dots\}.
\end{equation*}
Fix a depth $D\in\Z_+$. The vertex set is
\begin{equation}
\label{V:VertexDef}
\VerticesD\Def\bigl\{(k,d)\in\Z\times\Z_+:0\le d\le D,\ |k|\le d\bigr\},
\end{equation}
and the directed edge set is
\begin{equation}
\label{E:EdgeDef}
\begin{aligned}
    \EdgesD \Def \bigl\{&
        \bigl((k,d-1),(k+s,d)\bigr):
        \\
        &d\in\{1,\dots,D\},\quad
        k\in\Z,\quad |k|\le d-1,\quad
        s\in\{-1,0,1\}
    \bigr\}.
\end{aligned}
\end{equation}
For a generic node $(k,d)$  in $V_D$, we will think of $d$ as the \emph{depth} coordinate and $k$ as the \emph{position} coordinate. In addition, $\EdgesD$ encodes the rule that from any node at position $k$ and depth $d-1$, we may move to depth $d$ by stepping to $k-1, k, k+1$. The resulting graph is denoted $\TreeD = (\VerticesD, \EdgesD)$, a formalization of the tree described in \Cref{introduction}, and is symmetric under reflection across the position axis $k = 0$, a property that will be important in later sections. We show a visualization of this graph in \Cref{fig:nodelabels} to build the graphical intuition for the reader:
\begin{figure}[H]
\centering
\begin{tikzpicture}[node distance=2cm, every node/.style={circle, draw, minimum size=0.8cm}, >=Stealth]
    \node (0) at (0, 4) {0,0};   
    \node (1) at (-1.5, 2.5) {-1,1}; 
    \node (2) at (0, 2.5) {0,1};  
    \node (3) at (1.5, 2.5) {1,1};  
    \node (4) at (-3, 1) {-2,2}; 
    \node (5) at (-1.5, 1) {-1,2}; 
    \node (6) at (0, 1) {0,2};  
    \node (7) at (1.5, 1) {1,2};  
    \node (8) at (3, 1) {2,2};  

    \draw [->] (0) -- (1); 
    \draw [->] (0) -- (2); 
    \draw [->] (0) -- (3); 

    \draw [->] (1) -- (4); 
    \draw [->] (1) -- (5); 
    \draw [->] (1) -- (6); 

    \draw [->] (2) -- (5); 
    \draw [->] (2) -- (6);
    \draw [->] (2) -- (7);

    \draw [->] (3) -- (6); 
    \draw [->] (3) -- (7);
    \draw [->] (3) -- (8);
\end{tikzpicture}
\caption{$\mathscr{T}_{D}$ with Node-Labelings for $V_{D}$}
\label{fig:nodelabels} 
\end{figure}
We are interested in accumulating values along the \emph{paths} in the tree.
A path in $\TreeD$ is a sequence $(\node_0,\node_1,\dots,\node_D)$ of vertices such that $\node_0=(0,0)$, each $\node_d\in\VerticesD$, and $(\node_{d-1},\node_d)\in\EdgesD$ for $d\in\{1,2,\dots,D\}$. We provide an example of such a path in \Cref{fig:recombining_tree}:
\begin{figure}[H]
\centering
\begin{tikzpicture}[node distance=2cm, every node/.style={circle, draw, minimum size=0.8cm}, >=Stealth]
    \node (0) at (0, 4) {0,0};   
    \node (1) at (-1.5, 2.5) {-1,1}; 
    \node (2) at (0, 2.5) {0,1};  
    \node (3) at (1.5, 2.5) {1,1};  
    \node (4) at (-3, 1) {-2,2}; 
    \node (5) at (-1.5, 1) {-1,2}; 
    \node (6) at (0, 1) {0,2};  
    \node (7) at (1.5, 1) {1,2};  
    \node (8) at (3, 1) {2,2};  

    \draw [->, color=red, thick] (0) -- (1); 
    \draw [->] (0) -- (2); 
    \draw [->] (0) -- (3); 

    \draw [->] (1) -- (4); 
    \draw [->, color=red, thick] (1) -- (5); 
    \draw [->] (1) -- (6); 

    \draw [->] (2) -- (5); 
    \draw [->] (2) -- (6);
    \draw [->] (2) -- (7);

    \draw [->] (3) -- (6); 
    \draw [->] (3) -- (7);
    \draw [->] (3) -- (8);
\end{tikzpicture}
\caption{Recombining Tree with Sample Path}
\label{fig:recombining_tree} 
\end{figure}

Let $\PathsD$ be the collection of all depth-$D$ paths in $\TreeD$. Each such path can be represented by its sequence of step increments. For each $(s_1,\dots,s_D)\in\{-1,0,1\}^D$, define
\begin{equation}
\label{E:walktopathA}
\node_d \Def \left(\sum_{d'=1}^{d}s_{d'},\,d\right),
\quad d\in\{0,\dots,D\},
\end{equation}
where an empty sum is zero. The resulting sequence satisfies
\begin{equation}
\label{E:walktopathB}
\ppath\Def(\node_0,\dots,\node_D)\in\PathsD.
\end{equation}
Indeed, $\node_0=(0,0)$,
\begin{equation*}
\left|\sum_{d'=1}^{d}s_{d'}\right|\le\sum_{d'=1}^{d}|s_{d'}|
\le d,
\end{equation*}
and consecutive positions differ by an element of $\{-1,0,1\}$. Conversely, for $\ppath=(\node_0,\dots,\node_D)\in\PathsD$, write $\node_d=(k_d,d)$, with $k_0=0$, and define
\begin{equation*}
s_d \Def k_d-k_{d-1},\quad d\in\{1,\dots,D\}.
\end{equation*}
The edge condition implies $s_d\in\{-1,0,1\}$. These constructions are inverse to one another: differences of
successive partial sums recover the increments, while
\begin{equation*}
\sum_{d'=1}^{d}(k_{d'}-k_{d'-1})=k_d-k_0=k_d
\end{equation*}
recovers each position. Therefore, $|\PathsD|=|\{-1,0,1\}^D|=3^D.$ The position projection
$\pi:\PathsD\to\Z^{D+1}$ is defined by
\begin{equation}
\label{E:piDef}
\pi(\ppath)\Def(k_0,\dots,k_D)=(0,k_1,\dots,k_D).
\end{equation}
The vertex path is recovered from its position sequence by restoring the depth coordinate: the entry $k_d$ corresponds to the vertex $(k_d,d)$. Throughout, depth counts steps. Thus a depth-$D$ path contains $D+1$ vertex visits, including the initial and terminal vertices. When $D=0$, the increment sequence is empty, the sole vertex path is $((0,0))$, and its position sequence is $(0)$.

\section{Aggregation}
Let's return to \Cref{fig:invariance}. The blue, green, and red
paths all sum to $102$ and end at the position-$0$ terminal node.
Their common sum can be written as
\[
    20\times4+22\times1=102,
\]
reflecting the fact that each path visits position $0$ four times
and position $1$ once. Rather than treating these three raw paths
as separate objects, we represent their common pattern of position
visits once. The construction developed below formalizes these
visit profiles as cardinality tuples and enumerates the distinct
profiles directly. To ground the discussion, we open this section with the most straightforward technique for indexing (i.e., enumerating) every realization of a recombining trinomial tree: a depth-first recursive traversal that visits all branches until the target depth $D$. This classical approach was extremely important throughout our experiments because it served as a very stable baseline to check almost all of our algorithms for correctness. And while recursion does scale exponentially,  this baseline serves three goals
\begin{enumerate}
    \item  its transparency made manual verification easy;
    \item  it provided a clean performance yard-stick for the more sophisticated algorithms introduced later; and
    \item  it exposed the combinatorial structure of the tree in the clearest possible way.
\end{enumerate} 
Our implementation also includes the hash-map variant in
\Cref{alg:gencomb_hashing}. The map stores the value associated with each reachable position, so repeated visits to the same position can reuse the corresponding lookup rather than reconstructing that local information. It is therefore a table-backed traversal, where different prefixes reaching the same vertex must remain distinct when all raw paths are required. This provides a useful intermediate representation: repeated local structure can be reused without altering the exhaustive path family generated by the recursive search. For a depth-$D$ trinomial tree, the recursive traversal explores
\[
    \sum_{d=0}^{D}3^d
    =
    \frac{3^{D+1}-1}{2}
    =
    \Theta(3^D)
\]
recursive states before filtering by the terminal position.
For a fixed endpoint $k^*$, only the
$|\PathsDk{k^*}|$ retained paths are copied to the output,
and each such path contains $D+1$ vertices. Hence the
explicit-output runtime is
\begin{equation}
\label{E:naivetime}
    \mathrm{Time}_{\mathrm{DFS}}(D,k^*)
    =
    \Theta\!\left(
        3^D+(D+1)|\PathsDk{k^*}|
    \right).
\end{equation}
In particular, since $|\PathsDk{k^*}|\le3^D$, this gives the uniform bound
\[
    \mathrm{Time}_{\mathrm{DFS}}(D,k^*)
    =
    \Order(D3^D)
    \quad (D\ge1).
\]
Using a reusable path buffer and expected constant-time hash
lookups, the underlying traversal work remains
\begin{equation}
\label{E:memoisedtime}
    \text{Time}_{\text{traversal}}
    =
    \Theta\!\left(3^D\right),
\end{equation}
excluding the cost of materializing separate copies of the
paths retained at the prescribed endpoint. Thus the hash table reduces repeated local value-access work, but does not remove the exponential number of raw path prefixes that must be traversed. These baselines separate two sources of computational cost: the recursive control flow used to traverse the tree and the repeated structure encountered during that traversal. We first address the former by removing recursion while preserving the same admissible path family.

\section{Removing Recursion from the Enumeration Process}
\label{recursiveless}
Recursive traversal is a classical and transparent way to
enumerate a tree, but repeated function calls and stack
management can introduce substantial overhead in large
combinatorial searches \cite{irons1961recursive}. A depth-$D$
DFS maintains at most $\Order(D)$ simultaneously active stack
frames and reuses them as it backtracks through the tree.
Motivated by this observation, we next replace the recursive
control flow with explicit state updates while preserving the
same admissible path family. Our recursion-free construction is designed to (i) exhaustively enumerate all valid paths,
(ii) produce no spurious output, (iii) avoid recursive function calls and their associated control-flow overhead, and
(iv) provide the algorithmic scaffold that will later support
the unique-representative construction in \Cref{uniquepaths}.

Once paths are projected via \Cref{E:piDef}, it is natural to impose an order that supports a nonrecursive, “ping-pong’’ traversal: rather than recursing, we reset a loop cursor to a designated index and deterministically march through the next valid configuration. To set the stage, we introduce a canonical indexing by terminal node. Fix \(D\in\Z_{+}\) and, for any \(k^*\in\{-D,-D+1,\dots,D\}\), define the set of all admissible depth-\(D\) paths that terminate at \((k^*,D)\):
\[
  \PathsDk{k^*}
  \;\Def\;
  \bigl\{
    (\node_{0},\node_{1},\dots,\node_{D})\in\PathsD
    :\ \node_{D}=(k^*,D)
  \bigr\}.
\]
This viewpoint lets us either sweep over all attainable \(k^*\) for a fixed depth \(D\), or restrict to a single target \(k^*\) when only that terminal node is of interest—providing flexibility in what the generator emits. We impose a total order via the position projection \(\pi\) of \Cref{E:piDef}. For distinct paths $\ppath,\ppath'\in\PathsD$, write
\[
\pi(\ppath)=(k_0,\dots,k_D),
\quad
\pi(\ppath')=(k'_0,\dots,k'_D).
\]
Then
\begin{equation}
\label{S:LexOrd}
\ppath\pathprec\ppath'
\quad\Longleftrightarrow\quad
k_{d^\star}<k'_{d^\star},
\end{equation}
where
\[
d^\star
\Def
\min\left\{
d\in\{0,\dots,D\}:k_d\ne k'_d
\right\}.
\]
For distinct paths, we write
$\ppath\pathsucc\ppath'$ when
$\ppath'\pathprec\ppath$.
In practice, this ordering mirrors ``walking down'' the tree while updating only a small number of coordinates at each step. For the remainder of the nonnegative-path construction, assume $0\le k^*\le D.$ Negative endpoints are recovered afterward by global reflection, so no generality is lost in developing the generator for nonnegative terminal positions. For implementation, we first enumerate the paths that remain nonnegative at every depth. We refer to these as the
\emph{positive-path family} and define
\begin{equation}
\label{eq:positive-path-family}
  \PosPathsDk{k^*}
  \Def
  \left\{
    \ppath\in\PathsDk{k^*}:
    \pi(\ppath)\in\Z_+^{D+1}
  \right\}.
\end{equation}
Under the ordering \eqref{S:LexOrd}, the lexicographically
\emph{lowest} positive path has position sequence
\[
  \bigl(\underbrace{0,\dots,0}_{D+1-k^*},
  \,1,2,\dots,k^*\bigr).
\]
The lexicographically \emph{highest} positive path—i.e., the \(\pathprec\)-maximizer in \(\PosPathsDk{k^*}\)—has \(\pi\)-image
\begin{equation}\label{E:highestpath}
  \begin{cases}
    \bigl(0,1,2,\dots,\tfrac{D+k^*}{2},\ \tfrac{D+k^*}{2}-1,\dots,k^*\bigr),
    & \text{if } D-k^*\ \text{is even},\\[6pt]
    \bigl(0,1,2,\dots,\tfrac{D+k^*-1}{2},\ \tfrac{D+k^*-1}{2},\ \tfrac{D+k^*-1}{2}-1,\dots,k^*\bigr),
    & \text{if } D-k^*\ \text{is odd},
  \end{cases}
\end{equation}
i.e., an initial monotone rise that peaks at \(\bigl\lfloor\tfrac{D+k^*}{2}\bigr\rfloor\) and then (possibly after a single stay at the peak in the odd case) descends to \(k^*\). This position sequence serves as the maximal initial path for the nonrecursive raw-path enumeration.

\paragraph{Seed--and--march (recursion-free control)}
Beginning from the maximal path in \eqref{E:highestpath},
\Cref{alg:recursivelessgen} alternates a tick--down with a
sweep across the remaining depths. The tick--down selects
the largest internal depth $j\in\{1,\dots,D-1\}$ satisfying
\begin{equation}
\label{eq:raw-tick-feasibility}
    k_j>\max\{0,\,k_{j-1}-1,\,k^*-(D-j)\}.
\end{equation}
It then decreases $k_j$ by one. This condition keeps the
prefix nonnegative, preserves its last edge, and leaves
sufficient steps to reach $k^*$. If no such $j$ exists,
the traversal is complete. The sweep rebuilds the greatest feasible continuation of that prefix, proceeding from left to right:
\begin{equation}
\label{eq:raw-suffix-completion}
    k_d\leftarrow\min\{k_{d-1}+1,\,k^*+D-d\},
    \quad d=j+1,\dots,D.
\end{equation}
Only the completed path is recorded. The first term limits
how far the next step can rise; the second preserves the
ability to reach the endpoint. Together these operations
give the next path in descending order under \eqref{S:LexOrd}. The initialization and correctness
argument accompany \Cref{alg:recursivelessgen,alg:init_array}.
For comparison, retaining only the nonnegative paths from
\Cref{alg:gencomb_dfs,alg:gencomb_hashing} gives the same
path family. With their child order $(-1,0,1)$, those
recursive traversals list the family in increasing order;
the seed--and--march routine lists it in decreasing order.
The subsequent excursion construction extends this
nonnegative family to paths that visit negative positions.

\paragraph{From positive representatives to all paths}
Paths reaching $k^*$ may also visit negative positions. The excursion construction in \Cref{appendix:excursions} recovers them from the nonnegative representatives generated by \Cref{alg:recursivelessgen}. For each such representative $\ppath$, the flip family $\Flip(\pi(\ppath))$ contains the representative itself through the choice $A=\varnothing$ together with every path obtained by reflecting a subset of its unlocked excursions. By \Cref{prop:flip-representation}, these flip families form a disjoint decomposition of the full endpoint family. Operationally, one therefore emits each member of $\Flip(\pi(\ppath))$ exactly once in any fixed subset order.
Equivalently, one may emit $\ppath$ first and then emit only the $2^{\nu-\epsilon(\bk)}-1$ nontrivial flips.

\paragraph{Parity for seeding}
If a path has \(j_+\) up-steps, \(j_-\) down-steps, and \(j_0\) stays, then the relations in \Cref{E:algebra} and the parity constraint \Cref{eq:parity-j0} (\Cref{appendix:step-counts}) ensure that the chosen seed obeys the even–odd compatibility between \(D\), \(k^*\), and \(j_0\). In particular, the maximal seed \Cref{E:highestpath} is consistent with admissible counts and enables a deterministic, recursion-free enumeration pipeline: lexicographically generate \(\PosPathsDk{k^*}\) (\Cref{alg:recursivelessgen}) from the maximal seed, then expand each representative by excursion flips to obtain \(\PathsDk{k^*}\), preserving determinism and avoiding stack-based recursion, while remaining compatible with the unique-representative machinery in \Cref{uniquepaths}.

While our excursion-based treatment and the permutation of negatives across excursion blocks (Appendix \ref{appendix:excursions}) are conceptually interesting, we relegate their detailed mechanics to the appendices to avoid obscuring the core contributions. What matters here are the ideas of \emph{maximal paths} and \emph{lexicographical ordering} (see \Cref{S:LexOrd}), which not only eliminate recursion from the forthcoming algorithms but also guide the combinatorial structure underlying them. By linking “walking down’’ the graph to “marching through’’ path tuples, these tools provide intuitive evidence for the correctness of our approach. With this groundwork laid, we can now turn to the main problem: enumerating the distinct cardinality classes directly, without first generating their raw-path representatives.

\section{Enumeration of Unique Path Combinations}
\label{uniquepaths}

As before, every path $\ppath\in\PathsD$ has the position projection
$\pi(\ppath)$ defined in \Cref{E:piDef}. Our goal here is to enumerate only the \emph{unique} paths in the sense developed below: paths are identified whenever they have the same multiplicities of visited positions, rather than by the ordering of those visits. Such repeated position patterns proliferate rapidly with depth, creating substantial redundancy. Removing this redundancy does not make the problem polynomial, but it reduces the worst-case exponential growth base of the direct enumeration from $3$ to $2$, up to polynomial factors, as shown in \Cref{computation_complexity}.

The recursion-free raw-path construction in \Cref{recursiveless} provides two ingredients that remain useful: a maximal representative and a deterministic ordering of admissible paths. The unique-path construction developed here operates at a different level. Rather than marching through raw paths and removing duplicates afterward, we work directly with the multiplicities of visited positions. Paths terminating at the same node $(k^*,D)$ may have the same multiset of position values, and therefore the same aggregate value, even though their visits occur in different orders (\Cref{fig:invariance}). We represent this multiset by a \emph{cardinality tuple}, or count vector, recording how many times each position is visited. The construction below then has three stages:
\begin{enumerate}
    \item \textbf{Cardinality representation:}
    Map each raw path to its position-visit count vector and identify
    paths having the same vector.

    \item \textbf{Direct positive-tuple generation:}
    Construct the admissible nonnegative cardinality tuples directly
    through the seed, mass-layer, and weak-composition structure,
    without first enumerating their raw-path representatives.

    \item \textbf{Stage extension to negative support:}
    Partition the remaining tuples by their lowest occupied position
    and generate each stage through its canonical seed and associated
    positive-tuple reduction.
\end{enumerate}
Because the algorithm is constructed in this manner, duplicates are suppressed \emph{a priori}: they are never generated and subsequently removed. The resulting procedure remains non-recursive while substantially reducing the number of objects that must be generated. The exact combinatorial counts and the resulting complexity reduction are derived below.

\subsection{Closed-Form Combinatorics for Unique Path Enumeration}
\label{subsec:closed_form_cardinality}

Fix $D\in\Z_+$ and $k^*\in\Z$ with $|k^*|\le D$.
Throughout the following definitions, intervals used to index
positions contain only their integer points. The positions that can occur among paths in $\PathsDk{k^*}$
form the common index range $[k_-,k_+]$, where
\begin{equation}
\label{E:Kbounds}
    k_-
    =
    \left\lceil\frac{k^*-D}{2}\right\rceil,
    \quad
    k_+
    =
    \left\lfloor\frac{D+k^*}{2}\right\rfloor.
\end{equation}
These bounds depend only on $D$ and $k^*$; they are not the
minimum and maximum positions visited by an individual path.
For $\ppath\in\PathsDk{k^*}$ with
$\pi(\ppath)=(k_0,\dots,k_D)$, define
\begin{equation*}
    c_k(\ppath)
    \Def
    \left|
        \left\{
            d\in\{0,\dots,D\}: k_d=k
        \right\}
    \right|,
    \quad k\in[k_-,k_+].
\end{equation*}
Thus $c_k(\ppath)$ counts visits to position $k$, including
visits at the initial and terminal depths. The cardinality map is
\begin{equation}
\label{card_map}
    \hat{C}_{D,k^*}:
    \PathsDk{k^*}
    \longrightarrow
    \Z_+^{[k_-,k_+]},
    \quad
    \hat{C}_{D,k^*}(\ppath)
    =
    \bigl(c_k(\ppath)\bigr)_{k=k_-}^{k_+}.
\end{equation}
Here $\Z_+^{[k_-,k_+]}$ denotes the set of nonnegative-integer
tuples with one coordinate for each position in $[k_-,k_+]$.
With $D$ and $k^*$ fixed, we abbreviate $\hat{C}_{D,k^*}(\ppath)$ as $\hat C(\ppath)$ and write its value as $\hat c$:
\begin{equation}
\label{cardinality_tuple_formal_def}
    \hat c
    =
    \hat C(\ppath)
    =
    \bigl(
        c_{k_-}(\ppath),\,
        c_{k_-+1}(\ppath),\,
        \dots,\,
        c_{k_+}(\ppath)
    \bigr).
\end{equation}
The set of admissible cardinality tuples is the image of this map:
\begin{equation}
\label{full_card_set}
    \mathcal C_{D,k^*}
    \Def
    \hat{C}_{D,k^*}\bigl(\PathsDk{k^*}\bigr)
    =
    \left\{
        \hat{C}_{D,k^*}(\ppath):
        \ppath\in\PathsDk{k^*}
    \right\}.
\end{equation}
In particular, admissibility means realizability by a path,
not merely satisfaction of a total-mass constraint. For
$\hat c=(c_k)_{k=k_-}^{k_+}\in\mathcal C_{D,k^*}$,
define its actual support endpoints by
\begin{equation}
\label{eq:actual-support}
a(\hat c)
\Def
\min\{k\in[k_-,k_+]:c_k>0\},
\quad
b(\hat c)
\Def
\max\{k\in[k_-,k_+]:c_k>0\}.
\end{equation}
Thus $[k_-,k_+]$ is the common feasible coordinate range,
whereas $[a(\hat c),b(\hat c)]$ is the occupied interval of
the individual tuple. Every $\hat c=(c_k)_{k=k_-}^{k_+}\in\mathcal C_{D,k^*}$ satisfies
\begin{equation}
\label{E:totalmass}
    \sum_{k=k_-}^{k_+} c_k = D+1,
\end{equation}
because a depth-$D$ path has $D+1$ vertex visits. Let $\vvalue:[k_-,k_+]\to\mathbb R$ assign a weight to each
position, and write $\vvalue_k=\vvalue(k)$. For every $\ppath\in\PathsDk{k^*}$,
\begin{equation}
\label{E:efficientaggregation}
    \sum_{d=0}^{D}\vvalue_{k_d}
    =
    \sum_{k=k_-}^{k_+}c_k(\ppath)\,\vvalue_k.
\end{equation}
Equal cardinality tuples therefore give equal aggregates for every choice of position-dependent weights. Distinct tuples, however, need not give distinct aggregate values for a particular choice of weights. The bounds and identities above are justified in \Cref{proof_er}. Before constructing the enumerator, we record three structural facts used throughout: the positive/mixed split, the
equivalence-class interpretation of cardinality tuples, and the ordering imposed on those tuples.

\begin{remark}[Set splitting]
\label{set_splitting}
We separate the admissible cardinality tuples according to
whether their realizing paths visit a negative position:
\[
\begin{aligned}
    \mathcal C^{+}_{D,k^*}
    &\Def
    \left\{
        \hat c\in\mathcal C_{D,k^*}:
        c_k=0
        \text{ for every }k\in[k_-,-1]
    \right\},\\
    \mathcal C^{-}_{D,k^*}
    &\Def
    \mathcal C_{D,k^*}\setminus\mathcal C^{+}_{D,k^*}.
\end{aligned}
\]
When $k_-=0$, the condition on the negative coordinates
is vacuous. We call the first set \emph{positive}, meaning that its paths visit only nonnegative positions. We call the second set \emph{mixed}, meaning that its paths visit at least one negative position; visits to strictly positive positions are not required. These terms refer to visited positions, not to the signs of the counts, which are always nonnegative. A path never visits a negative position exactly when all its negative-index visit counts vanish. Consequently,
\[
    \mathcal C^{+}_{D,k^*}
    =
    \hat{C}_{D,k^*}\bigl(\PosPathsDk{k^*}\bigr).
\]
Likewise, a tuple belongs to $\mathcal C^{-}_{D,k^*}$
exactly when $c_k>0$ for at least one $k\in[k_-,-1]$.
The two sets are disjoint and together exhaust
$\mathcal C_{D,k^*}$. Both sets retain the full coordinate range $[k_-,k_+]$. For the positive set, the identically zero negative coordinates may be omitted using the truncation map
in \Cref{truncation}. The positive-first construction below assumes $0\le k^*\le D$. If $k^*<0$, then $\mathcal C^{+}_{D,k^*}=\varnothing$; the full tuple set for that endpoint is obtained by reflecting positions in the construction for endpoint $-k^*$, as described in \Cref{thm:reflection-symmetry}.
\end{remark}

\begin{remark}[Cardinality tuples as equivalence-class keys]
For $\ppath,\ppath'\in\PathsDk{k^*}$, define
\begin{equation}
\label{equivalence_relation}
\ppath\sim\ppath'\;\Longleftrightarrow\;\hat{C}_{D,k^*}(\ppath)=\hat{C}_{D,k^*}(\ppath').
\end{equation}
Thus two paths are equivalent precisely when they visit every
position the same number of times, regardless of visit order.
The equivalence class of $\ppath$ is
\begin{equation}
\label{equivalence_class}
\begin{aligned}
[\ppath]_\sim
&\Def\left\{\ppath'\in\PathsDk{k^*}: \hat{C}_{D,k^*}(\ppath') = \hat{C}_{D,k^*}(\ppath)\right\}\\
&=\hat{C}_{D,k^*}^{-1}\left(\left\{\hat{C}_{D,k^*}(\ppath)\right\}\right).
\end{aligned}
\end{equation}
The inverse-image notation here denotes a set of paths,
not a unique inverse path. The quotient $\PathsDk{k^*}/\sim$ is the set of these equivalence classes. The correspondence
\[
    \PathsDk{k^*}/\sim
    \longrightarrow
    \mathcal C_{D,k^*},
    \quad
    [\ppath]_\sim
    \longmapsto
    \hat{C}_{D,k^*}(\ppath)
\]
is a bijection, as proved in \Cref{proof_er}.
Consequently, generating each admissible tuple exactly once
is equivalent to generating one key for each path class.
Establishing that the proposed enumeration achieves this
requires separate soundness, uniqueness, and completeness proofs.
\end{remark}

\begin{remark}[Cardinality-tuple ordering]
\label{card_ordering}
Fix $0\le k^*\le D$. For
$\hat c,\hat c'\in\mathcal C_{D,k^*}$,
define the comparison key
\begin{equation}
\label{eq:cardinality-order-key}
\Phi(\hat c)
\Def
\bigl(
a(\hat c),
c_{k_+},
c_{k_+-1},
\dots,
c_{k_-}
\bigr).
\end{equation}
For distinct tuples, define
\begin{equation}
\label{eq:cardinality-order}
\hat c\cardprec\hat c'
\quad\Longleftrightarrow\quad
\Phi(\hat c)
\text{ is lexicographically smaller than }
\Phi(\hat c'),
\end{equation}
and write $\hat c\cardsucc\hat c'$ when $\hat c'\cardprec\hat c$. Thus descending cardinality order first compares the lowest
occupied position: tuples with minimum $0$ precede those with
minimum $-1$, followed by minimum $-2$, and so on. Among
tuples having the same lowest occupied position, comparison
proceeds from the highest position downward, so the largest
position index at which their counts differ determines the
order. For $\hat c\in\mathcal C^+_{D,k^*}$, we have
$a(\hat c)=0$ and all negative-index coordinates vanish.
Hence the restriction of this order to
$\mathcal C^+_{D,k^*}$ is simply the right-to-left
lexicographic comparison of
\[
(c_{k_+},c_{k_+-1},\dots,c_0).
\]
Through the bijection established in \Cref{proof_er}, the
same order also orders the corresponding path-equivalence
classes independently of their representatives. The raw-path order \eqref{S:LexOrd} and cardinality-tuple order \eqref{eq:cardinality-order} are distinct: the former compares positions across depths, whereas the latter compares position-visit multiplicities.
\end{remark}

The enumeration problem is therefore to generate every admissible cardinality tuple in $\mathcal C_{D,k^*}$ exactly once. Each emitted tuple identifies one equivalence class and determines its position-dependent aggregates through
\eqref{E:efficientaggregation}; it does not specify a unique
raw path. The positive-tuple enumeration begins from the following seed, whose maximality is with respect to the ordering in \Cref{card_ordering}.

\begin{proposition}[Maximal positive seed]
\label{prop:maximal-positive-seed}
Fix $D\in\Z_+$ and $k^*\in\{0,\dots,D\}$. Define the seed tuple $\hat{s}^{(0)}=(c_k)_{k=k_-}^{k_+}$ by
\begin{equation}
\label{E:starting_tuple}
    c_k =
    \begin{cases}
        0,
        & k_-\le k<0,\\[2pt]
        1,
        & 0\le k\le k^*-1,\\[2pt]
        2,
        & k^*\le k\le k_+-1,\\[6pt]
        \begin{cases}
            1, & \text{$D-k^*$ even},\\
            2, & \text{$D-k^*$ odd},
        \end{cases}
        & k=k_+.
    \end{cases}
\end{equation}
Empty index ranges contribute no entries. Then $\hat{s}^{(0)}$ is the cardinality tuple of the path
in \Cref{E:highestpath}, belongs to $\mathcal C^+_{D,k^*}$, and is the unique largest tuple in that set under \Cref{card_ordering}.
\end{proposition}

The proof is given in \Cref{app:maximal-positive-seed}. We use $\hat{s}^{(0)}$ to initialize the mass-shift enumeration. Since this maximal tuple lies in
$\mathcal C^{+}_{D,k^*}$, we first develop the construction for the positive family with $0\le k^*\le D$; the extension to mixed tuples is introduced afterward.

\begin{remark}[Structural mass constraints]
\label{rem:mass_constraints}
Fix $D\in\Z_+$ and $k^*\in\{0,\dots,D\}$. Let $\hat c=(c_k)_{k=k_-}^{k_+}$ be a nonnegative-integer tuple with
\[
    c_k=0 \quad (k_-\le k<0),
    \quad
    \sum_{k=k_-}^{k_+}c_k=D+1.
\]
Then $\hat c\in\mathcal C^+_{D,k^*}$ if and only if
the following conditions hold:
\begin{enumerate}
    \item Every position from the root to the terminal
    position retains at least one visit:
    \[
        c_k\ge 1
        \quad (0\le k\le k^*).
    \]

    \item A visit above a position at or above $k^*$
    requires visits to that position both before and
    after the higher visit:
    \[
        c_{k+1}>0
        \;\Longrightarrow\;
        c_k\ge 2
        \qquad (k^*\le k<k_+).
    \]
\end{enumerate}
These are lower bounds, not fixed values: an entry may
exceed its required minimum. The second condition allows the current highest occupied position to have count $1$. If that position is strictly above $k^*$, its last visit may be moved to a lower occupied position. If it is $k^*$ itself, however, its last visit must be retained by the first condition. In particular, $c_{k^*}\ge 2$ whenever a higher position is occupied, and $c_{k^*}\ge 1$ in every admissible tuple. The figure below illustrates these requirements for $D=4$ and $k^*=2$. The highlighted path visits positions $0$ and $1$ once, position $2$ twice because it reaches position $3$, and the peak at position $3$ once.

\begin{figure}[H]
\centering
\begin{tikzpicture}[
  x=0.3cm,y=0.3cm, >=Stealth,
  every node/.style={circle, draw, minimum size=0.5cm, inner sep=0pt},
  edgepos/.style={}
]
  \node[fill=cyan!30] (n0) at (0,0) {0};
  
  \node[fill=cyan!30] (n1u) at (3,  2) {1};
  \node (n1m) at (3,  0) {0};
  \node (n1d) at (3, -2) {-1};
  
  \node[fill=red!30] (n2uu) at (6,  4) {2};
  \node (n2um) at (6,  2) {1};
  \node (n2mm) at (6,  0) {0};
  \node (n2dm) at (6, -2) {-1};
  \node (n2dd) at (6, -4) {-2};
  
  \node[fill=green!30] (n3uuu) at (9,  6) {3};
  \node (n3uum) at (9,  4) {2};
  \node (n3umm) at (9,  2) {1};
  \node (n3mmm) at (9,  0) {0};
  \node (n3dmm) at (9, -2) {-1};
  \node (n3ddm) at (9, -4) {-2};
  \node (n3ddd) at (9, -6) {-3};
  
  \node (n4uuuu) at (12,  8) {4};
  \node (n4uuum) at (12,  6) {3};
  \node[fill=red!30] (n4uumm) at (12,  4) {2};
  \node (n4ummm) at (12,  2) {1};
  \node (n4mmmm) at (12,  0) {0};
  \node (n4dmmm) at (12, -2) {-1};
  \node (n4ddmm) at (12, -4) {-2};
  \node (n4dddm) at (12, -6) {-3};
  \node (n4dddd) at (12, -8) {-4};

  \draw[red, dotted, very thick] (-0.5,4) -- (15,4);
  \draw[green, thick] (n0) -- (n1u);
  \draw               (n0) -- (n1m);
  \draw               (n0) -- (n1d);
  
  \draw[green, thick] (n1u) -- (n2uu);
  \draw               (n1u) -- (n2um);
  \draw               (n1u) -- (n2mm);
  
  \draw               (n1m) -- (n2um);
  \draw               (n1m) -- (n2mm);
  \draw               (n1m) -- (n2dm);
  
  \draw               (n1d) -- (n2mm);
  \draw               (n1d) -- (n2dm);
  \draw               (n1d) -- (n2dd);
  
  \draw[green, thick] (n2uu) -- (n3uuu);
  \draw               (n2uu) -- (n3uum);
  \draw               (n2uu) -- (n3umm);
  
  \draw               (n2um) -- (n3uum);
  \draw               (n2um) -- (n3umm);
  \draw               (n2um) -- (n3mmm);
  
  \draw               (n2mm) -- (n3umm);
  \draw               (n2mm) -- (n3mmm);
  \draw               (n2mm) -- (n3dmm);
  
  \draw               (n2dm) -- (n3mmm);
  \draw               (n2dm) -- (n3dmm);
  \draw               (n2dm) -- (n3ddm);
  
  \draw               (n2dd) -- (n3dmm);
  \draw               (n2dd) -- (n3ddm);
  \draw               (n2dd) -- (n3ddd);
  
  \draw               (n3uuu) -- (n4uuuu);
  \draw               (n3uuu) -- (n4uuum);
  \draw[green, thick] (n3uuu) -- (n4uumm);
  
  \draw               (n3uum) -- (n4uuum);
  \draw               (n3uum) -- (n4uumm);
  \draw               (n3uum) -- (n4ummm);
  
  \draw               (n3umm) -- (n4uumm);
  \draw               (n3umm) -- (n4ummm);
  \draw               (n3umm) -- (n4mmmm);
  
  \draw               (n3mmm) -- (n4ummm);
  \draw               (n3mmm) -- (n4mmmm);
  \draw               (n3mmm) -- (n4dmmm);
  
  \draw               (n3dmm) -- (n4mmmm);
  \draw               (n3dmm) -- (n4dmmm);
  \draw               (n3dmm) -- (n4ddmm);
  
  \draw               (n3ddm) -- (n4dmmm);
  \draw               (n3ddm) -- (n4ddmm);
  \draw               (n3ddm) -- (n4dddm);
  
  \draw               (n3ddd) -- (n4ddmm);
  \draw               (n3ddd) -- (n4dddm);
  \draw               (n3ddd) -- (n4dddd);
\end{tikzpicture}
\caption{Example of Graphically Informed Constraints on $c_{k}$}
\label{fig:mass_constraints_graph}
\end{figure}

Necessity and sufficiency are proved in
\Cref{lem:positive-admissibility}. Thus a completed redistribution preserves admissibility exactly when its resulting tuple retains nonnegative entries, zero negative-index entries, total mass $D+1$, and the two conditions above.
This characterizes admissible outputs; it does not by itself establish that the enumeration produces every such output exactly once.
\end{remark}

\paragraph{\textbf{Working Example}}

Let's work through a simple example, taking $D=7$ and $k^*=2$. Using \Cref{E:Kbounds}, we have
\begin{equation*} k_-=-2 \qquad \text{and}\qquad k_+=4. \end{equation*}

From \Cref{E:highestpath}, the highest path with this $(k^{*},D) = (2,7)$ has position sequence
\begin{equation}\label{E:samplemaxpath}
(\blue{0},\blue{1},\red{2},\red{3},\green{4},\green{4},\red{3},\red{2})
\end{equation}
and this has cardinality tuple $\hat{c}$: 
\begin{equation}\label{E:toptuple}
\begin{blockarray}{ccccccc}
 -2 & -1 & 0 & 1 & 2 & 3 & 4   \\
    \begin{block}{(ccccccc)}
    0, & 0, & \blue{1}, & \blue{1}, & \red{2}, & \red{2}, & \green{2} \\
    \end{block}
    \end{blockarray}
\end{equation}
In line with \Cref{E:totalmass}, we have
\begin{equation*} 0 + 0 + 1 + 1 + 2 + 2 + 2 = 8. \end{equation*}
The colors in \Cref{E:samplemaxpath,E:toptuple} make the
visit structure explicit. The path first ascends from $0$ to
$k^*=2$ (blue), continues above the endpoint to the peak at
$4$ (green), and then returns to $k^*=2$ (red). Consequently,
positions $0$ and $1$ are visited once, positions $2$ and $3$
twice, and position $4$ twice. The second visit to $4$ occurs
because $D-k^*=5$ is odd; when $D-k^*$ is even, the maximal
path reaches its peak only once. This is precisely the
maximal-seed structure described in \Cref{E:starting_tuple}. The next cardinality tuple in descending order in
$\mathcal C^+_{7,2}$ is
\begin{equation}
\label{E:nexthighestcardinality}
\begin{blockarray}{ccccccc}
 -2 & -1 & 0 & 1 & 2 & 3 & 4 \\
\begin{block}{(ccccccc)}
 0, & 0, & 1, & 1, & 2, & 3, & 1 \\
\end{block}
\end{blockarray}.
\end{equation}
It is realized, for example, by the two raw paths
\[
(0,1,2,3,3,4,3,2),
\quad
(0,1,2,3,4,3,3,2).
\]
To see directly why this tuple follows the maximal tuple,
recall that the coordinates themselves remain indexed in the
fixed order
\[
-2,-1,0,1,2,3,4,
\]
while \Cref{card_ordering} compares their counts from the
highest position downward. For this comparison only, we
display the same coordinates in the order
$4,3,2,1,0,-1,-2$:
\begin{equation}
\label{eq:running-first-card-comparison}
\begin{blockarray}{ccccccc}
 4 & 3 & 2 & 1 & 0 & -1 & -2 \\
\begin{block}{(ccccccc)}
 \green{2}, & \red{2}, & \red{2}, &
 \blue{1}, & \blue{1}, & 0, & 0 \\
\end{block}
\end{blockarray}
\;\cardsucc\;
\begin{blockarray}{ccccccc}
 4 & 3 & 2 & 1 & 0 & -1 & -2 \\
\begin{block}{(ccccccc)}
 1, & 3, & 2, & 1, & 1, & 0, & 0 \\
\end{block}
\end{blockarray}.
\end{equation}
The first coordinate at which the two tuples differ is the
highest position $k=4$, where $2>1$. Hence the maximal tuple
precedes \Cref{E:nexthighestcardinality}; combinatorially, one
visit has moved from position $4$ to position $3$. The third
cardinality tuple in descending order is
\begin{equation}
\label{E:secondcardinalitytuple}
\begin{blockarray}{ccccccc}
 -2 & -1 & 0 & 1 & 2 & 3 & 4 \\
\begin{block}{(ccccccc)}
 0, & 0, & 1, & 1, & 3, & 2, & 1 \\
\end{block}
\end{blockarray}.
\end{equation}
It is realized, for example, by
\[
(0,1,2,2,3,4,3,2),
\qquad
(0,1,2,3,4,3,2,2).
\]
Compared with \Cref{E:nexthighestcardinality}, the count at
$k=4$ is unchanged, while the first difference encountered
from the right is at $k=3$, where $3>2$. Thus
\[
(0,0,1,1,2,2,2)
\;\cardsucc\;
(0,0,1,1,2,3,1)
\;\cardsucc\;
(0,0,1,1,3,2,1).
\]
The second transition therefore moves one visit from
position $3$ to position $2$. The cardinality-tuple order should not be confused with the raw-path order of \eqref{S:LexOrd}. The representatives of these three classes interleave under the latter:
\[
\begin{aligned}
&(0,1,2,3,4,4,3,2)
    \pathsucc
 (0,1,2,3,4,3,3,2)
    \pathsucc
 (0,1,2,3,4,3,2,2)\\
&\qquad
    \pathsucc
 (0,1,2,3,3,4,3,2)
    \pathsucc
 (0,1,2,2,3,4,3,2).
\end{aligned}
\]
Thus raw paths are ordered by their positions across depths,
whereas the equivalence classes are ordered directly by
their cardinality tuples. The latter order is therefore
independent of which raw-path representative is chosen from
a class.

\begin{remark}[Truncation of positive tuples]
\label{truncation}
For every $\hat c\in\mathcal C^+_{D,k^*}$, all coordinates
indexed by $k<0$ vanish. We therefore suppress these fixed
zeros when describing the positive mass redistribution and define
the truncation map
\begin{equation}
\label{eq:positive-truncation}
\mathrm{Tr}:
\mathcal C^+_{D,k^*}
\longrightarrow
\mathrm{Tr}(\mathcal C^+_{D,k^*})
\subseteq
\Z_+^{[0,k_+]},
\qquad
\mathrm{Tr}\bigl((c_k)_{k=k_-}^{k_+}\bigr)
\Def
(c_k)_{k=0}^{k_+}.
\end{equation}
Its inverse on the image restores the omitted zero
coordinates:
\begin{equation}
\label{eq:positive-truncation-inverse}
\mathrm{Tr}^{-1}\bigl((c_k)_{k=0}^{k_+}\bigr)
\Def
(\underbrace{0,\dots,0}_{k=k_-,\dots,-1},
c_0,\dots,c_{k_+}).
\end{equation}
Thus $\mathrm{Tr}$ is a bijection onto its image and preserves
the cardinality order on $\mathcal C^+_{D,k^*}$, since the
discarded coordinates are identically zero.
\end{remark}
Returning to the running example with \(D=7\), \(k^*=2\), \(k_-=-2\), \(k_+=4\), we see that the length-\(7\) tuples in \(\mathcal{C}^{+}_{7,2}\) (indexed by \(-2,-1,0,1,2,3,4\)) correspond bijectively (c.f. \Cref{truncation}) to length-\(5\) truncated tuples (indexed by \(0,1,2,3,4\)):
\begin{gather*}
(0,0,1,1,2,2,2)\ \longleftrightarrow\ (1,1,2,2,2),\\
(0,0,1,1,2,3,1)\ \longleftrightarrow\ (1,1,2,3,1),\\
(0,0,1,1,3,2,1)\ \longleftrightarrow\ (1,1,3,2,1),\\
(0,0,1,2,2,2,1)\ \longleftrightarrow\ (1,2,2,2,1),\\
(0,0,2,1,2,2,1)\ \longleftrightarrow\ (2,1,2,2,1).
\end{gather*}
The first five tuples therefore correspond to moving one
unit of residual mass through the truncated coordinates.
Equivalently, they are obtained from the common baseline
$(1,1,2,2,1)$ as follows:
\begin{equation}
\label{mass_process}
\begin{aligned}
&(1,1,2,2,1) + (0,0,0,0,1) = (1,1,2,2,2)\\
\Longrightarrow\ &(1,1,2,2,1) + (0,0,0,1,0) = (1,1,2,3,1)\\
\Longrightarrow\ &(1,1,2,2,1) + (0,0,1,0,0) = (1,1,3,2,1)\\
\Longrightarrow\ &(1,1,2,2,1) + (0,1,0,0,0) = (1,2,2,2,1)\\
\Longrightarrow\ &(1,1,2,2,1) + (1,0,0,0,0) = (2,1,2,2,1)
\end{aligned}
\end{equation}
Applying $\mathrm{Tr}^{-1}$ restores the two omitted coordinates $k=-2,-1$. Thus the truncated enumeration is bijective with the corresponding enumeration in $\mathcal C^+_{7,2}$ and preserves its cardinality order. The process in \Cref{mass_process} can also be captured by graphically traversing the tree. To do this, we consider a process where a unit of mass $i=1$ is subtracted from the final entry of $\hat{c}$, and this mass is ``moved'' leftward through the tuple, one position at a time. This results in the following sequence of configurations which we can interpret as a walk down the unique path representatives in the graph:
\begin{align}
\begin{array}{c}
(0, 0, \ldots, 0, 1) \\
(0, 0, \ldots, 1, 0) \\
\vdots \\
(0, 1, \ldots, 0, 0) \\
(1, 0, \ldots, 0, 0)
\end{array}
\quad \Longrightarrow \quad
\begin{array}{c}
    \text{Sample Path Configuration}\\
    \includegraphics[width=0.5\textwidth]{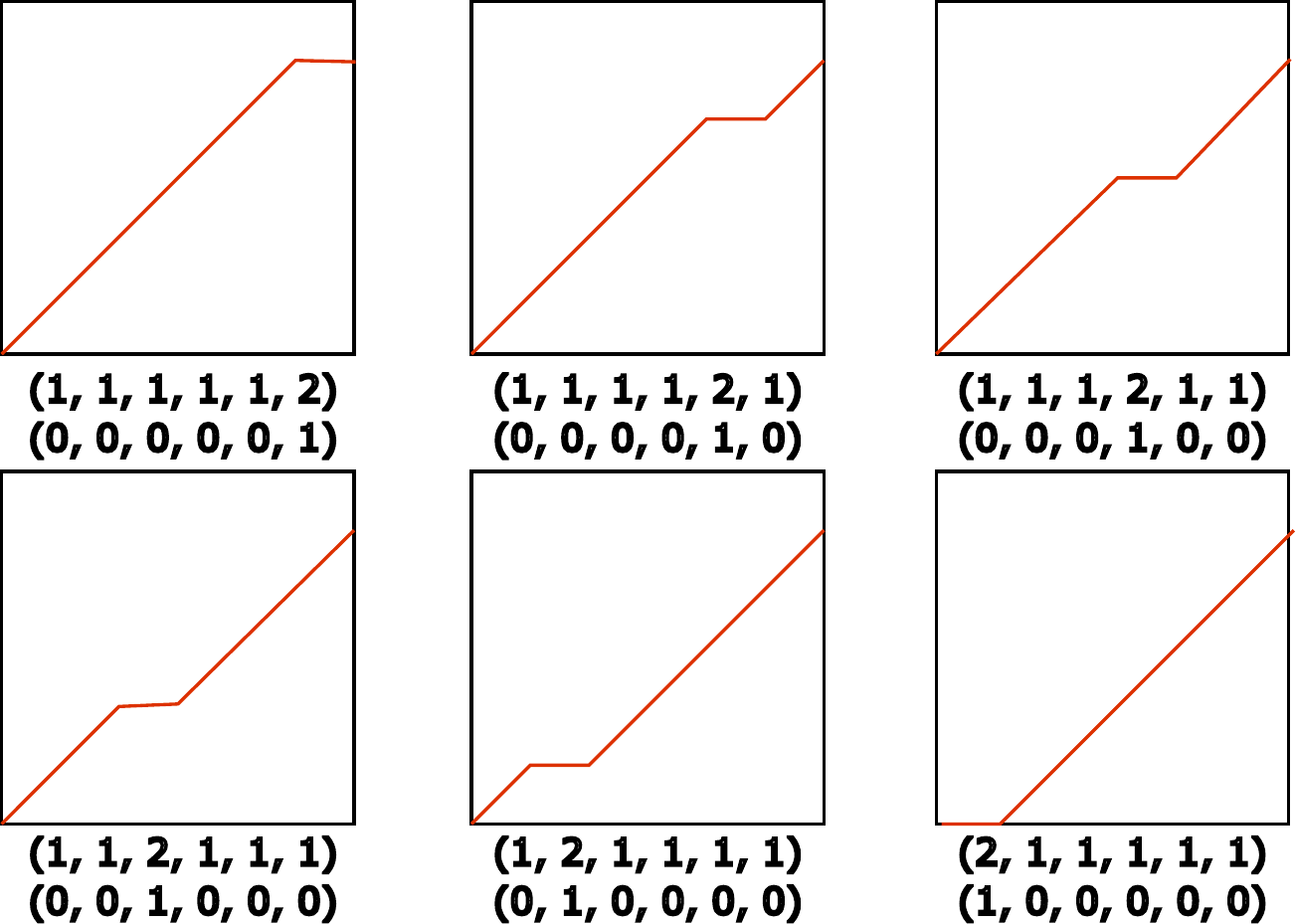}
    \label{first_partition}
\end{array}
\end{align}
As can be seen in \Cref{first_partition}, the “one–mass sweep’’ already exhausts all unique configurations when \(k^{*}=D-1\): by the equivalence relation \Cref{equivalence_relation} and the graphical observations in \Cref{first_partition}, shifting zero mass (the starting tuple \Cref{E:starting_tuple}) and then one unit from the rightmost entry across the support generates every admissible class at that depth. For deeper targets \(k^{*}<D-1\) (e.g., the toy case \(D=7,\ k^{*}=2\)), additional configurations appear and we must move larger amounts of mass while preserving feasibility of paths (hence of the cardinality tuples) to generate \textbf{all} the possible configurations. To organize the redistribution, we distinguish the \emph{mass-layer index} $i$ from the number $\ell$ of coordinates available to receive that mass. As $i$
increases, the admissible right boundary may retreat, so the
slot count is layer-dependent; below we write it as
$\ell(i)$. For example, in the running case $D=7$ and $k^*=2$,
the initial truncated tuple has five positions
$0,1,2,3,4$. The seed requires no redistribution.
The $i=1$ layer distributes one unit over four available
slots; the $i=2$ layer distributes two units over the same
four slots; the next two layers use three slots; and the
final nonempty layer uses two slots. We formalize this organization as follows.

\begin{proposition}[Positive mass-layer decomposition]
\label{prop:positive-mass-layers}
Fix $D\in\Z_+$ and $k^*\in\{0,\dots,D\}$. For
$\hat c=(c_k)_{k=k_-}^{k_+}\in\mathcal C^+_{D,k^*}$,
write $b\Def b(\hat c)$ for its highest occupied position from
\eqref{eq:actual-support}. Assign $\hat c$ the mass-layer index
\begin{equation}
\label{eq:positive-layer-index}
    i
    =
    \begin{cases}
        D+k^*-2b,
        & c_b=1,\\[2pt]
        D+k^*-2b-1,
        & c_b\ge2.
    \end{cases}
\end{equation}
Then $0\le i\le D-k^*,$ and every tuple in $\mathcal C^+_{D,k^*}$ belongs to exactly one such layer. Layer $i=0$ consists only of the maximal seed $\hat{s}^{(0)}$ in \eqref{E:starting_tuple}. For each integer $i$ with $1\le i\le D-k^*$, define the number of available mass-receiving slots by
\begin{equation}
\label{eq:available-slots}
    \ell(i)
    \Def
    \left\lceil
        \frac{D+k^*-i}{2}
    \right\rceil.
\end{equation}
When a single layer is under discussion, we write simply
$\ell=\ell(i)$. If $\ell(i)\ge1$, the tuples in layer $i$ are in bijection with the weak compositions of mass $i$ into those $\ell$ available slots. Consequently, the number of tuples in that layer is
\begin{equation}
\label{eq:positive-layer-count}
    \binom{i+\ell-1}{\ell-1},
    \quad
    \ell=\ell(i).
\end{equation}
If $i>0$ and $\ell(i)=0$, the layer is empty. Listing the layers in increasing order of $i$, and listing the weak compositions within each layer from the highest-position slot toward the lowest-position slot, recovers the descending cardinality-tuple order of \Cref{card_ordering}.
\end{proposition}

The proof is given in \Cref{app:positive-mass-layers}. In the running example,
\[
\begin{array}{c|cccccc}
i       & 0 & 1 & 2 & 3 & 4 & 5\\
\hline
\ell    & 5 & 4 & 4 & 3 & 3 & 2
\end{array}
\]
where the value $\ell=5$ in the seed row records the five-position truncated support; no mass is distributed in that row. Thus \Cref{mass_process} consists of the single seed tuple followed by the four tuples in layer $i=1$. For layer $i=2$, two units of mass are distributed over the four currently available positions $0,1,2,3$. Starting from the fixed seed, removing the two units that
initiate this layer gives the baseline $(1,1,2,2,0).$ Returning either of those units to position $4$ would recreate a tuple belonging to an earlier layer, so
position $4$ is no longer an available receiving slot. The transition from the final tuple of layer $i=1$ to the first tuple of layer $i=2$ is:
\begin{equation}
\label{starting_tuple_1}
\begin{tikzpicture}[baseline=(current bounding box.center)]
  \node (lp1) at (0,0) {$\Bigl($};

  \node (a1)  at (0.60,0) {$\boxed{2},$};
  \node (a2)  at (1.35,0) {$1,$};
  \node (a3)  at (2.10,0) {$2,$};
  \node (a4)  at (2.85,0) {$\boxed{2},$};
  \node (a5)  at (3.60,0) {$\boxed{1}$};

  \node (rp1) at (4.05,0) {$\Bigr)$};
  \node (b1)  at (5.40,0) {$\Longrightarrow (1,1,2,4,0)$};

  \draw[-{Latex[length=2mm]}]
    (a5.south) .. controls ($(a5.south)+(0,-0.9)$) and ($(a4.south)+(0,-0.9)$) ..
    node[pos=0.20, below=6pt, font=\footnotesize] {$-1$}
    (a4.south);

  \draw[-{Latex[length=2mm]}]
    (a1.south) .. controls ($(a1.south)+(0,-1.2)$) and ($(a4.south)+(0,-1.2)$) ..
    node[pos=0.18, below=8pt, font=\footnotesize] {$-1$}
    (a4.south);
\end{tikzpicture}
\end{equation}
\begin{remark}
    The same mechanism also appears in the transition from the seed
    layer $i=0$ to the first positive mass layer $i=1$, naturally
    excluding the seed tuple $\hat{s}^{(0)}$ from the $i=1$ layer:
    \begin{equation}
    \label{starting_tuple_2}
    \begin{tikzpicture}[baseline=(current bounding box.center)]
      \node (lp1) at (0,0) {$\Bigl($};
    
      \node (a1)  at (0.60,0) {$\boxed{1},$};
      \node (a2)  at (1.35,0) {$1,$};
      \node (a3)  at (2.10,0) {$2,$};
      \node (a4)  at (2.85,0) {$\boxed{2},$};
      \node (a5)  at (3.60,0) {$\boxed{2}$};
    
      \node (rp1) at (4.05,0) {$\Bigr)$};
      \node (b1)  at (5.40,0) {$\Longrightarrow (1,1,2,3,1)$};
    
      \draw[-{Latex[length=2mm]}]
        (a5.south) .. controls ($(a5.south)+(0,-0.9)$) and ($(a4.south)+(0,-0.9)$) ..
        node[pos=0.20, below=6pt, font=\footnotesize] {$-1$}
        (a4.south);
    
      \draw[-{Latex[length=2mm]}]
        (a1.south) .. controls ($(a1.south)+(0,-1.2)$) and ($(a4.south)+(0,-1.2)$) ..
        node[pos=0.18, below=8pt, font=\footnotesize] {$-0$}
        (a4.south);
    \end{tikzpicture}
    \end{equation}
In this transition, one visit moves from level $4$ to
level $3$, while the count at level $0$ remains unchanged.
The result is the first tuple in layer $i=1$, not the seed. Thus the seed contributes the unique $i=0$ configuration, while layer $i=1$ redistributes one unit of mass across $\ell=4$ available slots. These give the five tuples in \Cref{mass_process}. The truncated tuples still have five coordinates; $\ell$ records only the coordinates currently available to receive the redistributed mass.
\end{remark}

By \Cref{prop:positive-mass-layers}, distinct mass layers
contain disjoint sets of cardinality tuples, and distinct
weak compositions within a fixed layer produce distinct
tuples. For the running example, layer $i=2$ has $\ell=4$ available slots and therefore contains $\binom{2+4-1}{4-1}=
\binom{5}{3}=10$ tuples. These are:
\begin{equation}
\label{second_partition}
\begin{aligned}
    &(1,1,2,2,0) + (0,0,0,2,0) = (1,1,2,4,0)\\
    \implies &(1,1,2,2,0) + (0,0,1,1,0) = (1,1,3,3,0)\\
    \implies &(1,1,2,2,0) + (0,1,0,1,0) = (1,2,2,3,0)\\
    \implies &(1,1,2,2,0) + (1,0,0,1,0) = (2,1,2,3,0)\\
    \implies &(1,1,2,2,0) + (0,0,2,0,0) = (1,1,4,2,0)\\
    \implies &(1,1,2,2,0) + (0,1,1,0,0) = (1,2,3,2,0)\\
    \implies &(1,1,2,2,0) + (1,0,1,0,0) = (2,1,3,2,0)\\
    \implies &(1,1,2,2,0) + (0,2,0,0,0) = (1,3,2,2,0)\\
    \implies &(1,1,2,2,0) + (1,1,0,0,0) = (2,2,2,2,0)\\
    \implies &(1,1,2,2,0) + (2,0,0,0,0) = (3,1,2,2,0)\\
\end{aligned}
\end{equation}
The residual distributions in
\Cref{mass_process,second_partition} are \emph{weak
compositions} rather than ordinary integer partitions:
their coordinate positions matter. For example, assigning
two units to the first available slot and assigning two
units to the second produce distinct cardinality tuples,
even though both distributions have the same underlying
integer partition $(2)$. At mass layer $i$, with $\ell\ge1$ receiving slots, these weak compositions are counted by
\[
    \binom{i+\ell-1}{\ell-1},
\]
as already established in \eqref{eq:positive-layer-count} and
\Cref{prop:positive-mass-layers} \cite{Page2013}. The same
proposition gives the bijection between these residual
distributions and the admissible positive cardinality tuples in the layer. Thus the weak-composition count is not an auxiliary approximation: it counts exactly the tuples generated by the mass-layer construction. For finite $D$ and $k^*$, the same correspondence can also
be inspected directly on the tree:
\input{treefigure}

The running example therefore supplies both the graphical
and combinatorial picture of the positive-stage
enumeration. Before passing to the general construction, we
examine the simplest endpoint cases explicitly. These cases
make clear how the mass layers appear at the boundary of the
tree and motivate the parity structure used below.

\subsubsection{Case Examination:}

\section*{\textbf{Case 1}: $k^{*} = D,\, D \geq 0$}

In \Cref{E:starting_tuple} we introduced the maximal
positive seed. When $k^*=D$, a path ending at $D$ must
take an upward step at every one of its $D$ steps.
Therefore there is exactly one raw path and exactly one
cardinality tuple:
\begin{equation}
\label{case_1}
|\mathcal C_{D,D}|=1.
\end{equation}
In the mass-layer description this is simply the seed
layer $i=0$; no positive mass redistribution occurs.
The statement also includes $D=k^*=0$, where the sole
path consists of the root and has $c_0=1$.

\section*{\textbf{Case 2}: $k^{*} = D - 1,\, D - 1 \geq 0$}
When $k^*=D-1$, the required net displacement is $D-1$. Hence every admissible path consists of $D-1$ upward steps and exactly one stay. A downward step would reduce the displacement by two and therefore cannot occur. The stay can occur at any one of the $D$ step indices. Each choice produces a distinct cardinality tuple, and no path can visit a negative position. Consequently,
\begin{equation}
\label{case_2}
|\mathcal C_{D,D-1}|
=
|\mathcal C^+_{D,D-1}|
=
D.
\end{equation}
In the mass-layer decomposition, the maximal seed contributes one tuple at $i=0$. For $D\ge2$, layer $i=1$ has $\ell(1)=D-1$ available receiving slots and therefore contributes
\[
    \binom{1+(D-1)-1}{(D-1)-1}
    =
    D-1
\]
additional tuples. Thus the total is $1+(D-1)=D.$ For $D=1$, the seed is the only tuple and the nominal positive-mass layer has no receiving slots.

\section*{Recognizing the Algorithm}
Cases 1 and 2 illustrate the first two boundary behaviors of
the mass-layer construction. For the general positive family,
the layer index $i$ and its receiving-slot count $\ell(i)$
have already been defined in \Cref{prop:positive-mass-layers}. What remains is to make explicit how the parity of the maximal seed determines the successive slot schedule. By \eqref{E:starting_tuple},
\begin{equation}
\label{switching_term}
    \beta
    \Def
    1+D+k^*-2k_+
    =
    \begin{cases}
        1, & D+k^*\text{ is even},\\[2pt]
        2, & D+k^*\text{ is odd}.
    \end{cases}
\end{equation}
Because $D+k^*$ and $D-k^*$ have the same parity,
this is exactly the parity rule already appearing in
\eqref{E:starting_tuple}. The parameter $\beta$ therefore records the \emph{initial seed peak count}. It is not the current rightmost count of every tuple produced later in the enumeration. For every positive-mass layer, the receiving-slot count is $\ell(i)$ defined in \eqref{eq:available-slots}. Using $D+k^*=2k_++\beta-1$ in that expression yields the two explicit slot schedules determined by the parity of the seed.
\begin{table}[H]
\small
\centering
\renewcommand{\arraystretch}{1.15}
\begin{tabular}{c|cc|cc}
\hline
&
\multicolumn{2}{c|}{$\beta=1$:
even $D+k^*$}
&
\multicolumn{2}{c}{$\beta=2$:
odd $D+k^*$}
\\
$i$
&
$\ell(i)$
&
\# tuples
&
$\ell(i)$
&
\# tuples
\\
\hline
$0$
&
seed
&
$1$
&
seed
&
$1$
\\
$1$
&
$k_+$
&
$\displaystyle\binom{k_+}{k_+-1}$
&
$k_+$
&
$\displaystyle\binom{k_+}{k_+-1}$
\\
$2$
&
$k_+-1$
&
$\displaystyle\binom{k_+}{k_+-2}$
&
$k_+$
&
$\displaystyle\binom{k_++1}{k_+-1}$
\\
$3$
&
$k_+-1$
&
$\displaystyle\binom{k_++1}{k_+-2}$
&
$k_+-1$
&
$\displaystyle\binom{k_++1}{k_+-2}$
\\
$4$
&
$k_+-2$
&
$\displaystyle\binom{k_++1}{k_+-3}$
&
$k_+-1$
&
$\displaystyle\binom{k_++2}{k_+-2}$
\\
$5$
&
$k_+-2$
&
$\displaystyle\binom{k_++2}{k_+-3}$
&
$k_+-2$
&
$\displaystyle\binom{k_++2}{k_+-3}$
\\
\hline
\end{tabular}
\caption{Initial positive mass layers in the two seed-parity
regimes. Layer $i=0$ is the maximal seed. For $i\ge1$,
each entry in the ``\# tuples'' columns is the
weak-composition count
$\binom{i+\ell(i)-1}{\ell(i)-1}$.
Only layers with $i\le D-k^*$ and $\ell(i)\ge1$ contribute.}
\label{tab:weak_compositions}
\end{table}
The table makes explicit the effect of the seed parity.
When $\beta=1$, the available-slot sequence is
\[
    k_+,\;
    k_+-1,\;
    k_+-1,\;
    k_+-2,\;
    k_+-2,\ldots,
\]
for $i=1,2,3,4,5,\ldots$. When $\beta=2$, it is
\[
    k_+,\;
    k_+,\;
    k_+-1,\;
    k_+-1,\;
    k_+-2,\ldots.
\]
Thus the two regimes agree on alternating layers and differ
on the intervening ones, which motivates the unified parity
correction introduced below. For the running example $D=7$ and $k^*=2$,
\[
    k_+=4,
    \quad
    D+k^*=9,
    \quad
    \beta=2.
\]
Hence the positive layers have
\[
\begin{array}{c|cccccc}
i & 0 & 1 & 2 & 3 & 4 & 5\\
\hline
\ell(i) & \text{seed} & 4 & 4 & 3 & 3 & 2\\
\# & 1 & 4 & 10 & 10 & 15 & 6
\end{array}
\]
and therefore $|\mathcal C^+_{7,2}|=1+4+10+10+15+6=46.$ In particular, \Cref{mass_process} displays the seed and the complete layer $i=1$, while \Cref{second_partition} displays the complete layer $i=2$. The remaining layers follow the same weak-composition construction with the decreasing slot counts shown above. We now combine the two parity schedules into a single expression. Define, for $i\in\Z_+$,
\begin{equation}
\label{delta}
    \delta(i)
    \Def
    \begin{cases}
        1, & i\text{ is even},\\
        0, & i\text{ is odd}.
    \end{cases}
\end{equation}
Then
\[
    \left\lceil\frac{i+1}{2}\right\rceil
    =
    \left\lfloor\frac{i+1}{2}\right\rfloor
    +
    \delta(i).
\]
For every positive-mass layer,
\begin{equation}
\label{eq:uniform-delta}
    \ell(i)
    =
    k_++1
    -
    \left\lceil\frac{i+1}{2}\right\rceil
    +
    (\beta-1)\delta(i),
    \qquad
    1\le i\le D-k^*.
\end{equation}
Thus the correction from the $\beta=1$ schedule to the
$\beta=2$ schedule occurs only when $i$ is even. The local parity factor $\delta(i)$ is therefore essential; it cannot be replaced by the global factor $\beta-1$. By \Cref{prop:positive-mass-layers}, a nonempty positive-mass layer contains $\binom{i+\ell(i)-1}{\ell(i)-1}$ tuples. Hence the entire positive count is
\begin{equation}
\label{eq:positive-count-layers}
    |\mathcal C^+_{D,k^*}|
    =
    1
    +
    \sum_{i=1}^{D-k^*}
    \binom{i+\ell(i)-1}{\ell(i)-1},
\end{equation}
The leading $1$ counts the seed, and $\ell(i)$ retains its definition in \eqref{eq:available-slots}. An empty sum is zero, and a binomial coefficient with nonnegative upper entry is taken to be zero when its lower entry is negative or exceeds its upper entry; consequently, a positive-mass layer with $\ell(i)=0$ contributes nothing. To express the same layer count directly in terms of $k_+$, $\beta$, and the parity correction $\delta(i)$, substitute \eqref{eq:uniform-delta} into the
weak-composition term. When $(\beta-1)\delta(i)=1$, Pascal's identity gives
\begin{equation}
\label{eq:single-summand}
\begin{aligned}
    \binom{i+\ell(i)-1}{\ell(i)-1}
    &=
    \binom{
        i+k_+-\left\lceil\frac{i+1}{2}\right\rceil
    }{
        k_+-\left\lceil\frac{i+1}{2}\right\rceil
    }\\
    &\quad+
    (\beta-1)\delta(i)
    \binom{
        i+k_+-\left\lceil\frac{i+1}{2}\right\rceil
    }{
        k_++1-\left\lceil\frac{i+1}{2}\right\rceil
    }.
\end{aligned}
\end{equation}
Consequently,
\begin{equation}
\label{eq:final-closed-form}
\begin{aligned}
    |\mathcal C^+_{D,k^*}|
    &=
    1+
    \sum_{i=1}^{D-k^*}
    \Bigg[
    \binom{
        i+k_+-\left\lceil\frac{i+1}{2}\right\rceil
    }{
        k_+-\left\lceil\frac{i+1}{2}\right\rceil
    }\\
    &\qquad\qquad+
    (\beta-1)\delta(i)
    \binom{
        i+k_+-\left\lceil\frac{i+1}{2}\right\rceil
    }{
        k_++1-\left\lceil\frac{i+1}{2}\right\rceil
    }
    \Bigg].
\end{aligned}
\end{equation}
The proof that the mass layers are disjoint and exhaustive
is given in \Cref{app:positive-mass-layers}; the algebraic
equivalence of \eqref{eq:positive-count-layers} and
\eqref{eq:final-closed-form} is proved in
\Cref{app:positive-parity-count}.
\begin{remark}[Positive-stage stopping range]
The upper layer index $D-k^*$ has a direct occupancy
interpretation. Every positive tuple must retain at least
one visit at each position $0,\dots,k^*$ by
\Cref{rem:mass_constraints}. These are lower bounds rather
than fixed entries: any surplus may be redistributed so
long as the resulting tuple remains admissible. Reserving
these $k^*+1$ mandatory visits leaves
\[
\begin{aligned}
    \sum_{k=0}^{k^*}(c_k-1)
    +\sum_{k=k^*+1}^{k_+}c_k
    &=(D+1)-(k^*+1)\\
    &=D-k^*.
\end{aligned}
\]
Thus $D-k^*$ is the total excess above the one-visit
baseline, which explains the layer range
$0\le i\le D-k^*$ in
\Cref{prop:positive-mass-layers}. This does not mean that
all excess visits may be moved independently: every emitted
tuple must continue to satisfy the admissibility conditions
of \Cref{rem:mass_constraints}. Layer $i=0$ is the maximal
seed, and the positive-mass layers begin at $i=1$.
\end{remark}

We now extend the construction beyond nonnegative paths.
Reflection first reduces the treatment of negative endpoints
to their nonnegative counterparts; this endpoint symmetry is
distinct from the positive/mixed split at a fixed endpoint.

\paragraph{\textbf{Symmetry of the Recombining Tree}}

Position reflection preserves the root, depth, and permitted
increments, leading to the following bijection.

\begin{theorem}[Reflection symmetry of counts]
\label{thm:reflection-symmetry}
Fix $D\in\Z_+$ and $k^*\in\Z$ with $|k^*|\le D$.
Reflecting every position gives a bijection
\[
    \PathsDk{k^*}\longrightarrow\PathsDk{-k^*},
    \qquad
    \ppath\longmapsto\ppath',
\]
where, for $\pi(\ppath)=(k_0,\dots,k_D)$, $\ppath'\Def\bigl((-k_d,d)\bigr)_{d=0}^{D}.$ The induced bijection on cardinality tuples is
\begin{equation}
\label{eq:tuple-reflection}
\begin{aligned}
    \mathcal C_{D,k^*}&\longrightarrow\mathcal C_{D,-k^*},\\
    (c_k)_{k=k_-}^{k_+}
    &\longmapsto(c_{-k})_{k=-k_+}^{-k_-}.
\end{aligned}
\end{equation}
Here $k_-,k_+$ are the bounds for endpoint $k^*$; the reflected bounds are $-k_+,-k_-$. Consequently, $|\PathsDk{k^*}|=|\PathsDk{-k^*}|$ and 
$|\mathcal C_{D,k^*}|=|\mathcal C_{D,-k^*}|.$ The same reflection maps nonnegative paths ending at $k^*$ to nonpositive paths ending at $-k^*$.
\end{theorem}

The proof is given in \Cref{app:reflection-symmetry}. Hence it suffices to construct the full tuple family for nonnegative endpoints; negative endpoints follow by reflection. This symmetry should not be confused with the positive/mixed split at a fixed endpoint: $\mathcal C^-_{D,k^*}$ contains tuples that visit negative positions while still terminating at the same $k^*$.

\section*{\textbf{Case 3: $k^*=D-2$, $D\ge2$}}

Here $k_+=D-1$, $D+k^*=2D-2$ is even, and $\beta=1$.
The mass-layer range is $0\le i\le2$, with $\ell(1)=D-1,\;
\ell(2)=D-2.$ The seed contributes one tuple, layer $i=1$ contributes $D-1$ tuples, and layer $i=2$ contributes
$\binom{D-1}{2}$ tuples. When $D=2$, the last layer has
zero receiving slots and contributes zero, consistently
with $\binom{1}{2}=0$. Therefore
\begin{equation}
\label{eq:case-three-positive}
    |\mathcal C^+_{D,D-2}|
    =1+(D-1)+\binom{D-1}{2}
    =1+\binom{D}{2}.
\end{equation}
Exactly one path ending at $D-2$ visits a negative position:
it takes one downward step from the root, followed by
$D-1$ upward steps. Its tuple is the unique element of
$\mathcal C^-_{D,D-2}$. Consequently,
\begin{equation}
\label{eq:case-three-count}
\begin{aligned}
    |\mathcal C_{D,D-2}|
    =
    |\mathcal C^+_{D,D-2}|
    +
    |\mathcal C^-_{D,D-2}|
    =
    \binom{D}{2}+2.
\end{aligned}
\end{equation}
The uniqueness of the mixed path and the boundary cases are proved in \Cref{app:case-three-count}.This is the first endpoint case for which the positive family does not exhaust $\mathcal C_{D,k^*}$: one additional mixed tuple is required. For example, when $D=7$ this case has
$k^*=D-2=5$ and $|\mathcal C_{7,5}|=23.$ It is therefore distinct from the running example $(D,k^*)=(7,2)$. For general nonnegative endpoints, mixed paths may reach several negative positions, so we next characterize their admissible cardinality tuples systematically.

\section*{\textbf{Case N:} $k^{*} \in \{-D,-D+1,\dots,D\}$}

By \Cref{thm:reflection-symmetry}, it suffices to construct
and count the full tuple family for nonnegative endpoints;
negative endpoints are then obtained by reflection.
Fix $D\in\Z_+$ and $k^*\in\{0,\dots,D\}$, and let
$\ppath\in\PathsDk{k^*}$ have a cardinality tuple in
$\mathcal C^-_{D,k^*}$. Because such a path visits at least one negative position but terminates at the nonnegative position $k^*$, it must revisit position $0$ after entering the negative region. This is an additional \emph{visit to position $0$}, not a required stay: a stay has increment $0$, whereas the return to $0$ may occur through an ordinary upward step. The magenta path below illustrates this distinction.
\begin{figure}[H]
\centering
\begin{tikzpicture}[
  x=0.3cm,y=0.3cm, >=Stealth,
  every node/.style={circle, draw, minimum size=0.5cm, inner sep=0pt},
  edgepos/.style={}
]
  \node[fill=yellow!40] (n0) at (0,0) {0};
  
  \node (n1u) at (3,  2) {1};
  \node (n1m) at (3,  0) {0};
  \node[fill=yellow!40] (n1d) at (3, -2) {-1};
  
  \node (n2uu) at (6,  4) {2};
  \node (n2um) at (6,  2) {1};
  \node[fill=yellow!40] (n2mm) at (6,  0) {0};
  \node (n2dm) at (6, -2) {-1};
  \node (n2dd) at (6, -4) {-2};
  
  \node (n3uuu) at (9,  6) {3};
  \node (n3uum) at (9,  4) {2};
  \node (n3umm) at (9,  2) {1};
  \node (n3mmm) at (9,  0) {0};
  \node (n3dmm) at (9, -2) {-1};
  \node (n3ddm) at (9, -4) {-2};
  \node (n3ddd) at (9, -6) {-3};
  
  \node (n4uuuu) at (12,  8) {4};
  \node (n4uuum) at (12,  6) {3};
  \node (n4uumm) at (12,  4) {2};
  \node (n4ummm) at (12,  2) {1};
  \node (n4mmmm) at (12,  0) {0};
  \node (n4dmmm) at (12, -2) {-1};
  \node (n4ddmm) at (12, -4) {-2};
  \node (n4dddm) at (12, -6) {-3};
  \node (n4dddd) at (12, -8) {-4};
  
  \draw[green, thick] (n0) -- (n1u);
  \draw               (n0) -- (n1m);
  \draw[magenta, thick] (n0) -- (n1d);
  
  \draw[green, thick] (n1u) -- (n2uu);
  \draw               (n1u) -- (n2um);
  \draw               (n1u) -- (n2mm);
  
  \draw               (n1m) -- (n2um);
  \draw               (n1m) -- (n2mm);
  \draw               (n1m) -- (n2dm);
  
  \draw[magenta, thick] (n1d) -- (n2mm);
  \draw               (n1d) -- (n2dm);
  \draw               (n1d) -- (n2dd);
  
  \draw[green, thick] (n2uu) -- (n3uuu);
  \draw               (n2uu) -- (n3uum);
  \draw               (n2uu) -- (n3umm);
  
  \draw               (n2um) -- (n3uum);
  \draw               (n2um) -- (n3umm);
  \draw               (n2um) -- (n3mmm);
  
  \draw[magenta, thick] (n2mm) -- (n3umm);
  \draw               (n2mm) -- (n3mmm);
  \draw               (n2mm) -- (n3dmm);
  
  \draw               (n2dm) -- (n3mmm);
  \draw               (n2dm) -- (n3dmm);
  \draw               (n2dm) -- (n3ddm);
  
  \draw               (n2dd) -- (n3dmm);
  \draw               (n2dd) -- (n3ddm);
  \draw               (n2dd) -- (n3ddd);
  
  \draw               (n3uuu) -- (n4uuuu);
  \draw               (n3uuu) -- (n4uuum);
  \draw[green, thick] (n3uuu) -- (n4uumm);
  
  \draw               (n3uum) -- (n4uuum);
  \draw               (n3uum) -- (n4uumm);
  \draw               (n3uum) -- (n4ummm);
  
  \draw[magenta, thick] (n3umm) -- (n4uumm);
  \draw               (n3umm) -- (n4ummm);
  \draw               (n3umm) -- (n4mmmm);
  
  \draw               (n3mmm) -- (n4ummm);
  \draw               (n3mmm) -- (n4mmmm);
  \draw               (n3mmm) -- (n4dmmm);
  
  \draw               (n3dmm) -- (n4mmmm);
  \draw               (n3dmm) -- (n4dmmm);
  \draw               (n3dmm) -- (n4ddmm);
  
  \draw               (n3ddm) -- (n4dmmm);
  \draw               (n3ddm) -- (n4ddmm);
  \draw               (n3ddm) -- (n4dddm);
  
  \draw               (n3ddd) -- (n4ddmm);
  \draw               (n3ddd) -- (n4dddm);
  \draw               (n3ddd) -- (n4dddd);
\end{tikzpicture}
\caption{Nonnegative (green) and mixed (magenta) paths ending
at $(k^*,D)=(2,4)$. The magenta path has position sequence
$(0,-1,0,1,2)$ and revisits $0$ without taking a stay.}
\label{fig:mixed-path-kstar-D-2}
\end{figure}
For $\hat c=\hat C_{D,k^*}(\ppath)=(c_k)_{k=k_-}^{k_+}$, write
\[
a\Def a(\hat c),
\quad
b\Def b(\hat c)
\]
for the actual support endpoints defined in \eqref{eq:actual-support}.
Since $\hat c\in\mathcal C^-_{D,k^*}$ and $k^*\ge0$, $k_-\le a<0\le k^*\le b\le k_+.$ The visit count at position $0$ satisfies
\begin{equation}
\label{condition}
    c_0\ge
    \begin{cases}
        2,
        & k^*>0,\\[2pt]
        2,
        & k^*=0\text{ and }b=0,\\[2pt]
        3,
        & k^*=0\text{ and }b>0.
    \end{cases}
\end{equation}
These are lower bounds rather than fixed counts. A mixed path
must return through position $0$ after entering the negative
region. Thus $0$ is visited at least twice. When $k^*=0$ and
the path also reaches a strictly positive position, the
initial, intervening, and terminal visits to $0$ are distinct,
which gives the third case in \eqref{condition}. The same crossing argument applies throughout the support.
Since $a<0$ is the lowest occupied position, reaching $a$ and
then returning to the nonnegative endpoint forces at least two
visits to every intermediate negative position
$a<k<0$, while $a$ itself requires at least one visit.
Likewise, every path ending at $k^*\ge0$ must visit each
position $0,\dots,k^*$ at least once, with additional return
visits forced whenever the path reaches above $k^*$. These requirements depend on the tuple's actual support
$[a,b]$, not on the full feasible range $[k_-,k_+]$.
They are combined exactly by the following criterion.
\begin{proposition}[Admissibility with prescribed support]
\label{prop:support-admissibility}
Fix $D\in\Z_+$ and $k^*\in\{0,\dots,D\}$, and let
$a,b$ be integers satisfying
\[
    k_-\le a\le0\le k^*\le b\le k_+.
\]
Let $\hat c=(c_k)_{k=k_-}^{k_+}\in\Z_+^{[k_-,k_+]}$
have total mass $D+1$.
Then $\hat c\in\mathcal C_{D,k^*}$ with occupied positions
exactly $[a,b]$ if and only if
\begin{equation}
\label{eq:support-admissibility}
\begin{aligned}
    c_k&=0,
        && k\in[k_-,k_+]\setminus[a,b],\\
    c_k&\ge
        1+\mathbf 1_{\{a<k\le0\}}
         +\mathbf 1_{\{k^*\le k<b\}},
        && k\in[a,b].
\end{aligned}
\end{equation}
Each indicator is $1$ when its subscript condition holds
and $0$ otherwise. The case $a=0$ describes positive tuples;
the case $a<0$ describes mixed tuples.
\end{proposition}

\begin{proof}
For necessity, a path visiting $a$ and $b$ must visit every
integer position between them. Each position $a<k\le0$
is visited before reaching $a$ and again afterward, because
the endpoint is nonnegative. Likewise, each position
$k^*\le k<b$ is visited before reaching $b$ and again
afterward, because the path ends at $k^*$. The two ranges overlap only at $k=0$ when $k^*=0$ and $a<0<b$. In that case \eqref{condition} records the three distinct initial, intervening, and terminal visits, so both indicator contributions are required in \eqref{eq:support-admissibility}. For sufficiency, start with the path that descends from $0$ to $a$, ascends from $a$ to $b$, and descends from $b$ to $k^*$, with no stays. It has $(-a)+(b-a)+(b-k^*)=2(b-a)-k^*$ steps. Counting the junction visits at $a$ and $b$ only once, the ascent accounts for one visit to each position in $[a,b]$. Positions in $(a,0]$ occur once more before this ascent, and positions in $[k^*,b)$ occur once more afterward. Hence its visit counts equal the lower bounds in \eqref{eq:support-admissibility}. At the first visit to each position, insert enough stays to raise its count to the prescribed $c_k$.
The inequalities guarantee that the number inserted at each position is nonnegative. Their total is
\[
\begin{aligned}
    &\sum_{k=a}^{b}
       \left(c_k-1-\mathbf 1_{\{a<k\le0\}}
                   -\mathbf 1_{\{k^*\le k<b\}}\right)\\
    &\qquad=(D+1)-\bigl[(b-a+1)+(-a)+(b-k^*)\bigr]\\
    &\qquad=D-2(b-a)+k^*.
\end{aligned}
\]
The resulting path therefore has exactly $\bigl(2(b-a)-k^*\bigr)+\bigl(D-2(b-a)+k^*\bigr)=D$ steps, ends at $k^*$, has occupied positions exactly $[a,b]$, and realizes $\hat c$.
When $a=b=0$, necessarily $k^*=0$; the construction is $D$ stays at the root and also covers $D=0$.
\end{proof}

This criterion establishes which tuples are realizable.
To make the stage structure induced by the cardinality order
explicit, it is useful to separate the nonnegative and
negative coordinate blocks.

\begin{remark}[Stage-compatible representation of the cardinality order]
\label{rem:stage-cardinality-order}
Fix $D\in\Z_+$ and $k^*\in\{0,\dots,D\}$, and define the
ordered index lists
\begin{equation}
\label{eq:stage-index-blocks}
    I_+
    \Def
    (k_+,k_+-1,\dots,0),
    \quad
    I_-
    \Def
    (-1,-2,\dots,k_-).
\end{equation}
When $k_-=0$, the list $I_-$ is empty. For
$\hat c=(c_k)_{k=k_-}^{k_+}\in\mathcal C_{D,k^*}$, the
comparison key already defined in
\eqref{eq:cardinality-order-key} may equivalently be written as
\begin{equation}
\label{eq:stage-cardinality-key}
    \Phi(\hat c)
    =
    \bigl(
        a(\hat c);
        (c_k)_{k\in I_+};
        (c_k)_{k\in I_-}
    \bigr)
    \in
    \{k_-,\dots,0\}
    \times
    \Z_+^{|I_+|}
    \times
    \Z_+^{|I_-|}.
\end{equation}
Thus the first component records the lowest occupied
position, while the remaining components contain every
visit count in decreasing position order
$k_+,k_+-1,\dots,k_-$. This is only a block decomposition of the cardinality key; it does not define a new ordering. Under
\Cref{card_ordering}, tuples are first compared by
$a(\hat c)$. Hence descending order places the positive
family $a=0$ first, followed by the families with minima
$-1,-2,\dots,k_-$. Tuples with a common minimum therefore
form a contiguous block. Within such a block, the first
component of $\Phi$ is fixed, so comparison reduces to the
right-to-left comparison of the visit counts. Because every visit count occurs in $\Phi$, the key is injective. In particular, this representation loses no information from the cardinality tuple; it simply exposes the outer stage variable separately from the coordinates
used by the within-stage enumeration.
\end{remark} 

We now use the contiguous blocks identified above to
formalize the stage construction and its successor
operations.

\paragraph{\textbf{The Shift Function}}
Recall the positive maximal seed $\hat s^{(0)}$ from
\Cref{E:starting_tuple}. The intended stages have the following lowest occupied positions $a=0,-1,-2,\dots,k_-$, in that order. Within each stage, tuples are compared from the highest position downward, as in \Cref{card_ordering}. Each completed tuple must satisfy \Cref{prop:support-admissibility}; the stage construction must establish that all such tuples are generated without repetition.

\paragraph{Stage and Mass-Layer Indices.}

We now distinguish the \emph{outer stage index} from the mass-layer index used within a stage. Fix $D\in\Z_+$ and $k^*\in\{0,\dots,D\}$. The physical endpoint $k^*$ remains fixed throughout the construction. The outer stage index $M$ instead records the lowest occupied position, $a=-M.$ Thus $M=0$ is the positive stage, $M=1$ contains exactly the tuples whose lowest occupied position is $-1$, and so on. For each integer $M$ with $0\le M\le-k_-$, define
\begin{equation}
\label{eq:stage-blocks}
    \mathcal B_M
    \Def
    \left\{
        \hat c\in\mathcal C_{D,k^*}:
        c_{-M}>0,\quad
        c_k=0
        \text{ for }k_-\le k<-M
    \right\}.
\end{equation}
Equivalently, $\mathcal B_M$ consists exactly of the
cardinality tuples whose lowest occupied position is $-M$.
Consequently,
\begin{equation}
\label{eq:stage-partition}
    \mathcal C_{D,k^*}
    =
    \bigsqcup_{M=0}^{-k_-}\mathcal B_M,
    \quad
    \mathcal B_0=\mathcal C^+_{D,k^*}.
\end{equation}
The union is disjoint because every admissible tuple has a
unique lowest occupied position. Under the stage-compatible
ordering of \Cref{card_ordering}, the blocks occur in increasing
$M$ when the tuples are listed in descending order: $\mathcal B_0,\;\mathcal B_1,\;\mathcal B_2,\;\dots,\;\mathcal B_{-k_-}.
$ The possible stage indices also follow directly from the
geometry of the tree. A path that reaches $-M$ and ends at
$k^*$ requires at least $M+(M+k^*)=2M+k^*$ steps. Hence $\mathcal B_M$ can be nonempty only if $2M+k^*\le D.$
Conversely, whenever this inequality holds, one may take
$M$ downward steps to $-M$, followed by $M+k^*$ upward steps to $k^*$, and use the remaining $D-2M-k^*$ steps as stays at $k^*$. Therefore the admissible outer-stage range is exactly
\[
    0\le M\le
    \left\lfloor\frac{D-k^*}{2}\right\rfloor
    =
    -k_-
    =
    k_+-k^*.
\]

Now fix a stage $M$ and let $b$ denote the actual highest
occupied position of a tuple in $\mathcal B_M$.
Since its actual support is $[-M,b]$,
\Cref{prop:support-admissibility} implies
\begin{equation}
\label{eq:stage-support-bound}
    2(b+M)-k^*\le D.
\end{equation}
Thus $k^*\le b\le k_+-M.$ Every tuple in $\mathcal B_M$ is therefore zero outside the common active window $[-M,k_+-M].$
Its actual highest occupied position need not equal the right endpoint $k_+-M$ of this window. It is important to distinguish the outer stage index $M$ from the residual mass-layer index $i$ used in the positive construction. The stage specifies the lowest occupied position. Within one stage there may be several mass layers, and each mass layer may itself contain several cardinality tuples.

\begin{proposition}[Maximal seed of a stage]
\label{prop:stage-seed}
For each $M\in\{0,\dots,-k_-\}$, the unique largest tuple
of $\mathcal B_M$ under the ordering of \Cref{card_ordering}
is $\hat s^{(M)}=(c_k)_{k=k_-}^{k_+},$ where
\begin{equation} \label{eq:stage-seed} c_k = \begin{cases} \begin{aligned} 1 &+\mathbf 1_{\{-M<k\le 0\}} +\mathbf 1_{\{k^*\le k<k_+-M\}}\\ &+(\beta-1)\mathbf 1_{\{k=k_+-M\}}, \end{aligned} & -M\le k\le k_+-M,\\[8pt] 0, & k\in[k_-,k_+]\setminus[-M,k_+-M]. \end{cases} \end{equation}
Here $\beta$ is the initial parity parameter defined in
\eqref{switching_term}. When $M=0$,
\eqref{eq:stage-seed} reduces exactly to the positive seed
\eqref{E:starting_tuple}.
\end{proposition}

\begin{proof}
For the support $[-M,k_+-M],$ the minimum occupancies prescribed by \eqref{eq:support-admissibility} have total
$2\bigl((k_+-M)+M\bigr)-k^*+1=2k_+-k^*+1.$ Equation \eqref{eq:stage-seed} adds exactly $\beta-1$ further visits at the highest occupied position. Hence
\[
\begin{aligned}
    \sum_{k=k_-}^{k_+}c_k
    =
    2k_+-k^*+1+(\beta-1)
    =
    2k_+-k^*+\beta
    =
    D+1,
\end{aligned}
\]
where the final equality follows from
\eqref{switching_term}.
Therefore \Cref{prop:support-admissibility} gives $\hat s^{(M)}\in\mathcal B_M.$ Now let
$\hat c'=(c'_k)_{k=k_-}^{k_+}\in\mathcal B_M$. By \eqref{eq:stage-support-bound}, no position above $k_+-M$ can be occupied. If $c'_{k_+-M}=0,$ then the seed is strictly larger because its count at the highest possible occupied position is positive. Otherwise $\hat c'$ has actual support $[-M,k_+-M].$ For every $k<k_+-M$,
\Cref{prop:support-admissibility} gives $c'_k\ge c_k,$
because the seed uses exactly the minimum required occupancy at each such position. Since both tuples have total mass $D+1$,
\[
\begin{aligned}
    c'_{k_+-M}
    =
    D+1-\sum_{k=-M}^{k_+-M-1}c'_k
    \le
    D+1-\sum_{k=-M}^{k_+-M-1}c_k
    =
    c_{k_+-M}.
\end{aligned}
\]
If the inequality is strict, the seed is strictly larger.
If equality holds, every lower-coordinate inequality must
also be an equality, and therefore
$\hat c'=\hat s^{(M)}$.
Thus $\hat s^{(M)}$ is the unique largest tuple of
$\mathcal B_M$.
\end{proof}

The proposition gives the exact seed of every stage, but it
is useful to retain the explicit first-stage transition that
motivates the shift construction. When $k^*>0$ and $-k_-\ge1$, the first two stage seeds have the following schematic form. Empty index ranges in this display are simply omitted:
\begin{equation}
\label{eq:first-stage-seed-transition}
\begin{split}
\hat{s}^{(M=0)}
&=
\Bigl(
\underbrace{0,\;0,\;\dots,\;0}
    _{\scriptscriptstyle c_k,\;k\in[k_-,-1]},
\;
\underset{c_0}{1},
\;
\underbrace{1,\;\dots,\;1}
    _{\scriptscriptstyle c_k,\;k\in[1,k^*-1]},
\;
\underbrace{2,\;\dots,\;2}
    _{\scriptscriptstyle c_k,\;k\in[k^*,k_+-1]},
\;
\underset{c_{k_+}}{\beta}
\Bigr)
\\[4pt]
&\hspace{2.2cm}\Downarrow
\qquad
\text{activate the new lowest occupied position $-1$}
\\[4pt]
\hat{s}^{(M=1)}
&=
\Bigl(
\underbrace{0,\;0,\;\dots,\;0}
    _{\scriptscriptstyle c_k,\;k\in[k_-,-2]},
\;
\underset{c_{-1}}{1},
\;
\underset{c_0}{2},
\;
\underbrace{1,\;\dots,\;1}
    _{\scriptscriptstyle c_k,\;k\in[1,k^*-1]},
\;
\underbrace{2,\;\dots,\;2}
    _{\scriptscriptstyle c_k,\;k\in[k^*,k_+-2]},
\;
\underset{c_{k_+-1}}{\beta},
\;
\underset{c_{k_+}}{0}
\Bigr).
\end{split}
\end{equation}
Thus the first shift changes the common active window from
$[0,k_+]$ to  $[-1,k_+-1].$ The physical endpoint $k^*$ has not moved. Rather, the new stage activates the position $-1$, imposes the required second visit to $0$, and moves the largest possible occupied position one place to the left. The remaining $\beta-1$ visits are again placed at the highest occupied position, as prescribed by \eqref{eq:stage-seed}.
For the running example
\[
    D=7,\qquad k^*=2,\qquad
    k_-=-2,\qquad k_+=4,\qquad\beta=2,
\]
the same transition is
\[
\begin{array}{c|ccccccc}
k
    & -2 & -1 & 0 & 1 & 2 & 3 & 4\\
\hline
\hat s^{(0)}
    & 0 & 0 & 1 & 1 & 2 & 2 & 2\\
\hat s^{(1)}
    & 0 & 1 & 2 & 1 & 2 & 2 & 0
\end{array}
\]
and the stage-$1$ seed is realized, for example, by $(0,-1,0,1,2,3,3,2).$ The following stage has seed:
\begin{equation*}
    \hat s^{(2)}=(1,2,2,1,2,0,0)
\end{equation*}
realized by $(0,-1,-2,-1,0,1,2,2).$ Again, the terminal position remains $k^*=2$ throughout. When $k^*=0$, the same general seed formula \eqref{eq:stage-seed} applies, but the overlap of the negative-return and positive-return requirements at $0$ must be retained. In particular, if $M>0$ and the support extends strictly above $0$, then $c_0=3$ in the minimal stage seed. If the highest occupied position is itself $0$, then the seed has $c_0=\beta+1$. This agrees with
\eqref{condition}. We can now retain the staircase interpretation of the original construction. The reader should picture each stage as beginning at its unique maximal seed, followed by all remaining tuples in that stage before the next negative level is activated:
\begin{align}
\hat s^{(0)}
    &\leadsto \cdots
    \leadsto
    \bigl[\mathcal B_0\text{ exhausted}\bigr]
    \leadsto
    \hat s^{(1)}
    \leadsto \cdots
    \leadsto
    \bigl[\mathcal B_1\text{ exhausted}\bigr]
    \notag\\
    &\leadsto
    \hat s^{(2)}
    \leadsto \cdots
    \leadsto
    \bigl[\mathcal B_2\text{ exhausted}\bigr]
    \leadsto \cdots
    \notag\\
    &\leadsto
    \hat s^{(-k_-)}
    \leadsto \cdots
    \leadsto
    \bigl[\mathcal B_{-k_-}\text{ exhausted}\bigr].
\label{eq:stage-staircase}
\end{align}
The arrows in \eqref{eq:stage-staircase} describe the
intended enumeration order. Within $\mathcal B_M$, every
tuple has the same lowest occupied position $-M$ and is
compared by the within-stage right-to-left rule of
\Cref{card_ordering}. Once that block has been exhausted, the
enumeration proceeds to the maximal seed $\hat s^{(M+1)}$ of the next stage. Thus the outer staircase moves through the lowest occupied positions $0,-1,-2,\dots,k_-,$ while the physical endpoint $k^*$ remains fixed. The remaining task is to prove that the within-stage successor and the reseeding operation actually realize \eqref{eq:stage-staircase} without omission or duplication.

\paragraph{\textbf{Formalizing the Shift Function}}

For $M\in\{0,1,\dots,-k_-\}$, the support bound
\eqref{eq:stage-support-bound} gives the common active
index window as the following left translate of
$[0,k_+]$:
\begin{equation}
\label{eq:IM-def}
  I_M \Def \bigl([0,k_{+}] - M\bigr)\cap\mathbb{Z}
  = \{\,k-M : k\in[0,k_{+}]\,\}\cap\mathbb{Z}
  = [-M,\,k_{+}-M]\cap\mathbb{Z}
\end{equation}
Equivalently, the blocks satisfy the recursion
\begin{equation}
\label{eq:IM-recursion}
  I_{M+1} \;=\; I_M - 1 \;=\; \{\,k-1 : k\in I_M\,\}.
\end{equation}
Thus $I_0 = [0,k_{+}]$ and $I_1 = [-1,\,k_{+}-1]$, etc., i.e., the active set is pushed one unit to the left at each stage.

\begin{remark}[Position of $k^*$]
For a fixed stage $M$, define the coordinate translation
\[
    \tau_M:\Z\to\Z,
    \quad
    \tau_M(k)\Def k+M.
\]
Since
\[
    \tau_M(I_M)=[0,k_+],
    \quad
    \tau_M(0)=M,
    \quad
    \tau_M(k^*)=k^*+M,
\]
the root and endpoint have local coordinates $M$ and
$k^*+M$, respectively. Neither physical position changes.
In particular, translation alone does not turn the stage
into a positive-path problem rooted at local position $0$.
\end{remark}

We now define the seed-to-seed map $\mathcal S_M$ for
$0\le M<-k_-$. Its input is the canonical seed
$\hat s^{(M)}=(c_k)_{k=k_-}^{k_+}$ from
\eqref{eq:stage-seed}, not an arbitrary tuple or the last
emitted tuple of $\mathcal B_M$. The strict bound on $M$ guarantees $k_+-M>k^*\ge0,$ so the current seed's highest position is positive and its count is exactly $\beta$. Form the temporary array, indexed locally by $0,\dots,k_+$,
\[
    v=(1,c_{-M},c_{-M+1},\dots,c_{k_+-M-1}).
\]
The old peak is omitted, and the leading $1$ supplies the
visit to the new minimum. Entry $v_j$ represents physical
position $j-M-1$. Apply these updates in the order shown:
\begin{itemize}
    \item \textbf{Even $D+k^*$: $\beta=1$.}
    Reduce the rightmost retained count by one and add
    one visit at the former minimum:
    \[
        v_{k_+}\leftarrow v_{k_+}-1,
        \qquad
        v_1\leftarrow v_1+1.
    \]

    \item \textbf{Odd $D+k^*$: $\beta=2$.}
    The rightmost retained count is already correct.
    Add one visit at the former minimum:
    \[
        v_1\leftarrow v_1+1.
    \]
\end{itemize}
The first update uses a position, not a search for a count
equal to $2$. Before the updates, $v_{k_+}\ge2$; it can
equal $3$ when that position is $0$ and both sides have
been visited. Also $k_+\ge1$ whenever a shift is possible.
If $k_+=1$, the two assignments in the $\beta=1$ case
act on the same entry and cancel. Return the full, zero-padded tuple on the unchanged coordinate range $[k_-,k_+]$:
\begin{equation}
\label{shift_process_A}
\begin{aligned}
    \mathcal S_M(\hat s^{(M)})&\Def(c'_k)_{k=k_-}^{k_+},\\
    c'_k&=
    \begin{cases}
        v_{k+M+1},&-M-1\le k\le k_+-M-1,\\
        0,&k\in[k_-,k_+]\setminus[-M-1,k_+-M-1].
    \end{cases}
\end{aligned}
\end{equation}

\paragraph{Correctness of reseeding.}
In physical coordinates, the change from the input is
\[
\begin{aligned}
    c'_k-c_k
    &=\mathbf 1_{\{k=-M-1\}}+\mathbf 1_{\{k=-M\}}\\
    &\quad-\beta\mathbf 1_{\{k=k_+-M\}}
          -(2-\beta)\mathbf 1_{\{k=k_+-M-1\}},
    \qquad k\in[k_-,k_+].
\end{aligned}
\]
Extending the left support boundary adds one visit at the
new minimum and one at the old minimum. At the right
boundary, the old peak is removed and the retained
neighbor loses $2-\beta$ visits. Comparing
\eqref{eq:stage-seed} at $M$ and $M+1$ gives exactly these
coordinate changes, including coincident update positions.
The total change is $1+1-\beta-(2-\beta)=0.$ The result is therefore exactly $\hat s^{(M+1)}$ from \eqref{eq:stage-seed}, with total mass $D+1$. By \Cref{prop:stage-seed}, it is admissible and is the unique largest tuple of $\mathcal B_{M+1}$. For $k^*=0$ with new highest position $0$, the old count there is $2+\mathbf 1_{\{M>0\}}$. The increment at the
former minimum affects $0$ exactly when $M=0$. Its new count is consequently $2+\mathbf 1_{\{M>0\}}+\mathbf 1_{\{M=0\}}
-(2-\beta)=\beta+1,$ as required by the stage-seed formula.
The map is used after the current block has been exhausted, but is evaluated on its saved or reconstructed canonical seed. Generating the intervening tuples is a separate within-stage operation. No reseeding is performed at $M=-k_-$. The physical endpoint remains $k^*$; the next window is $I_{M+1}=I_M-1$.

\paragraph{Residual depth within a stage}
At stage $M$, reaching $-M$ and then terminating at $k^*$
requires at least $M+(M+k^*)=k^*+2M$ steps, as in the stage-feasibility argument above. Consequently, $D-k^*-2M$ steps remain beyond this mandatory traversal. This residual depth determines the range of the within-stage mass layers. The physical endpoint remains $k^*$ throughout.

\paragraph{Reseeded counting at each stage}
Translation alone does not identify a stage with a
positive-path problem: its translated root is at $M$.
The following tuple correspondence also removes one
mandatory return visit at each position $-M+1,\dots,0$.

\begin{proposition}[Reduction of a stage to positive tuples]
\label{prop:stage-positive-reduction}
Fix $D\in\Z_+$, $k^*\in\{0,\dots,D\}$, and
$M\in\{0,\dots,-k_-\}$.
The assignment
\begin{equation}
\label{eq:stage-positive-map}
\begin{aligned}
    \mathcal B_M&\longrightarrow\mathcal C^+_{D-M,k^*+M},
        &\hat c&\longmapsto\hat c',\\
    c'_j&\Def c_{j-M}-\mathbf 1_{\{1\le j\le M\}},
        &j&\in[k_-+M,k_+],
\end{aligned}
\end{equation}
where $\hat c'=(c'_j)_{j=k_-+M}^{k_+}$, is a bijection.
It preserves the within-stage right-to-left order and
sends $\hat s^{(M)}$ to the maximal positive seed of
this associated problem. Each stage therefore inherits the positive mass layers, with indices $0\le i\le(D-M)-(k^*+M)=D-k^*-2M.$ For $i\ge1$ in this range, the number of receiving slots is
\[
    \ell(i)
    =\left\lceil\frac{(D-M)+(k^*+M)-i}{2}\right\rceil
    =\left\lceil\frac{D+k^*-i}{2}\right\rceil.
\]
This is the same slot function as in the positive stage.
Consequently,
\begin{equation}
\label{eq:stage-contrib-clean}
\begin{aligned}
    |\mathcal B_M|
    =|\mathcal C^+_{D-M,k^*+M}|
    =1+\sum_{i=1}^{D-k^*-2M}
        \binom{i+\ell(i)-1}{\ell(i)-1}.
\end{aligned}
\end{equation}
The leading $1$ counts the stage seed. Empty sums and
positive-mass layers with zero receiving slots contribute
zero, as in the positive count.
\end{proposition}
\begin{proof}
The stage bound gives $k^*+M\le D-M$, so the associated
positive problem has admissible depth and endpoint.
Its feasible coordinate bounds are
\[
\begin{aligned}
    \left\lceil\frac{(k^*+M)-(D-M)}{2}\right\rceil
        &=k_-+M,\\
    \left\lfloor\frac{(D-M)+(k^*+M)}{2}\right\rfloor
        &=k_+.
\end{aligned}
\]
Let $\hat c\in\mathcal B_M$ have actual support $[-M,b]$.
For $0\le j\le b+M$, the support criterion gives
\[
    c_{j-M}\ge
    1+\mathbf 1_{\{1\le j\le M\}}
     +\mathbf 1_{\{k^*+M\le j<b+M\}}.
\]
After subtracting this indicator, the remaining bounds
are those for positive tuples ending at $k^*+M$.
All coordinates outside $[0,b+M]$ are zero. Moreover,
\[
    \sum_{j=k_-+M}^{k_+}c'_j=(D+1)-M=(D-M)+1.
\]
Thus $\hat c'\in\mathcal C^+_{D-M,k^*+M}$ by
\Cref{prop:support-admissibility} applied to that problem. Conversely, from any tuple in that positive family,
recover the original coordinates by
\[
    c_k=
    \begin{cases}
        c'_{k+M}+\mathbf 1_{\{-M<k\le0\}},
            &k\in I_M,\\
        0,  &k\in[k_-,k_+]\setminus I_M.
    \end{cases}
\]
This restores exactly the negative-return lower bounds,
adds $M$ visits, and gives total mass $D+1$ and minimum
occupied position $-M$. The support criterion therefore
places the reconstructed tuple in $\mathcal B_M$.
The two coordinate operations are inverse, proving
bijectivity. Translation preserves the order of position indices, and subtracting the same fixed vector from every tuple preserves the first differing coordinate and its comparison. The bijection therefore preserves the within-stage order and sends its unique largest tuple to the unique largest positive tuple. Applying \Cref{prop:positive-mass-layers} with depth $D-M$ and endpoint $k^*+M$ proves the layer
statement and \eqref{eq:stage-contrib-clean}.
\end{proof}

This is a bijection of cardinality tuples, not a claim
that deleting selected steps gives a bijection of raw paths.
It justifies using the positive weak-composition construction
and restoring the fixed return counts; it does not require
first enumerating raw paths. Since $(D-M)+(k^*+M)=D+k^*$, the associated problem has the same $k_+$ and initial parity parameter $\beta$. Each positive-mass summand may therefore also be written using \eqref{eq:single-summand}, including the factor $(\beta-1)\delta(i)$. Across stages, the slot count at a fixed $i$ is unchanged; the upper mass index decreases by two. Replacing the endpoint by $k^*+M$ while keeping depth $D$ would describe a different counting problem. The associated positive family changes both parameters, from $(D,k^*)$ to
$(D-M,k^*+M)$. For the running example $D=7$, $k^*=2$, the correspondence and layer counts are
\[
\begin{array}{c|c|c|r}
 M & \text{associated positive family}
   & \text{layer contributions} & |\mathcal B_M|\\
\hline
 0 & \mathcal C^+_{7,2} & 1+4+10+10+15+6 & 46\\
 1 & \mathcal C^+_{6,3} & 1+4+10+10       & 25\\
 2 & \mathcal C^+_{5,4} & 1+4             & 5
\end{array}
\]
The initial $1$ in each row is its stage seed. These
counts concern distinct tuples, not raw-path multiplicities.

\paragraph{Total reseeded count and stopping}
The stage partition gives the final stage index
\begin{equation}
\label{eq:outer-stop}
    T\Def\min\{k_+-k^*,-k_-\}
      =-k_-
      =\left\lfloor\frac{D-k^*}{2}\right\rfloor.
\end{equation}
A next stage exists exactly when $M<T$.
After exhausting $\mathcal B_M$, its next seed is computed
from the saved or reconstructed canonical seed by
\[
    \hat s^{(M+1)}=\mathcal S_M(\hat s^{(M)}),
    \qquad 0\le M<T.
\]
The last emitted tuple is not used as the input to this
seed-to-seed map. Summing the disjoint stage counts yields
\begin{equation}
\label{eq:dynamic-total-count-final}
\begin{aligned}
    |\mathcal C_{D,k^*}|
    =
    \sum_{M=0}^{T}|\mathcal B_M|
    =(T+1)+\sum_{M=0}^{T}\sum_{i=1}^{D-k^*-2M}
        \binom{i+\ell(i)-1}{\ell(i)-1}.
\end{aligned}
\end{equation}
In particular, $|\mathcal C_{7,2}|=46+25+5=76$.
Negative endpoints are handled by
\Cref{thm:reflection-symmetry}.

\begin{remark}[Stopping and boundary layers]
The stopping index follows from the minimum required
$2M+k^*$ steps, not from a change in the physical endpoint.
At $M=T$, the remaining mass-layer range ends at
\[
    D-k^*-2T=\beta-1\in\{0,1\}.
\]
When it ends at zero, only the stage seed remains.
When it ends at one, the seed and the one-unit layer
are counted, with a zero-slot layer contributing zero.
The formulas also cover $D=k^*$, including $D=0$.
The stage partition and layer bijections establish the
count. The tuple-emitting routine in \Cref{alg:unique_path}
traverses these same blocks and constructs a full tuple
for each residual composition. Its correctness is proved
in \Cref{thm:tuple-enumeration-correctness}.
\end{remark}

\subsection{Exact Degeneracy of a Cardinality Tuple}
\label{subsec:tuple-degeneracy}

The preceding construction counts each admissible cardinality tuple once. To recover a count over raw paths, we must also determine how many paths each tuple represents. We call this multiplicity its \emph{path degeneracy}. Fix $D\in\Z_+$ and $k^*\in\Z$ with $|k^*|\le D$. The bounds $k_-,k_+$ and the cardinality map retain their definitions in \eqref{E:Kbounds} and \eqref{card_map}.

\begin{definition}[Path degeneracy]
\label{def:path-degeneracy}
For $\hat c\in\mathcal C_{D,k^*}$, define
\begin{equation}
\label{eq:path-degeneracy-definition}
\begin{aligned}
    \deg_1(\hat c;k^*)
    &\Def\left|\hat C_{D,k^*}^{-1}(\{\hat c\})\right|\\
    &=\left|\{\ppath\in\PathsDk{k^*}:
              \hat C_{D,k^*}(\ppath)=\hat c\}\right|.
\end{aligned}
\end{equation}
The subscript $1$ records the step radius for increments in
$\{-1,0,1\}$.
\end{definition}

\paragraph{Relation to classical occupation-count enumeration}
Exact enumeration by state-occupation counts has classical
antecedents in the frequency-count theory of finite Markov chains
and in the joint occupation-total distribution of one-dimensional
random walks \cite{Good1961,Barnett1964}. Likewise, Whittle's
formula counts sequences after the complete matrix of transition
multiplicities has been fixed \cite{Whittle1955}. Our setting is
different: the cardinality tuple fixes only the vertex-occupation
profile. We therefore sum over all transition-count configurations
compatible with that profile and the prescribed endpoint. For the
trinomial step set $\{-1,0,1\}$, the one-dimensional support further
allows the directed matrix-tree cofactor
\cite{ChaikenKleitman1978} associated with each transition class to
collapse to an explicit product, yielding the finite-sum formula in
\Cref{thm:exact-trinomial-degeneracy}.

Let $a,b$ be the actual support endpoints of $\hat c$, as before.
Its occupied positions form the integer interval $[a,b]$, which
contains $0$ and $k^*$. To recover the multiplicity, we distinguish
upward and downward crossings of each adjacent pair of positions.
The endpoint determines their difference. Define
\begin{equation}
\label{eq:eta-cut-imbalance}
\begin{aligned}
    \eta_j(k^*)&\Def\mathbf 1_{\{0\le j\}}
                         -\mathbf 1_{\{k^*\le j\}},
                         \qquad j\in\Z,\\
    \eta_j^+&\Def\max\{\eta_j(k^*),0\},\qquad
    \eta_j^-\Def\max\{-\eta_j(k^*),0\}.
\end{aligned}
\end{equation}
For $m=(m_a,\dots,m_{b-1})\in\Z_+^{b-a}$, write
\begin{equation}
\label{eq:ud-m-parameterization}
    u_j(m)=m_j+\eta_j^+,\qquad
    d_j(m)=m_j+\eta_j^-,
    \qquad a\le j<b.
\end{equation}
Here $u_j(m)$ and $d_j(m)$ are the proposed numbers of crossings
$j\to j+1$ and $j+1\to j$, respectively. Set
\[
    u_{a-1}(m)=d_{a-1}(m)=u_b(m)=d_b(m)=0.
\]
The proposed number of stays at position $j$ is then
\begin{equation}
\label{eq:h-from-c-and-m}
    h_j(m)\Def c_j-\mathbf 1_{\{j=k^*\}}
                     -u_j(m)-d_{j-1}(m),
    \qquad a\le j\le b.
\end{equation}
The terminal visit is subtracted because it has no outgoing step.
The admissible crossing vectors form the finite set
\begin{equation}
\label{eq:M-feasible}
\mathcal M(\hat c;k^*)\Def
\left\{
    m\in\Z_+^{b-a}:
    \begin{aligned}
        h_j(m)&\ge0 &&(a\le j\le b),\\
        u_j(m)&>0   &&(a\le j<k^*),\\
        d_j(m)&>0   &&(k^*\le j<b)
    \end{aligned}
\right\}.
\end{equation}
The last two conditions ensure that every occupied position has a
directed route to the terminal position.

\begin{theorem}[Exact degeneracy of a trinomial cardinality tuple]
\label{thm:exact-trinomial-degeneracy}
For every $\hat c\in\mathcal C_{D,k^*}$,
\begin{equation}
\label{eq:degeneracy-explicit-finite-sum}
\boxed{
\deg_1(\hat c;k^*)
=
\left(\prod_{j=a}^{b}(c_j-1)!\right)
\sum_{m\in\mathcal M(\hat c;k^*)}
\frac{
    \displaystyle
    \left(\prod_{j=a}^{k^*-1}u_j(m)\right)
    \left(\prod_{j=k^*}^{b-1}d_j(m)\right)
}{
    \displaystyle
    \left(\prod_{j=a}^{b}h_j(m)!\right)
    \left(\prod_{j=a}^{b-1}u_j(m)!\,d_j(m)!\right)
}.
}
\end{equation}
Empty products equal $1$, $0!=1$, and $\Z_+^0$ contains the single
empty vector. In particular, the formula includes the all-stay path
with $a=b=k^*=0$, and gives degeneracy $1$ also when $D=0$.
\end{theorem}

The full proof is given in \ref{app:trinomial-degeneracy}, together
with the determinant representation
\eqref{eq:degeneracy-determinant-form} and its reduction to the
product in \eqref{eq:arborescence-product-collapse}.

\begin{corollary}[Recovery of every raw path]
\label{cor:degeneracy-recovers-raw-paths}
For the fixed depth and endpoint,
\begin{equation}
\label{eq:fiber-sum-raw-count}
    \sum_{\hat c\in\mathcal C_{D,k^*}}\deg_1(\hat c;k^*)
    =|\PathsDk{k^*}|
    =[z^{k^*}](z^{-1}+1+z)^D.
\end{equation}
Here $z$ is a formal variable, and $[z^{k^*}]$ extracts its indicated
coefficient from the Laurent polynomial.
\end{corollary}

\begin{corollary}[Exact aggregation by cardinality tuples]
\label{cor:degeneracy-weighted-aggregation}
With the position values $\vvalue_j$ from
\eqref{E:efficientaggregation}, the uniform mean of the path sum over
$\PathsDk{k^*}$ is
\begin{equation}
\label{eq:exact-average-with-degeneracy}
\frac{
    \displaystyle\sum_{\hat c\in\mathcal C_{D,k^*}}
    \deg_1(\hat c;k^*)
    \left(\sum_{j=k_-}^{k_+}c_j\vvalue_j\right)
}{
    \displaystyle\sum_{\hat c\in\mathcal C_{D,k^*}}
    \deg_1(\hat c;k^*)
}.
\end{equation}
For a fixed numerical aggregate $p$, its multiplicity $N_p$ is
\begin{equation}
\label{eq:Np-from-degeneracy}
    N_p=
    \sum_{\substack{
        \hat c\in\mathcal C_{D,k^*}\\
        \sum_{j=k_-}^{k_+}c_j\vvalue_j=p
    }}\deg_1(\hat c;k^*).
\end{equation}
\end{corollary}

The proofs of both corollaries are included in
\ref{app:trinomial-degeneracy}. These are uniform raw-path counts and averages; additional path probabilities are not included in the degeneracy.

\paragraph{Recovery of the three paths in \Cref{fig:invariance}}
For $D=4$ and $k^*=0$, take $c_0=4$, $c_1=1$, and all other
counts zero. Thus $\hat c=(0,0,4,1,0)$ on positions $-2,-1,0,1,2$,
and $a=0$, $b=1$. Since $\eta_0(0)=0$, feasibility requires
$m_0\ge1$, while $h_1(m)=1-m_0\ge0$. Hence $m_0=1$ is the only
feasible choice, giving
\[
    u_0=d_0=1,\qquad h_0=2,\qquad h_1=0.
\]
Substitution in \eqref{eq:degeneracy-explicit-finite-sum} yields
\[
    \deg_1(\hat c;0)
    =\frac{(4-1)!(1-1)!}{2!\,0!\,1!\,1!}=3.
\]
The three position sequences are
\[
    (0,1,0,0,0),\qquad
    (0,0,1,0,0),\qquad
    (0,0,0,1,0),
\]
and each has aggregate $4\cdot20+1\cdot22=102$. Degeneracy evaluation is an additional computation. The complexity
analysis that follows concerns the cardinality-tuple enumerator itself, rather than evaluation of the finite sum above.

\section{Computational Complexity Analysis}
\label{computation_complexity}

Having established how the cardinality tuples are generated, we now turn to the cost of carrying out the construction. The distributional viewpoint on recombining split trees in \cite{LebedevBanerjee2024}, where probability mass varies across terminal positions, motivates an analogous question for enumeration: how does the number of distinct objects to be generated vary with the endpoint? In our setting, the
relevant output size is $|\mathcal C_{D,k^*}|$, and its dependence on $k^*$ need not follow either the terminal probability distribution or the corresponding number of raw paths. The endpoint must therefore remain explicit in the complexity analysis. This distinction is already visible in the running example. At $D=7$, exhaustive raw-path enumeration together with \eqref{eq:dynamic-total-count-final} gives
\[
\begin{array}{c|rrrrrrrr}
    k^* & 0 & 1 & 2 & 3 & 4 & 5 & 6 & 7\\
    \hline
    |\PathsDk{k^*}|
        &393&357&266&161&77&28&7&1\\
    |\mathcal C_{7,k^*}|
        &59&71&76&68&47&23&7&1
\end{array}
\]
Thus the raw-path family is largest at $k^*=0$, whereas the
largest cardinality-tuple families at this depth occur at
$k^*=\pm2$. The central endpoint is therefore not
automatically the worst case for the tuple enumerator.
By \Cref{thm:reflection-symmetry}, it suffices to analyze
$0\le k^*\le D$; negative endpoints have the same tuple
counts by reflection. Before stating the complexity bound, we fix the
computational model. We use a unit-cost word-RAM model
\cite{Hagerup1998}, with word length
\[
    w_{\mathrm{RAM}}=\Theta(D+1)
\]
bits. This word length is sufficient for every quantity
manipulated by \Cref{alg:unique_path}. Position indices, visit
counts, stage indices, layer indices, and residual coordinates
require only $\Order(\log(D+2))$ bits, while the number of
emitted cardinality tuples is at most the number of depth-$D$
raw paths, namely $3^D$, so the output counter requires only
$\Order(D+1)$ bits. We count assignment, comparison,
addition or subtraction, division by two with floor or ceiling, and random access to a machine word as constant-time
operations. For the residual-composition loop in
\Cref{alg:unique_path}, we use the counter-array realization
described with the algorithm and maintain the remaining
unassigned mass incrementally. Thus each counter transition,
loop-bound update, and residual-mass update requires constant
word-RAM work; no repeated summation of previously chosen
coordinates is required. The resulting bounds measure word-RAM operation complexity. They are not bit-complexity bounds; a bit-complexity analysis would additionally charge for the implementation cost of operations on $w_{\mathrm{RAM}}$-bit words. The following theorem analyzes the cardinality-tuple enumerator exactly as it is stated in \Cref{alg:unique_path}. No constant-delay or Gray-code implementation is assumed.

\begin{theorem}[Complexity of the cardinality-tuple enumerator]
\label{thm:total-time}
Fix $D\in\Z_+$ and $k^*\in\{0,\dots,D\}$, and let $T$ be the final stage index defined in \eqref{eq:outer-stop}. Let $\beta\in\{1,2\}$ be the parity parameter from \eqref{switching_term}, and let
\[
    L
    \Def
    k_+-k_-+1
    =
    D+2-\beta
\]
be the length of a full cardinality tuple. Then the number of tuples generated by
\Cref{alg:unique_path} satisfies
\begin{equation}
\label{eq:complexity-cardinality-bound}
\binom{D}{k^*}
\le
|\mathcal C_{D,k^*}|
\le
\sum_{M=0}^{T}
\sum_{j=k^*+M}^{k_+}
\binom{D-M}{j}
\le
2^{D+1}
\left(1-2^{-(T+1)}\right).
\end{equation}
Under the computational model fixed above, the running time
of \Cref{alg:unique_path} is
\begin{equation}
\label{eq:complexity-endpoint-runtime}
    \mathrm{Time}_{\mathrm{enum}}(D,k^*)
    =
    \Theta\!\left(
        L\,|\mathcal C_{D,k^*}|
    \right).
\end{equation}
Consequently,
\begin{equation}
\label{eq:complexity-endpoint-upper}
    \mathrm{Time}_{\mathrm{enum}}(D,k^*)
    =
    \Order\!\left(
        L
        \sum_{M=0}^{T}
        \sum_{j=k^*+M}^{k_+}
        \binom{D-M}{j}
    \right)
    =
    \Order(D\,2^D)
\end{equation}
uniformly over $0\le k^*\le D$. Moreover, if $R_D\Def \max_{0\le k^*\le D}\mathrm{Time}_{\mathrm{enum}}(D,k^*),$ then, for $D\ge1$,
\begin{equation}
\label{eq:complexity-worst-case}
    R_D
    =
    \Omega\!\left(
        D\binom{D}{\lfloor D/2\rfloor}
    \right),
    \quad
    R_D
    =
    \Order(D\,2^D),
\end{equation}
and hence
\begin{equation}
\label{eq:complexity-exponential-base}
    \lim_{D\to\infty} R_D^{1/D}=2.
\end{equation}
By \Cref{thm:reflection-symmetry}, the same asymptotic worst-case bound holds over all endpoints $|k^*|\le D$: for $k^*<0$, one runs the construction for $|k^*|$ and reflects each emitted tuple, which requires at most $\Order(L)$ additional work per tuple and therefore does not change the asymptotic bound in \eqref{eq:complexity-endpoint-runtime}. The streaming working memory of the algorithm is $\Order(L)=\Order(D+1)$ words, excluding storage of previously
emitted tuples.
\end{theorem}

\begin{proof}
We begin with the positive stage
$\mathcal B_0=\mathcal C^+_{D,k^*}$.
The exact positive count is already given by
\eqref{eq:positive-count-layers}:
\[
    |\mathcal C^+_{D,k^*}|
    =
    1+
    \sum_{i=1}^{D-k^*}
    \binom{i+\ell(i)-1}{\ell(i)-1},
\]
where $\ell(i)$ is the receiving-slot function defined in
\eqref{eq:available-slots}. Because $D+k^*$ and $D-k^*$ have the same parity, \eqref{switching_term} and the definition of $T$ give
\[
    D-k^*=2T+\beta-1,
    \quad
    D+k^*=2k_++\beta-1.
\]
Adding and subtracting these identities yields
\begin{equation}
\label{eq:complexity-D-k-relations}
    D=k_++T+\beta-1,
    \quad
    k^*=k_+-T.
\end{equation}
We now pair consecutive terms of the existing mass-layer
sum. If $\beta=1$, then $D-k^*=2T$. For $1\le r\le T$,
\[
\begin{aligned}
    \ell(2r-1)&=k_+-r+1,\\
    \ell(2r)&=k_+-r,
\end{aligned}
\]
and the corresponding two layer counts are $\binom{k_++r-1}{k_+-r}$ and $\binom{k_++r-1}{k_+-r-1}.$ By Pascal's identity \cite[Sec.~5.1]{GKP}, their sum is $\binom{k_++r}{k_+-r}.$
Including the seed as the $r=0$ term therefore gives
\[
    |\mathcal C^+_{D,k^*}|
    =
    \sum_{r=0}^{T}
    \binom{k_++r}{k_+-r}
    \qquad(\beta=1).
\]
If $\beta=2$, then $D-k^*=2T+1$. The seed and layer $i=1$
satisfy $1+\binom{k_+}{k_+-1}=\binom{k_++1}{k_+}.$ For $1\le r\le T$,
\[
\begin{aligned}
    \ell(2r)&=k_+-r+1,\\
    \ell(2r+1)&=k_+-r,
\end{aligned}
\]
so Pascal's identity gives
\[
\begin{aligned}
\binom{k_++r}{k_+-r}+\binom{k_++r}{k_+-r-1}=\binom{k_++r+1}{k_+-r}.
\end{aligned}
\]
Hence
\[
    |\mathcal C^+_{D,k^*}|
    =
    \sum_{r=0}^{T}
    \binom{k_++r+1}{k_+-r}
    \qquad(\beta=2).
\]
Both parity cases are therefore summarized by
\begin{equation}
\label{eq:complexity-paired-positive}
    |\mathcal C^+_{D,k^*}|
    =
    \sum_{r=0}^{T}
    \binom{k_++r+\beta-1}{k_+-r}.
\end{equation}
This is only a pairing of the mass layers already present in
\eqref{eq:positive-count-layers}; it introduces no additional
classification of the tuples. For the $r$th term in
\eqref{eq:complexity-paired-positive},
\eqref{eq:complexity-D-k-relations} gives $k_++r+\beta-1\le D.
$ The lower indices $k_+-r$ run exactly through $k_+,k_+-1,\dots,k^*.$ For $0\le j\le n\le D$, $\binom{n}{j}\le\binom{D}{j},$ because every $j$-subset of an $n$-element set is also a $j$-subset of a $D$-element set. Therefore
\begin{equation}
\label{eq:complexity-positive-upper}
    |\mathcal C^+_{D,k^*}|
    \le
    \sum_{j=k^*}^{k_+}\binom{D}{j}
    \le
    2^D.
\end{equation}
At $r=T$, the identities in \eqref{eq:complexity-D-k-relations} give $\binom{k_++T+\beta-1}{k_+-T}=\binom{D}{k^*},$
so
\begin{equation}
\label{eq:complexity-positive-lower}
    |\mathcal C^+_{D,k^*}|
    \ge
    \binom{D}{k^*}.
\end{equation}
We now apply the already-established stage reduction.
By \Cref{prop:stage-positive-reduction}, $|\mathcal B_M|=
|\mathcal C^+_{D-M,k^*+M}|$. The associated positive problem has depth $D-M$, endpoint $k^*+M$, and the same upper feasible position $k_+$. Applying \eqref{eq:complexity-positive-upper} to this associated problem gives
\[
    |\mathcal B_M|
    \le
    \sum_{j=k^*+M}^{k_+}
    \binom{D-M}{j}
    \le
    2^{D-M}.
\]
Since the stages form the disjoint partition
\eqref{eq:stage-partition},
\[
\begin{aligned}
    |\mathcal C_{D,k^*}|
    &=
    \sum_{M=0}^{T}|\mathcal B_M|
    \le
    \sum_{M=0}^{T}
    \sum_{j=k^*+M}^{k_+}
    \binom{D-M}{j}\\
    &\le
    \sum_{M=0}^{T}2^{D-M}
    =
    2^{D+1}
    \left(1-2^{-(T+1)}\right).
\end{aligned}
\]
The lower bound in
\eqref{eq:complexity-cardinality-bound} follows from
$\mathcal B_0=\mathcal C^+_{D,k^*}$ and
\eqref{eq:complexity-positive-lower}. This proves the
cardinality estimate. We next turn to the actual loops of
\Cref{alg:unique_path}. From \eqref{E:Kbounds}, $k_-=k^*-k_+,$
and therefore
\[
\begin{aligned}
    L
    =k_+-k_-+1
    =2k_+-k^*+1
    =D+2-\beta,
\end{aligned}
\]
where the last equality is \eqref{switching_term}. Thus
$D\le L\le D+1$ whenever $D\ge1$, while $L=1$ when $D=0$. Let
$N\Def|\mathcal C_{D,k^*}|.$ There are $T+1$ stage seeds and therefore $N-(T+1)$ non-seed outputs. For every non-seed output, \Cref{alg:unique_path} first zeros the complete working tuple, which is a loop of exactly $L$ assignments. It then writes a baseline on a subinterval of the same tuple, adds at most $\ell(i)\le L$ residual entries, and emits the completed tuple. These operations require $\Order(L)$ work per emitted non-seed tuple. It remains to verify that traversing the residual compositions does not introduce a larger hidden cost. Fix a nonempty layer with residual mass $i\ge1$ and $\ell\ge1$ receiving slots, and write
\[
    Q_{i,\ell}
    \Def
    \binom{i+\ell-1}{\ell-1}
\]
for the number of complete residual compositions in that layer. Under the counter-array realization specified above, the remaining unassigned mass is maintained incrementally, so each partial coordinate choice requires only constant word-RAM
work. After $q$ coordinates have been chosen, where
$0\le q\le\ell-1$, the number of possible partial assignments
is $\binom{i+q}{q},$ because such assignments are precisely nonnegative $q$-tuples having total at most $i$. Hence the total number of partial counter states visited by the nested traversal is
\[
\begin{aligned}
    \sum_{q=0}^{\ell-1}\binom{i+q}{q}
    =
    \binom{i+\ell}{\ell-1}
    =
    \frac{i+\ell}{i+1}
    \binom{i+\ell-1}{\ell-1}
    \le
    \ell\,Q_{i,\ell}.
\end{aligned}
\]
Thus the entire residual-composition traversal for the layer
requires $\Order(\ell Q_{i,\ell})\subseteq\Order(LQ_{i,\ell})$
word-RAM operations. Summed over all nonempty layers, this is
therefore absorbed into $\Order(LN)$. No constant-delay or
Gray-code composition generator is required. The first stage seed is constructed on the same $L$ coordinates. Each later stage seed is obtained by \Cref{alg:init_neg_array}, whose zero-padding and copy loops also touch at most a constant number of arrays of length at most $L$. Thus the seed and reseeding work is $\Order(L)$ per stage. For completeness, the number of executions of the mass-layer loop is exactly
\[
\begin{aligned}
H=\sum_{M=0}^{T}(D-k^*-2M)=(T+1)(D-k^*-T).
\end{aligned}
\]
Because $D-k^*=2T+\beta-1$,
\[
    H=(T+1)(T+\beta-1).
\]
For $D\ge1$, $T+\beta-1\le D\le L,$ and the existence of one seed in every stage gives $T+1\le N$. Hence $H\le LN.$ The layer-loop control overhead is therefore no larger in order than the tuple-construction and composition-traversal work.
It follows that the total running time is
\[
    \Order(LN)
\]
Conversely, each non-seed output performs a complete
length-$L$ zeroing pass, and each stage seed requires a
length-$L$ initialization or reseeding pass. Since every
output is either a stage seed or a non-seed tuple, the
algorithm performs $\Omega(LN)$ scalar work. This proves
\eqref{eq:complexity-endpoint-runtime}. Substituting the
cardinality bound already proved gives
\eqref{eq:complexity-endpoint-upper}. For the worst-case lower bound, choose $k^*=\left\lfloor\frac D2\right\rfloor.$
Then \eqref{eq:complexity-positive-lower} and $L\ge D$ give
\[
    R_D
    =
    \Omega\!\left(
        D\binom{D}{\lfloor D/2\rfloor}
    \right).
\]
The largest coefficient in row $D$ of Pascal's triangle is
at least the average of that row, so
\[
    \binom{D}{\lfloor D/2\rfloor}
    \ge
    \frac{2^D}{D+1}.
\]
Consequently $R_D=\Omega(2^D)$. The uniform upper bound gives $R_D=\Order(D2^D)$. Taking $D$th roots between these two bounds yields $\lim_{D\to\infty}R_D^{1/D}=2.$ Finally, the recursive baseline established in
\eqref{E:memoisedtime} requires $\Theta(3^D)$ traversal work before output materialization, while \eqref{E:naivetime} gives its complete fixed-endpoint runtime when retained paths are copied. Thus exhaustive raw-path traversal has exponential growth base $3$. By contrast, the direct cardinality-tuple
enumerator has the uniform bound $\Order(D2^D)$ and worst-case exponential growth base $2$. At the level of these native enumeration workloads, their costs therefore differ by at least
\[
    \Omega\!\left(
        \frac{(3/2)^D}{D}
    \right)
\]
up to constant factors. This comparison reflects the reduction in enumeration space: the raw-path algorithms enumerate every element of $\PathsDk{k^*}$, whereas the direct algorithm enumerates one element of $\mathcal C_{D,k^*}$ for each path-equivalence class. It should therefore be interpreted as a comparison
between exhaustive raw-history enumeration and direct quotient enumeration, rather than as a same-output implementation speedup. The working state of \Cref{alg:unique_path} consists of a constant number of arrays whose lengths are at most $L$, together with scalar indices and counters. Its streaming working memory is therefore $\Order(L)=\Order(D+1)$ words. Retaining all emitted tuples requires, in addition, the corresponding output storage.
\end{proof}

\section{Computational Verification}
\label{sec:computational-verification}

The preceding analysis provides the exact counting formula and the
output-sensitive work decomposition of the direct enumeration procedure.
We next examine whether these structural differences are reflected by the
implementations themselves. The purpose of these experiments is not to
replace the combinatorial arguments, but to verify the generated families,
measure the reduction in enumerated objects, and compare the computational
behavior of the corresponding enumeration procedures.

\subsection{Implementation and Verification}

In the previous sections we specified three core algorithms: recursive depth-first enumeration of raw paths (\Cref{alg:gencomb_dfs}), recursion-free generation of nonnegative paths (\Cref{alg:recursivelessgen}), and mass-layer counting of cardinality tuples (\Cref{alg:unique_path}). Additionally, the excursion-flip construction extends the nonnegative family to the full endpoint-conditioned family (\Cref{prop:flip-representation}). To assess how the predicted advantage carries over to the implementations, we compare these procedures. This comparison shows how the work changes when recursion is removed, duplicates are filtered after traversal, or tuples are generated directly. The stack-free traversal is a timing baseline, distinct from the paper’s recursion-free path generator. Likewise, post-hoc hash deduplication is distinct from \Cref{alg:gencomb_hashing}, which uses a position-value lookup table while retaining the raw paths. Correctness is checked separately by comparing the direct output with exhaustive enumeration; the timing results show whether the observed scaling is consistent with the derived bound.

The direct generator and the exhaustive verifier are separate computational paths. The direct procedure constructs the stage seeds and applies the prescribed mass-shift, carry, and reseeding operations; it does not receive raw paths or profiles derived from them. Raw-path enumeration is used only as an independent correctness oracle and as a timing baseline. For every depth \(0\le D\le10\) and nonnegative endpoint \(0\le k^*\le D\), the direct output agrees with the exhaustively obtained set of cardinality tuples, with no duplicate
direct outputs and with counts matching the closed formula. All \(66\) tested depth--endpoint cases pass. A separate depth-four audit checks the entire tree: its \(81\) raw paths produce \(56\) endpoint-conditioned cardinality classes, exactly matching the direct generator. For runtime measurements, all methods use the common endpoint \(k^*=0\). Each plotted value is the median of five trials; the timed region includes enumeration and a streaming checksum, but excludes process startup, file output, and figure production.
\begin{figure}[H]
    \centering
    \includegraphics[width=\textwidth]
        {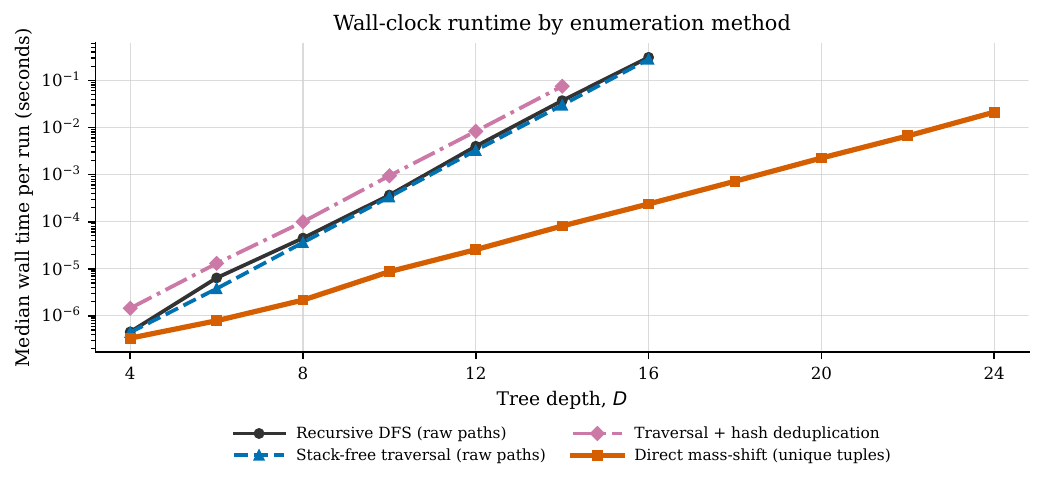}
    \caption{Median runtime of the four enumeration implementations at
    the common endpoint \(k^*=0\).}
    \label{fig:algorithm-runtime-comparison}
\end{figure}
Figure~\ref{fig:algorithm-runtime-comparison} shows a widening runtime
separation as depth increases. At \(D=14\), the direct procedure was
approximately \(459\), \(369\), and \(929\) times faster than recursive
traversal, stack-free traversal, and traversal with post-hoc hash
deduplication, respectively. At \(D=16\), it was approximately \(1298\)
and \(1186\) times faster than recursive and stack-free traversal;
the hash-deduplication run is not included at that depth. These
machine-dependent timings illustrate the cost of traversing raw
histories before removing duplicate profiles.

\subsection{Enumeration Growth and Computational Scaling}
To compare output growth, define $C_D^{\max}\Def\max_{|k|\le D}|\mathcal C_{D,k}|,$ and let \(k_D\) be an endpoint attaining this maximum. The left panel
of Figure~\ref{fig:complexity-validation} compares the full raw search
space \(3^D\), the maximum endpoint-conditioned output count
\(C_D^{\max}\), and the uniform output-count ceiling \(2^{D+1}\)
implied by \Cref{thm:total-time}. The latter is a proved count bound,
not a fitted curve or a runtime estimate. For the work measurements, let
\[
    m_D
    \Def
    \left\lfloor\frac{D+|k_D|}{2}\right\rfloor+1
\]
denote the width of the local tuple stored by the implementation, and
let \(W_D\) be its measured number of tuple-entry accesses. The
right panel separates dense emission, successor scans and updates,
and stage-seed construction and reseeding. Dense emission contributes
exactly \(m_D C_D^{\max}\) accesses. This local storage width is
distinct from \(L\), the full-tuple length used in the word-RAM theorem.
The measured ratio \(W_D/(m_D C_D^{\max})\) decreases from \(5.67\) at
\(D=2\) to \(1.76\) at \(D=24\).

\begin{figure}[H]
    \centering
    \includegraphics[width=\textwidth]{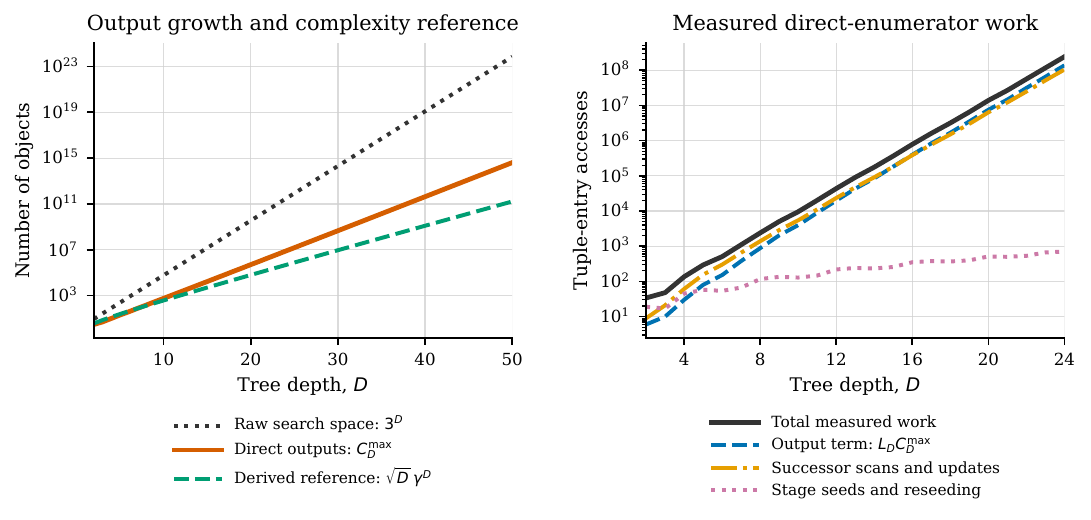}
    \caption{Exact output counts and instrumented work for the direct
    enumerator.}
    \label{fig:complexity-validation}
\end{figure}
The measurements are consistent with the output-sensitive complexity proved in \Cref{thm:total-time}, which gives $\mathrm{Time}_{\mathrm{enum}}(D,k^*) =\Theta\!\left(L|\mathcal C_{D,k^*}|\right)$ and the uniform worst-case bound \(O(D2^D)\), with exponential growth base \(2\). The finite-depth work ratios illustrate the implementation's output-dominated behavior although they are not themselves a proof of an asymptotic bound. The count data also show why the endpoint must remain explicit: the tuple-count-maximizing endpoint is not always \(k=0\), with the first strict departure occurring at \(D=5\). Taken together, the experiments verify the generated families, demonstrate the widening computational advantage of direct enumeration over all three raw-path alternatives, and confirm that the implementation is output-sensitive in the number and width of the tuples it emits. Having completed this computational analysis, we now turn to a different consequence of the same cardinality representation: its specialization to classical Motzkin paths and the generating functions associated with their occupation profiles.

\section[Motzkin Specialization and Occupation-Profile Generating Functions]{Motzkin Specialization and\\Occupation-Profile Generating Functions}
\label{ss:motzkin-specialization}

The nonnegative return case $k^*=0$ identifies the raw paths in
$\PosPathsDk{0}$ with a classical lattice-path family: Motzkin
paths. Their cardinality tuples retain the full profile of visits
to each height, while the degeneracy formula restores the number
of raw paths represented by each profile. This specialization
therefore connects the direct tuple construction with classical
first-return decompositions and generating functions
\cite{BanderierFlajolet2002,FlajoletSedgewick}, and subsequently
with the weighted-path continued-fraction construction of
\cite{Flajolet1980}. The latter provides a generating-series
interpretation of the multiplicities computed explicitly in
\Cref{thm:exact-trinomial-degeneracy}.

\subsection{Nonnegative return paths}

Set $k^*=0$. The relevant vertex-path family is the already defined
$\PosPathsDk{0}$. For $\ppath\in\PosPathsDk{0}$ with
$\pi(\ppath)=(k_0,\dots,k_D)$, encode each increment
$s_d=k_d-k_{d-1}$ by
\[
    +1\longleftrightarrow\mathrm U,\qquad
     0\longleftrightarrow\mathrm H,\qquad
    -1\longleftrightarrow\mathrm D.
\]
Let $\mathfrak M_D$ denote the resulting set of step words, so
$\mathfrak M_D$ is the word representation of
$\PosPathsDk{0}$; concatenation below always refers to words rather
than vertex sequences. Define the corresponding Motzkin numbers
and their ordinary generating function by
\begin{equation}
\label{eq:motzkin-tree-family}
    M_D\Def|\mathfrak M_D|=|\PosPathsDk{0}|,
    \qquad
    M(t)\Def\sum_{D\ge0}M_Dt^D.
\end{equation}
Here $M_D$ and $M(t)$ denote Motzkin counting quantities; the outer
stage index $M$ from the enumeration construction retains its
earlier meaning. At $D=0$, the sole vertex path is $((0,0))$, its
step word is empty, and it still contributes one root visit.

\begin{theorem}[Motzkin specialization]
\label{thm:motzkin-generation}
The increment encoding is a bijection between $\PosPathsDk{0}$ and
the Motzkin paths of length $D$. Their counts satisfy
\begin{equation}
\label{eq:motzkin-recurrence}
    M_0=1,\qquad
    M_D=M_{D-1}+\sum_{r=0}^{D-2}M_rM_{D-2-r},
    \qquad D\ge1,
\end{equation}
where an empty sum is zero. Consequently,
\begin{equation}
\label{eq:motzkin-functional-equation}
    M(t)=1+tM(t)+t^2M(t)^2,
\end{equation}
and the solution with constant term $1$ is
\begin{equation}
\label{eq:motzkin-generating-function}
    M(t)=\frac{1-t-\sqrt{1-2t-3t^2}}{2t^2}.
\end{equation}
For every $D\ge0$,
\begin{equation}
\label{eq:motzkin-explicit-count}
    M_D=\sum_{r=0}^{\lfloor D/2\rfloor}
        \binom{D}{2r}\frac{1}{r+1}\binom{2r}{r}.
\end{equation}
\end{theorem}

To refine the count by the number of horizontal steps, introduce a
formal marker $q$ and define
\begin{equation}
\label{eq:motzkin-horizontal-refinement}
    M(t,q)
    \Def
    \sum_{D\ge0}t^D
    \sum_{\ppath\in\PosPathsDk{0}}
        q^{j_0(\ppath)},
\end{equation}
where $j_0(\ppath)$ is the number of zero increments, equivalently
the number of horizontal steps, from \eqref{E:algebra}. Setting
$q=1$ removes this refinement and recovers $M(t)$.

\begin{corollary}[Horizontal-step refinement and Dyck specialization]
\label{cor:dyck-specialization}
The horizontal-step refinement satisfies
\begin{equation}
\label{eq:motzkin-bivariate-equation}
    M(t,q)=1+qtM(t,q)+t^2M(t,q)^2,
\end{equation}
with solution
\begin{equation}
\label{eq:motzkin-bivariate-solution}
    M(t,q)=\frac{1-qt-\sqrt{(1-qt)^2-4t^2}}{2t^2}.
\end{equation}
Setting $q=0$ retains exactly the paths with no horizontal steps:
\begin{equation}
\label{eq:dyck-generating-function}
    M(t,0)=\frac{1-\sqrt{1-4t^2}}{2t^2}
          =\sum_{r\ge0}\frac{1}{r+1}\binom{2r}{r}t^{2r}.
\end{equation}
\end{corollary}

All square roots above mean the formal square root with constant
term $1$; the displayed quotients are interpreted as power series.
The bijection, disjoint first-return construction, and both counting
formulas are proved in \ref{app:motzkin-integration}. These recurrences
describe the path family; they do not replace the recursion-free
tuple enumerator.

\subsection{The occupation-profile continued fraction}

The scalar marker $q$ records only the total number of horizontal
steps. To recover the finer information carried by a cardinality
tuple, we instead mark visits to each height separately. Let
$\mathbf y=(y_0,y_1,\dots)$ be commuting algebraically independent
variables. For $h\in\Z_+$, define
\begin{equation}
\label{eq:Eh-definition}
    E_h(\mathbf y;t)\Def
    \sum_{D\ge0}t^D
    \sum_{\substack{\ppath\in\PosPathsDk{0}\\
                    \pi(\ppath)=(k_0,\dots,k_D)}}
         \prod_{d=1}^{D}y_{h+k_d}.
\end{equation}
This records excursions translated to height $h$, weighting each
arrival but excluding their initial visit. The empty word contributes
$1$. Restoring the root visit defines
\begin{equation}
\label{eq:Fplus-definition}
    F^+(\mathbf y;t)\Def y_0E_0(\mathbf y;t).
\end{equation}
The series are formal in $t$, with polynomial coefficients in the
height variables. In sums over $D$, the path family and cardinality
map are understood at that depth. For $\hat c\in\mathcal C^+_{D,0}$,
set $c_h=0$ for $h>k_+$ when writing products over all heights; this
only extends the finite tuple by zeros.

\begin{theorem}[Occupation-profile $J$-fraction]
\label{thm:occupation-profile-j-fraction}
For every $h\ge0$,
\begin{equation}
\label{eq:Eh-functional-equation}
\begin{aligned}
    E_h(\mathbf y;t)
    &=1+ty_hE_h(\mathbf y;t)\\
    &\quad+t^2y_hy_{h+1}E_{h+1}(\mathbf y;t)E_h(\mathbf y;t),
\end{aligned}
\end{equation}
so
\begin{equation}
\label{eq:Eh-resolved}
    E_h(\mathbf y;t)
    =\frac{1}{1-ty_h-t^2y_hy_{h+1}E_{h+1}(\mathbf y;t)}.
\end{equation}
It follows that
\begin{equation}
\label{eq:occupation-profile-j-fraction}
    F^+(\mathbf y;t)
    =\cfrac{y_0}{1-ty_0-
       \cfrac{t^2y_0y_1}{1-ty_1-
       \cfrac{t^2y_1y_2}{1-ty_2-
       \cfrac{t^2y_2y_3}{1-ty_3-\ddots}}}}.
\end{equation}
The infinite fraction is interpreted coefficientwise through its
finite height truncations. Its expansion by cardinality tuples is
\begin{equation}
\label{eq:profile-series-degeneracy-expansion}
    F^+(\mathbf y;t)
    =\sum_{D\ge0}t^D
       \sum_{\hat c\in\mathcal C^+_{D,0}}
          \deg_1(\hat c;0)\prod_{h\ge0}y_h^{c_h}.
\end{equation}
In particular, for every $\hat c\in\mathcal C^+_{D,0}$,
\begin{equation}
\label{eq:profile-coefficient-degeneracy}
    \left[t^D\prod_{h\ge0}y_h^{c_h}\right]
       F^+(\mathbf y;t)=\deg_1(\hat c;0),
\end{equation}
where brackets extract the coefficient of the indicated monomial.
\end{theorem}

This is the weighted-Motzkin continued-fraction construction
\cite{Flajolet1980,Viennot1985} with weights chosen to record vertex
visits, including the root. Its coefficients are therefore the same
multiplicities evaluated explicitly in
\Cref{thm:exact-trinomial-degeneracy}. Related weighted-path
constructions underlie the connections with Jacobi parameters and
cumulants in \cite{Lehner2003} and the Motzkin-polynomial recurrences
in \cite{OsteVanDerJeugt2015}.

\begin{corollary}[Recovery of the Motzkin numbers from cardinality tuples]
\label{cor:motzkin-from-degeneracies}
Setting every $y_h=1$ gives $F^+(\mathbf 1;t)=M(t)$ and
\begin{equation}
\label{eq:motzkin-fiber-sum}
    M_D=\sum_{\hat c\in\mathcal C^+_{D,0}}\deg_1(\hat c;0),
    \qquad D\ge0.
\end{equation}
\end{corollary}

Thus the tuple count and the Motzkin count describe different stages
of the same calculation: the former counts each admissible profile
once, while the latter weights it by the number of paths it represents.
For example,
\[
\begin{aligned}
    [t^4]F^+(\mathbf y;t)
    &=y_0^5+3y_0^4y_1+3y_0^3y_1^2\\
    &\quad+y_0^2y_1^3+y_0^2y_1^2y_2.
\end{aligned}
\]
There are five profiles and nine paths. The coefficient of
$y_0^4y_1$ recovers the three-path example in
\Cref{subsec:tuple-degeneracy}. The full first-return, formal-series,
and coefficient-extraction proofs are given in
\ref{app:motzkin-integration}.

\section{Conclusions}

At its core, this work replaces the enumeration of raw path histories by direct enumeration of the quotient induced by their level-visit structure. Each path-equivalence class is represented by a unique admissible cardinality tuple, and the stage--mass-layer construction generates every such tuple exactly once. The resulting enumerator runs in $\Theta\!\left(L|\mathcal C_{D,k^*}|\right)$ time, with a uniform $\mathcal O(D2^D)$ bound and worst-case exponential growth base $2$, compared with base $3$ for exhaustive raw-path traversal. The cardinality representation also preserves the information needed for exact position-dependent aggregation, while the degeneracy formula recovers the number of raw paths represented by each tuple. The framework therefore separates two tasks that need not be performed simultaneously: unique structural profiles can be generated directly, and raw-path multiplicities can be recovered only when they are required. Several natural questions remain beyond the present construction. First, the weak-composition structure suggests that the tuples may admit more economical successor rules than the explicit reconstruction used here. Second, the stage and support arguments raise the question of which parts of the construction persist when the trinomial step set is replaced by a higher-branching recombining structure. We briefly describe these two directions below.

\subsection{Further Work}

\subsubsection{Gray-Code Optimization}

The weak-composition bijections suggest a minimal-change (Gray-code) traversal that maintains the current tuple in place \cite{ruskey2005prefix}. Such a traversal could reduce reconstruction
work, although preserving the present cardinality order would require
additional analysis of successor operations and transitions between
layers and stages. Any improvement would apply to in-place or
change-only generation: explicitly materializing a full tuple still
requires $\Theta(L)$ output work, where $L=k_+-k_-+1$, consistent with \Cref{thm:total-time}.

\subsubsection{Extension to Higher-Branching Recombining Trees}
\label{ss:extension-n-nomial}

The Gray-code question concerns how efficiently the present construction generates its tuples. A complementary direction is to determine how much of the same stage-and-support structure persists when the trinomial step set is enlarged. Extending the framework to \((2n-1)\)-nomial trees for \(n\ge3\) calls for a characterization of the corresponding support patterns and admissibility conditions. More broadly, it remains open whether comparable case-based constructions can be developed for general recombining trees, including even-branching ones.

\subsection{Applications}
\label{sec:applications}

The results also point to settings where the occupation-profile representation can be useful. When a trajectory-level quantity is determined by its level-visit counts, a cardinality tuple provides an exact summary for aggregating that quantity, while the degeneracy formula recovers the number of raw histories it represents. Turning this summary into a valuation or distribution requires the model-specific payoff or cost rule and, for probabilistic evaluations, the appropriate path weights. If the quantity depends on information the profile does not encode, the state representation must be augmented. With these conditions in mind, we describe several potential applications:

\paragraph{(A) Option pricing on recombining trees}
The constructive enumeration supports exact valuation of
path-dependent claims (e.g., Asian, barrier, lookback) whose
payoffs are determined by occupation profiles (cardinality tuples), provided the appropriate class probability weights are supplied. Closed-form counts can guide stratified/importance sampling aimed at variance reduction; early exercise (American-style) requires integration with an appropriate augmented-state dynamic program. The enumeration may also support risk analysis in conjunction with methods such as FSG (forward-shooting grid).

\paragraph{(B) Discrete-event control and scheduling}
Event sequences in queues or inventory systems can be modeled
as trinomial walks; the cardinality tuple records level visits
(e.g., buffer or stock levels). This supports cost aggregation
along trajectories and worst-case evaluations under resource
constraints, as well as reachability distributions and
risk-sensitive evaluations when the appropriate probability
weights are available.

\paragraph{(C) Planning and model-based reinforcement learning}
Finite-horizon profile enumeration on the recombining tree avoids recursion and duplicates, while occupation profiles serve as compact trajectory features for value expansion or policy improvement. For long horizons, the same structure supports density-aware pruning or stratified sampling, with error control requiring separate analysis. In moderate-depth regimes requiring auditability and tight error budgets, the proposed enumeration can be preferable to naive recursion or unconstrained Monte Carlo; for very deep horizons, it remains a principled backbone for hybrid sampling schemes.

\bibliographystyle{elsarticle-num}
\bibliography{references}
\appendix

\section{Raw Path Generation using DFS}
The recursive baseline below enumerates raw vertex paths at a
fixed endpoint. The target depth $D$ and endpoint $k^*$ remain
fixed; $d$ and $k$ denote the current depth and position.
The mutable path $\ppath$ and the output list
\texttt{All Paths} are passed by reference. Each stored output
is a copy, so subsequent backtracking does not change it.

\begin{algorithm}[H]
\small
\caption{Generate raw paths using recursive depth-first search}
\label{alg:gencomb_dfs}
\begin{algorithmic}[1]
\footnotesize
\REQUIRE $D\in\Z_+$ and $k^*\in\Z$ with $|k^*|\le D$.
\ENSURE Return every path in $\PathsDk{k^*}$ exactly once.
\FUNCTION{GenerateCombinationsViaRecursiveDFS}{$D,k^*$}
    \STATE $\texttt{All Paths}\leftarrow$ empty list
    \STATE $\ppath\leftarrow((0,0))$
    \STATE \CALL{DFS}{$0,0,D,k^*,\ppath,\texttt{All Paths}$}
    \STATE \textbf{return} \texttt{All Paths}
\ENDFUNCTION
\FUNCTION{DFS}{$d,k,D,k^*,\ppath,\texttt{All Paths}$}
    \IF{$d=D$}
        \IF{$k=k^*$}
            \STATE Append a copy of $\ppath$ to \texttt{All Paths}
        \ENDIF
        \STATE \textbf{return}
    \ENDIF

    \STATE Append $(k-1,d+1)$ to $\ppath$
    \STATE \CALL{DFS}{$d+1,k-1,D,k^*,\ppath,\texttt{All Paths}$}
    \STATE Remove the last vertex of $\ppath$

    \STATE Append $(k,d+1)$ to $\ppath$
    \STATE \CALL{DFS}{$d+1,k,D,k^*,\ppath,\texttt{All Paths}$}
    \STATE Remove the last vertex of $\ppath$

    \STATE Append $(k+1,d+1)$ to $\ppath$
    \STATE \CALL{DFS}{$d+1,k+1,D,k^*,\ppath,\texttt{All Paths}$}
    \STATE Remove the last vertex of $\ppath$
\ENDFUNCTION
\end{algorithmic}
\end{algorithm}

\paragraph{Correctness}
On entry to a recursive call, $\ppath$ contains exactly
$d+1$ vertices, starts at $(0,0)$, and ends at $(k,d)$.
This holds initially. Each branch appends the corresponding
permitted child at depth $d+1$, and removing that vertex
after the call restores the original prefix.
Every increment sequence in $\{-1,0,1\}^D$ determines
exactly one chain of branch choices. Conversely, successive
position differences recover those choices from the vertex
path. Thus distinct branch sequences cannot produce the
same raw path. At depth $D$, the endpoint test keeps exactly
the paths in $\PathsDk{k^*}$. Depth increases at every
recursive call and is bounded by $D$, so the procedure
terminates. When $D=0$, it returns the sole root path
$((0,0))$.

Different returned paths may have the same cardinality tuple;
this is intentional. To validate the direct tuple enumerator,
apply the cardinality map to these raw paths and compare its
image with the direct output. The multiplicity of each tuple
in the raw list is retained when checking degeneracies.
This raw-path baseline is not a step of the direct
mass-shifting enumeration.

\section{Hash-Table Raw-Path Baseline}
The hash table \texttt{input\_map} contains the precomputed
value $\vvalue_k$ at each integer position $k\in[-D,D]$.
A table with missing positions would describe a restricted
walk problem rather than the full trinomial tree.
The reusable array \texttt{buf} has indices $0,\dots,D$;
entry \texttt{buf}[$d$] is a record with fields
\texttt{key} and \texttt{value} for the position and its
stored value at depth $d$. Only the prefix through the
current depth is read. The table, buffer, and output list
are passed by reference to the helper routine.

\begin{algorithm}[H]
\small
\caption{Generate raw paths using a position-value hash table}
\label{alg:gencomb_hashing}
\begin{algorithmic}[1]
\footnotesize
\REQUIRE $D\in\Z_+$, $k^*\in\Z$ with $|k^*|\le D$;
\texttt{input\_map} contains every position in $[-D,D]$.
\ENSURE Return every path in $\PathsDk{k^*}$ exactly once.
\FUNCTION{GenerateCombinations}{$D,k^*,\texttt{input\_map}$}
    \STATE $\texttt{All Paths}\leftarrow$ empty list
    \STATE Create \texttt{buf} with indices $0,\dots,D$
    \STATE $\texttt{buf}[0]\leftarrow
        (\texttt{key}:0,\;\texttt{value}:\texttt{input\_map}[0])$
    \STATE \CALL{FillBuffer}{$0,D,k^*,\texttt{input\_map},\texttt{buf},\texttt{All Paths}$}
    \STATE \textbf{return} \texttt{All Paths}
\ENDFUNCTION
\FUNCTION{FillBuffer}{$d,D,k^*,\texttt{input\_map},\texttt{buf},\texttt{All Paths}$}
    \IF{$d=D$}
        \IF{$\texttt{buf}[D].\texttt{key}=k^*$}
            \STATE Append a copy of
                $\bigl((\texttt{buf}[j].\texttt{key},j)\bigr)_{j=0}^{D}$
                to \texttt{All Paths}
        \ENDIF
        \STATE \textbf{return}
    \ENDIF
    \STATE $\texttt{prev\_key}\leftarrow\texttt{buf}[d].\texttt{key}$
    \FOR{\texttt{offset} in $(-1,0,1)$, in that order}
        \STATE $\texttt{next\_key}\leftarrow
            \texttt{prev\_key}+\texttt{offset}$
        \IF{\texttt{next\_key} is a key in \texttt{input\_map}}
            \STATE $\texttt{buf}[d+1]\leftarrow
                (\texttt{key}:\texttt{next\_key},\;
                 \texttt{value}:\texttt{input\_map}[\texttt{next\_key}])$
            \STATE \CALL{FillBuffer}{$d+1,D,k^*,\texttt{input\_map},\texttt{buf},\texttt{All Paths}$}
        \ENDIF
    \ENDFOR
\ENDFUNCTION
\end{algorithmic}
\end{algorithm}

On entry to \textsc{FillBuffer}, the entries
$\texttt{buf}[0],\dots,\texttt{buf}[d]$ encode a valid
root-to-depth-$d$ prefix. Each child writes its destination
at depth $d+1$ without altering that prefix. Later branches
overwrite deeper entries, so stale values beyond the active
prefix are never read. Since every permitted child position
is present in the table, the branch choices and endpoint
filter are exactly those of \Cref{alg:gencomb_dfs}.
The routine therefore returns the same raw paths, in the
same order, including the sole root path when $D=0$.
Copies of completed paths remain unchanged when the buffer
is reused. Equal stored node values do not merge different
position sequences: the keys, not the values, determine
which path has been constructed. No recursive results are
memoised and no cardinality classes are discarded here. 

\section{Recursion-Free Nonnegative Raw-Path Generation}
The working array is the position sequence $\bk=(k_0,\dots,k_D)=\pi(\ppath)$.
Each recorded output restores the depth coordinates and copies the completed vertex path. The routine below returns the nonnegative path family; excursion augmentation is a separate operation.
\begin{algorithm}[H]
\small
\caption{Recursion-free generation of nonnegative raw paths}
\label{alg:recursivelessgen}
\begin{algorithmic}[1]
\footnotesize
\REQUIRE $D\in\Z_+$ and $k^*\in\{0,\dots,D\}$.
\ENSURE Return every path in $\PosPathsDk{k^*}$ exactly once,
in descending order under \eqref{S:LexOrd}.
\FUNCTION{RecursionFreeGenerateCombinations}{$D,k^*$}
    \STATE $\bk\leftarrow$ \CALL{Init-Array}{$D,k^*$}
    \STATE $\texttt{All Paths}\leftarrow$ empty list
    \WHILE{\textbf{true}}
        \STATE Append a copy of $\bigl((k_d,d)\bigr)_{d=0}^{D}$
            to \texttt{All Paths}
        \STATE $j\leftarrow$ \CALL{ComputeTickDown}{$\bk,D,k^*$}
        \IF{$j=0$}
            \STATE \textbf{break}
        \ENDIF
        \STATE $k_j\leftarrow k_j-1$
        \FOR{$d=j+1,\dots,D$}
            \STATE $k_d\leftarrow\min\{k_{d-1}+1,\,k^*+D-d\}$
        \ENDFOR
    \ENDWHILE
    \STATE \textbf{return} \texttt{All Paths}
\ENDFUNCTION
\FUNCTION{ComputeTickDown}{$\bk,D,k^*$}
    \FOR{$j=D-1,D-2,\dots,1$}
        \IF{$k_j>\max\{0,\,k_{j-1}-1,\,k^*-(D-j)\}$}
            \STATE \textbf{return} $j$
        \ENDIF
    \ENDFOR
    \STATE \textbf{return} $0$
\ENDFUNCTION
\end{algorithmic}
\end{algorithm}
The donor search does not modify $\bk$.
The return value $0$ indicates that there is no internal
coordinate to decrease; it is not an instruction to change
the root. For $D\le1$, the search range is empty and the
initial path is the only output. All stored copies remain
unchanged when the working array is updated.

\paragraph{Initialization (Path Init Algorithm)}
\begin{algorithm}[H]
\small
\caption{Initialize the maximal nonnegative position sequence}
\label{alg:init_array}
\begin{algorithmic}[1]
\footnotesize
\REQUIRE $D\in\Z_+$ and $k^*\in\{0,\dots,D\}$.
\ENSURE Return the position sequence in \eqref{E:highestpath}.
\FUNCTION{Init-Array}{$D,k^*$}
    \STATE Create $\bk=(k_0,\dots,k_D)$
    \FOR{$d=0,\dots,D$}
        \STATE $k_d\leftarrow\min\{d,\,k^*+D-d\}$
    \ENDFOR
    \STATE \textbf{return} $\bk$
\ENDFUNCTION
\end{algorithmic}
\end{algorithm}

\paragraph{Correctness of initialization and traversal.}
Any nonnegative path ending at $k^*$ satisfies
$k_d\le d$ and $k_d\le k^*+D-d$. The initializer attains their minimum at every depth. Its successive differences belong to $\{-1,0,1\}$, and it starts at $0$ and ends at $k^*$. It is therefore the largest path in the family. The two parity cases of \eqref{E:highestpath} are exactly this sequence, with a single stay at the peak when $D-k^*$ is odd. With the prefix through depth $j-1$ fixed, the feasible
values at depth $j$ form the integer interval
\[
\begin{split}
    \max\{0,k_{j-1}-1,k^*-(D-j)\}
    &\le k_j\\
    &\le\min\{k_{j-1}+1,k^*+D-j\}.
\end{split}
\]
The bounds express the permitted step and the remaining
distance to the endpoint. They are also sufficient:
a nonnegative position within that distance can reach
$k^*$ by monotone steps and stays. Thus \eqref{eq:raw-tick-feasibility} holds exactly when the current value can be decreased while retaining some valid completion of the prefix. After the decrement, \eqref{eq:raw-suffix-completion} chooses the largest feasible next position at each depth. Each choice preserves nonnegativity, the step bound, and endpoint reachability, so it constructs the largest valid suffix. Any path smaller than the current one first differs at
an internal depth where such a decrease is possible.
Choosing the largest such depth preserves the longest
possible prefix. Decreasing its value by one and taking
the largest suffix therefore gives the immediate next
path in descending order. If no such depth exists,
there is no smaller path. Starting at the largest path
and repeatedly taking its immediate successor exhausts
the finite family without omissions or repetitions.
The initializer also covers $D=0$, returning $(0)$,
which is recorded as the vertex path $((0,0))$.

\section{Stage Reseeding and Direct Tuple Enumeration}
\begin{algorithm}[H]
\small
\caption{Shift-and-reseed $\mathcal S_M$ from a canonical seed}
\label{alg:init_neg_array}
\begin{algorithmic}[1]
\REQUIRE $D\in\Z_+$, $0\le k^*\le D$, $0\le M\le-k_-$;
$\hat s^{(M)}=(c_k)_{k=k_-}^{k_+}$ satisfies \eqref{eq:stage-seed}.
\ENSURE The next canonical seed, or a report that no next stage exists.
\FUNCTION{Shift-and-Reseed}{$D,k^*,M,\hat s^{(M)}$}
    \IF{$M=-k_-$}
        \STATE \textbf{return} ``no next stage''
    \ENDIF
    \STATE $\beta\leftarrow1+D+k^*-2k_+$
    \STATE Create $v$ with local indices $0,\dots,k_+$
    \STATE $v_0\leftarrow1$
    \FOR{$j=1,\dots,k_+$}
        \STATE $v_j\leftarrow c_{j-M-1}$
    \ENDFOR
    \IF{$\beta=1$}
        \STATE $v_{k_+}\leftarrow v_{k_+}-1$
    \ENDIF
    \STATE $v_1\leftarrow v_1+1$
    \FOR{$k=k_-,\dots,k_+$}
        \STATE $c'_k\leftarrow0$
    \ENDFOR
    \FOR{$j=0,\dots,k_+$}
        \STATE $c'_{j-M-1}\leftarrow v_j$
    \ENDFOR
    \STATE \textbf{return} $\hat s^{(M+1)}=(c'_k)_{k=k_-}^{k_+}$
\ENDFUNCTION
\end{algorithmic}
\end{algorithm}

The following routine implements the stage and mass-layer
organization already established. In a positive-mass layer,
write
\[
    \boldsymbol{\alpha}=(\alpha_0,\dots,\alpha_{\ell-1}),
    \quad \alpha_j\in\Z_+,
    \quad \sum_{j=0}^{\ell-1}\alpha_j=i,
\]
for the residual distribution. These are the entries of the
additive vectors in \Cref{mass_process,second_partition};
slot $j$ corresponds to physical position $j-M$. The composition loop visits these vectors in descending
lexicographic order of $(\alpha_{\ell-1},\dots,\alpha_0)$.
For $\ell\ge2$, choose $\alpha_{\ell-1}$ from $i$ down to $0$.
For each fixed choice of the higher coordinates, choose
$\alpha_j$, for $j=\ell-2,\dots,1$, from $i-\sum_{r=j+1}^{\ell-1}\alpha_r$ down to $0$, and set
$\alpha_0=i-\sum_{r=1}^{\ell-1}\alpha_r$.
For $\ell=1$, the sole composition is $(i)$. These bounded nested choices specify the complete inner loop and its order; a counter-array implementation does not require recursive calls or sorting. For a fixed stage and layer, the original highest occupied position and the receiving-slot count are
\[
    b=\left\lfloor\frac{D+k^*-i}{2}\right\rfloor-M,
    \quad
    \ell=\left\lceil\frac{D+k^*-i}{2}\right\rceil.
\]
Thus $\ell$ is either $b+M$ or $b+M+1$.
When $\ell=b+M$, the associated positive tuple has a
single visit at its peak, which is excluded from the
receiving slots. When $\ell=b+M+1$, reserve a second
visit at that peak and include it among the receiving
slots. The negative-return visits are restored through
$\mathbf 1_{\{-M<k\le0\}}$, as in the inverse of
\eqref{eq:stage-positive-map}.

\begin{algorithm}[H]
\small
\caption{Generate cardinality tuples by stage and mass layer}
\label{alg:unique_path}
\begin{algorithmic}[1]
\footnotesize
\REQUIRE $D\in\Z_+$ and $k^*\in\{0,\dots,D\}$.
\ENSURE Emit each tuple in $\mathcal C_{D,k^*}$ exactly once;
return the number emitted.
\FUNCTION{Gen-Unique-Comb}{$D,k^*$}
    \STATE Compute $k_-,k_+$ from \eqref{E:Kbounds};
        set $T\leftarrow-k_-$
    \STATE Initialize $\hat s^{(0)}$ by \eqref{E:starting_tuple}
    \STATE $\texttt{total}\leftarrow0$
    \FOR{$M=0,\dots,T$}
        \STATE \textbf{emit} $\hat s^{(M)}$
        \STATE $\texttt{total}\leftarrow\texttt{total}+1$
        \FOR{$i=1,\dots,D-k^*-2M$}
            \STATE $\ell\leftarrow
                \left\lceil(D+k^*-i)/2\right\rceil$
            \IF{$\ell\ge1$}
                \STATE $b\leftarrow
                    \left\lfloor(D+k^*-i)/2\right\rfloor-M$
                \FORALL{residual compositions
                    $\boldsymbol{\alpha}$ in the order specified above}
                    \FOR{$k=k_-,\dots,k_+$}
                        \STATE $c_k\leftarrow0$
                    \ENDFOR
                    \FOR{$k=-M,\dots,b$}
                        \STATE $c_k\leftarrow
                            1+\mathbf 1_{\{-M<k\le0\}}
                            +\mathbf 1_{\{k^*\le k<b\}}$
                    \ENDFOR
                    \IF{$\ell=b+M+1$}
                        \STATE $c_b\leftarrow c_b+1$
                    \ENDIF
                    \FOR{$j=0,\dots,\ell-1$}
                        \STATE $c_{j-M}\leftarrow
                            c_{j-M}+\alpha_j$
                    \ENDFOR
                    \STATE \textbf{emit}
                        $\hat c=(c_k)_{k=k_-}^{k_+}$
                    \STATE $\texttt{total}\leftarrow
                        \texttt{total}+1$
                \ENDFOR
            \ENDIF
        \ENDFOR
        \IF{$M<T$}
            \STATE $\hat s^{(M+1)}\leftarrow$
                \CALL{Shift-and-Reseed}{$D,k^*,M,\hat s^{(M)}$}
        \ENDIF
    \ENDFOR
    \STATE \textbf{return} $\texttt{total}$
\ENDFUNCTION
\end{algorithmic}
\end{algorithm}

Empty index ranges execute no iterations. Each emission
supplies the current tuple's value to the output consumer;
a consumer that retains tuples must retain separate snapshots.
The working tuple $\hat c$ is separate from the saved
canonical seed $\hat s^{(M)}$, which is not modified by
within-stage enumeration. The reseeding call uses
\Cref{alg:init_neg_array} only after the current stage is
exhausted. The separately emitted seed is layer $i=0$;
positive-mass layers with $\ell=0$ are skipped.

\begin{theorem}[Correctness of tuple enumeration]
\label{thm:tuple-enumeration-correctness}
For $D\in\Z_+$ and $0\le k^*\le D$,
\Cref{alg:unique_path} terminates, emits every tuple in
$\mathcal C_{D,k^*}$ exactly once in descending order under
\Cref{card_ordering}, and returns the count in
\eqref{eq:dynamic-total-count-final}.
\end{theorem}

\begin{proof}
The initialization and \Cref{alg:init_neg_array} supply
the canonical seed of each stage. For fixed $M$ and
$i\ge1$, the associated positive problem has peak
$b+M=\lfloor(D+k^*-i)/2\rfloor$.
When $\ell=b+M$, its peak baseline is $1$ and the
receiving slots exclude that peak. When $\ell=b+M+1$,
its peak baseline is $2$ and the receiving slots include
it. These are exactly the two cases of
\Cref{prop:positive-mass-layers}.
The remaining baseline entries and the negative-return
indicator implement the inverse of
\eqref{eq:stage-positive-map}. Hence every output is an
admissible tuple in the specified layer of $\mathcal B_M$.

Each weak composition determines exactly one sequence
of the bounded coordinate choices, and every such sequence
gives a nonnegative vector of sum $i$.
Thus the inner loop visits every composition exactly once.
Adding the fixed baseline and applying the stage
correspondence are injective, so there are no repetitions
within a layer. The layers are disjoint and exhaustive by
\Cref{prop:positive-mass-layers}; the stages are disjoint
and exhaustive by \eqref{eq:stage-partition}.
Together with the separately emitted seeds, their outputs
are therefore exactly $\mathcal C_{D,k^*}$.

The bounded choices give descending right-to-left order.
The baseline addition and stage correspondence preserve
it; increasing $i$ and then $M$ gives the layer and stage
order already proved. All choices and loop ranges are
finite, so the procedure terminates.
There is one seed emission per stage and
$\binom{i+\ell(i)-1}{\ell(i)-1}$ emissions per nonempty
positive-mass layer. The counter therefore returns
\eqref{eq:dynamic-total-count-final}.
When $D=k^*$, including $D=0$, only the seed is emitted;
zero-slot layers emit nothing.
\end{proof}

For a negative endpoint, run the construction for its
absolute value and reflect every emitted tuple using
\eqref{eq:tuple-reflection}. The count is unchanged.
A count-only version omits tuple construction and emissions,
counts one seed per stage, and adds the binomial layer
cardinalities in \eqref{eq:dynamic-total-count-final}.
Constant-delay or loopless implementation claims require
a separate analysis of the composition traversal and
output costs.

\section{Excursions}
\label{appendix:excursions}
Another useful tool for constructing paths in $\PathsDk{k^*}$ is to exploit
\emph{excursions from $0$} in the position sequence $\pi(\ppath)$ (see \Cref{E:piDef}).

\paragraph{Motivating example}
Assume $D=9$ and $k^*=2$.  Consider the position sequence
\begin{equation}\label{E:basepath}
   (0,1,2,1,0,1,0,0,1,2).
\end{equation}
There are three maximal positive runs (excursions) bounded by zeros:
\[
   (0\,|\,\underbrace{1,2,1}_{\text{exc.\ 1}}\,|\,0\,|\,\underbrace{1}_{\text{exc.\ 2}}\,|\,0,0\,|\,\underbrace{1,2}_{\text{exc.\ 3}}),
\]
where the vertical bars mark zeros. Since the last run ends at $d=D$ with value
$k^*=2\neq 0$, the rightmost block is \emph{locked} (see below) and is not
flippable if we wish to preserve the terminal position. Flipping any subset of
the \emph{unlocked} excursions (here, excursions 1 and 2) negates the entries
inside those runs and yields valid paths that still end at $k^*$. For instance,
flipping excursion~1, excursion~2, both, or neither produces the four sequences
\begin{equation}\label{E:flippedexcursions}
\begin{aligned}
&(0,-1,-2,-1,0,1,0,0,1,2),\\
&(0,1,2,1,0,-1,0,0,1,2),\\
&(0,1,2,1,0,1,0,0,1,2),\\
&(0,-1,-2,-1,0,-1,0,0,1,2).
\end{aligned}
\end{equation}

\paragraph{Formal setup.}
Let $\bk\Def(k_0,k_1,\dots,k_D)\in\Z^{D+1}$ be a position sequence with
$k_0=0$ and $k_D=k^*$. We call a pair of indices $(l,r)$ with
$0\le l<r\le D+1$ and $l+1<r$ an \emph{excursion interval} if
\[
   k_l=0,\quad
   k_r=0\ \text{when }r\le D,\quad
   k_d>0\ \; \forall l<d<r,
\]
and $(l,r)$ is maximal with these properties. When the final positive run reaches depth $D$ with $k_D>0$, we \emph{close} it by the sentinel $r=D+1$. Thus every excursion interval contains at least one strictly positive entry. The open index set $(l,r)\cap\{0,1,\dots,D\}$ is the support of the excursion. Let $\nu=\nu(\bk)$ denote the number of excursion intervals. When $\nu\ge1$, list them from left to right as
\[
   0\le l_1<r_1\le l_2<r_2\le\cdots
   \le l_{\nu}<r_{\nu}\le D+1.
\]
When $\nu=0$, the path has no positive excursion intervals. Define the \emph{excursion indicators} $\bOne^{(m)}\in\{0,1\}^{D+1}$ by
\[
   \bOne^{(m)}_d
   =
   \begin{cases}
      1, & l_m<d<r_m,\\
      0, & \text{otherwise}.
   \end{cases}
\]
Set the \emph{lock flag}
\[
   \epsilon(\bk)\Def
   \begin{cases}
      1,
      & \nu\ge1\text{ and }r_{\nu}=D+1,\\[2pt]
      0,
      & \text{otherwise}.
   \end{cases}
\]
Thus $\epsilon(\bk)=1$ exactly when the rightmost positive
excursion reaches the terminal depth with $k_D\neq0$ and is
therefore locked. In particular, $\epsilon(\bk)=0$ when
$\nu=0$. The indices of the \emph{flippable} excursions are therefore
\[
   \mathcal I^*(\bk)
   \Def
   \{1,2,\dots,\nu-\epsilon(\bk)\}.
\]
Here $\{1,\dots,0\}$ is understood to be empty.

\paragraph{Flip operator}
For any subset $A\subseteq \mathcal{I}^*(\bk)$ define
\[
   \Flip_A(\bk)
   \;\Def\;
   \bk \;-\; 2\sum_{m\in A} \bigl(\bk\odot \bOne^{(m)}\bigr),
\]
where $\odot$ denotes the Hadamard (entrywise) product. In words: on each
excursion $m\in A$ we negate the entries of $\bk$ (equivalently, we reflect the
excursion about $0$), and we leave all other indices unchanged. The full
\emph{flip family} is
\begin{equation}\label{eq:FlipFamily}
   \Flip(\bk)\;\Def\;\bigl\{\Flip_A(\bk): A\subseteq \mathcal{I}^*(\bk)\bigr\}.
\end{equation}
The supports of distinct excursion intervals are disjoint, and by definition each contains at least one strictly positive entry. Hence the map $A\longmapsto\Flip_A(\bk)$ is injective: if $A\ne A'$, choose an excursion index in their symmetric difference. Its positive entries are
negated in exactly one of $\Flip_A(\bk)$ and $\Flip_{A'}(\bk)$, so the resulting position sequences are different. Therefore
\begin{equation}\label{eq:FlipCount}
   \bigl|\Flip(\bk)\bigr|
   =
   2^{\,\nu-\epsilon(\bk)}.
\end{equation}
If one wishes to exclude the trivial ``no-flip'' element corresponding to $A=\varnothing$, the count becomes
$2^{\,\nu-\epsilon(\bk)}-1$.

\begin{lemma}[Validity and endpoint preservation]
\label{lem:flip-valid}
Let $\ppath\in\PosPathsDk{k^*}$ and set $\bk\Def\pi(\ppath).$ For every $A\subseteq\mathcal I^*(\bk)$, the sequence $\widetilde{\bk}\Def\Flip_A(\bk)$ is the position sequence of a path in $\PathsDk{k^*}$. In particular,
\[
    \widetilde{k}_d-\widetilde{k}_{d-1}
    \in\{-1,0,1\},
    \quad d=1,\dots,D,
\]
and
\[
    \widetilde{k}_0=0,
    \quad
    \widetilde{k}_D=k^*.
\]
\end{lemma}

\begin{proof}
Consider a flipped excursion $(l_m,r_m)$. At indices strictly inside its support, both consecutive positions are negated. Hence, if $k_d-k_{d-1}\in\{-1,0,1\},$ then $\widetilde{k}_d-\widetilde{k}_{d-1} =-(k_d-k_{d-1})\in\{-1,0,1\}.$ In particular, horizontal steps remain horizontal after the reflection. At the left boundary, $k_{l_m}=0$ and the first positive entry, when present, must equal $1$ because the original path
has increments in $\{-1,0,1\}$. After reflection this entry becomes $-1$, so the boundary increment changes from $+1$ to $-1$ and remains valid. If $r_m\le D$, the same argument at the right boundary changes the final increment from $-1$ to $+1$. Thus every increment associated with a flipped excursion remains in $\{-1,0,1\}$. A terminal excursion with $r_m=D+1$ is excluded from $\mathcal I^*(\bk)$ by the lock flag and is therefore never flipped. Consequently the terminal value $k_D=k^*$ is unchanged. Every nonterminal flipped excursion begins and ends at $0$, so reflecting it also preserves its zero net
displacement. Hence $\widetilde{\bk}=\Flip_A(\bk)$ is a valid position sequence with the same initial and terminal positions.
\end{proof}

We use the nonnegative path family $\PosPathsDk{k^*}$ from \eqref{eq:positive-path-family}.

\begin{proposition}[Flip representation]\label{prop:flip-representation}
If $k^*\ge 0$, then
\begin{equation}\label{E:fliprep}
\PathsDk{k^*}
=
\bigsqcup_{\ppath\in\PosPathsDk{k^*}}
\pi^{-1}\!\bigl(\Flip(\pi(\ppath))\bigr).
\end{equation}
If $k^*<0$, then
\begin{equation}\label{eq:flip-negative}
   \PathsDk{k^*}
   \;=\;
   \bigsqcup_{\ppath\in \PosPathsDk{-k^*}}
   \ \pi^{-1}\!\Bigl(-\,\Flip\bigl(\pi(\ppath)\bigr)\Bigr).
\end{equation}
\end{proposition}

\begin{proof}
First suppose $k^*\ge0$. Let
$\ppath\in\PathsDk{k^*}$ have position sequence $\bk=(k_0,\dots,k_D).$
Define $\bk^{(+)}\Def(|k_0|,\dots,|k_D|).$ Since
\[
    \bigl||k_d|-|k_{d-1}|\bigr|
    \le
    |k_d-k_{d-1}|
    \le1,
\]
the sequence $\bk^{(+)}$ has increments in $\{-1,0,1\}$. It begins at $0$, ends at $k^*$, and is nonnegative. Hence $\bk^{(+)}=\pi(\ppath_+)$ for a unique
$\ppath_+\in\PosPathsDk{k^*}$. The negative entries of $\bk$ occur in maximal blocks bounded by zeros. These blocks correspond exactly to a subset of the flippable excursions of $\bk^{(+)}$. Flipping precisely those excursions therefore recovers $\bk$. Thus every path in $\PathsDk{k^*}$ belongs to one
of the displayed flip families. Conversely, \Cref{lem:flip-valid} shows that every member of a flip family is a valid path in $\PathsDk{k^*}$. Finally, the representation is unique. If the same path were obtained from two nonnegative representatives, taking coordinatewise absolute values would recover both of them, so the representatives must coincide. For a fixed representative, the injectivity of $A\mapsto\Flip_A(\bk)$ proved above implies that the flip
subset is also unique. Hence the union is disjoint. For $k^*<0$, apply the preceding argument at endpoint $-k^*>0$ and then globally reflect the resulting position sequences.
\end{proof}

\paragraph{Counting flips for a fixed representative}
Let $\bk=\pi(\ppath)\in\Z_+^{D+1}$ and let $\nu=\nu(\bk)$ be its number of excursions. By \Cref{eq:FlipCount}, $|\Flip(\bk)|=2^{\,\nu-\epsilon(\bk)},$ or $2^{\,\nu-\epsilon(\bk)}-1$ when the no-flip element is excluded. In \Cref{E:basepath}, $\nu=3$ and $\epsilon(\bk)=1$, so
$|\Flip(\bk)|=2^2=4,$ exactly the four sequences in \Cref{E:flippedexcursions}.

\section{Step counts and parity}
\label{appendix:step-counts}
Any $\ppath\in \PathsDk{k^*}$ can be decomposed into $j_+$ up-steps,
$j_-$ down-steps, and $j_0$ stays. Then
\begin{equation}\label{E:algebra}
   j_+ + j_- + j_0 \;=\; D,
   \qquad
   j_+ - j_- \;=\; k^*.
\end{equation}
Solving gives
\[
   j_+ \;=\; \frac{D+k^*-j_0}{2},
   \qquad
   j_- \;=\; \frac{D-k^*-j_0}{2}.
\]
Because $D+k^*$ and $D-k^*$ have the same parity, $j_+$ and $j_-$ are integers if and only if
\begin{equation}
\label{eq:parity-j0}
    j_0\equiv D+k^*\pmod 2.
\end{equation}
Their nonnegativity further requires $0\le j_0\le D-|k^*|.$ Conversely, every integer $j_0$ satisfying these two conditions produces nonnegative integers
\[
    j_+
    =
    \frac{D+k^*-j_0}{2},
    \qquad
    j_-
    =
    \frac{D-k^*-j_0}{2}.
\]
In particular,
\[
    j_+\le
    \left\lfloor\frac{D+k^*}{2}\right\rfloor,
    \qquad
    j_-\le
    \left\lfloor\frac{D-k^*}{2}\right\rfloor.
\]
Thus the admissible stay counts are exactly the integers $j_0\in[0,D-|k^*|]$ having the parity prescribed by \eqref{eq:parity-j0}. Once $j_0$ is fixed, the corresponding numbers of upward and downward steps are uniquely determined
by \eqref{E:algebra}. For the maximal path used to initialize the recursion-free
generator, the number of stays is $j_0=D+k^*-2k_+\in\{0,1\}.$ Since $2k_+$ is even, this value automatically satisfies \eqref{eq:parity-j0}. Thus the parity condition is built into the maximal seed rather than requiring a separate adjustment during path generation.

\paragraph{Integration with the positive-path generator.}
For $k^*\ge0$, \Cref{alg:recursivelessgen} enumerates the nonnegative representatives $\PosPathsDk{k^*}$ in descending raw-path order. By \Cref{prop:flip-representation}, applying the flip construction to these representatives recovers the full path family $\PathsDk{k^*}$. This representation is unique. Indeed, for any $\ppath\in\PathsDk{k^*}$, taking the coordinatewise absolute value of its position sequence recovers its unique
nonnegative representative. The negative excursion blocks then determine exactly which excursions of that representative were flipped. Hence distinct representative--flip choices cannot produce the same raw path, and the
two-stage procedure exhausts $\PathsDk{k^*}$ without duplication. Negative endpoints follow by the global reflection in \Cref{thm:reflection-symmetry}.

\begin{remark}[Ordering after excursion expansion]
The excursion expansion preserves deterministic enumeration, but it need not preserve the raw-path order \eqref{S:LexOrd}. One may emit each nonnegative
representative in the order generated by \Cref{alg:recursivelessgen} and enumerate its flip family in any fixed subset order, for example lexicographic order on the flippable excursion indices $\mathcal I^*(\bk)$. This defines a reproducible total output order. If the full family is required specifically in the raw-path order \eqref{S:LexOrd}, the expanded outputs must instead be
merged or ordered according to that relation. This issue is separate from the cardinality-tuple order used by the direct unique-tuple enumerator.
\end{remark}

\section{Cardinality Equivalence and Reordering Invariance}
\label{proof_er}

We now justify the quotient interpretation used throughout the paper and characterize precisely what information is discarded by the cardinality map. Fix $D\in\Z_+$ and $k^*\in\Z$ with $|k^*|\le D$. Throughout this section,
$k_-,k_+$ retain the meanings in \eqref{E:Kbounds}, and all position-indexing intervals contain their integer points.

\begin{proposition}[Cardinality representation and common position range]
Every path in $\PathsDk{k^*}$ has a unique cardinality tuple
under \eqref{card_map}. The positions that occur among these
paths form exactly the common range $[k_-,k_+]$, and the
total-mass and aggregation identities
\eqref{E:totalmass} and \eqref{E:efficientaggregation} hold.
\end{proposition}

\begin{proof}
Let $\ppath\in\PathsDk{k^*}$ have position sequence
$\pi(\ppath)=(k_0,\dots,k_D)$. Suppose $\ppath$ visits
position $k$ at depth $d$. Reaching $k$ from $0$ requires
at least $|k|$ steps, while reaching $k^*$ from $k$ requires
at least $|k^*-k|$ further steps. Hence
\[
    |k|+|k^*-k|
    \le d+(D-d)
    =D.
\]
Since
\[
    |k|+|k^*-k|\ge 2k-k^*,
    \quad
    |k|+|k^*-k|\ge k^*-2k,
\]
we obtain
\[
    \frac{k^*-D}{2}
    \le k\le
    \frac{D+k^*}{2}.
\]
Because $k$ is integral,
\[
    k_-=
    \left\lceil\frac{k^*-D}{2}\right\rceil
    \le k\le
    \left\lfloor\frac{D+k^*}{2}\right\rfloor
    =k_+.
\]
Conversely, let $k$ be any integer in $[k_-,k_+]$.
The definitions of the bounds imply
\[
    2k-k^*\le D,
    \quad
    k^*-2k\le D.
\]
Together with $|k^*|\le D$, these inequalities give
\[
    |k|+|k^*-k|
    =
    \max\bigl\{
        |k^*|,\,
        2k-k^*,\,
        k^*-2k
    \bigr\}
    \le D.
\]
A path may therefore move from $0$ to $k$ in $|k|$ steps, then from $k$ to $k^*$ in $|k^*-k|$ steps, and use the remaining $D-|k|-|k^*-k|$ steps as stays at $k^*$. Thus some path in $\PathsDk{k^*}$ visits $k$, so $[k_-,k_+]$ is exactly the common range of positions occurring among admissible paths. For a fixed path $\ppath$, the visit counts are uniquely determined by \eqref{card_map}. Since every depth $d\in\{0,\dots,D\}$ contributes to exactly one position count, $\sum_{k=k_-}^{k_+}c_k(\ppath)=D+1,$ which proves \eqref{E:totalmass}. Likewise, grouping the position-dependent weights by the position visited gives
\[
\begin{aligned}
    \sum_{d=0}^{D}\vvalue_{k_d}
    =
    \sum_{k=k_-}^{k_+}
    \sum_{\substack{d\in\{0,\dots,D\}\\k_d=k}}
    \vvalue_k 
    =
    \sum_{k=k_-}^{k_+}c_k(\ppath)\,\vvalue_k,
\end{aligned}
\]
which proves \eqref{E:efficientaggregation}. When $D=0$, necessarily $k^*=0$, the common range is $[0,0]$, and the sole path has $c_0(\ppath)=1$, so the same
statements remain valid
\end{proof}

\begin{proposition}[Equivalence classes and cardinality tuples]
The relation \eqref{equivalence_relation} is an equivalence
relation on $\PathsDk{k^*}$. Its classes partition
$\PathsDk{k^*}$, and the assignment
\[
    [\ppath]_\sim
    \longmapsto
    \hat{C}_{D,k^*}(\ppath)
\]
is a bijection from $\PathsDk{k^*}/\sim$ to
$\mathcal C_{D,k^*}$.
\end{proposition}

\begin{proof}
By \eqref{equivalence_relation}, two paths are equivalent
precisely when their cardinality tuples are equal.
Reflexivity, symmetry, and transitivity therefore follow
from the corresponding properties of equality. Every path belongs to its own equivalence class. If two classes share a path, their defining tuples are equal, so the classes coincide. Consequently, the classes are nonempty, pairwise disjoint, and cover $\PathsDk{k^*}$. The displayed assignment is well-defined because all paths in the same class have the same tuple. It is injective because equal tuples define the same equivalence class. It is surjective because $\mathcal C_{D,k^*}$ is the image of $\hat{C}_{D,k^*}$ by \eqref{full_card_set}.
\end{proof}

\begin{definition}[Legal reordering]
Let $\ppath\in\PathsDk{k^*}$ have position sequence
$\pi(\ppath)=(k_0,\dots,k_D)$.
For a bijection
$R:\{0,\dots,D\}\to\{0,\dots,D\}$, define the vertex sequence
\[
    \ppath^R
    \Def
    \bigl(
        (k_{R^{-1}(d)},d)
    \bigr)_{d=0}^{D}.
\]
The bijection $R$ is a \emph{legal reordering} for $\ppath$
when $\ppath^R\in\PathsDk{k^*}$. Thus positions are reassigned to depth indices, with the depth coordinate restored at each index. In particular, legality requires the reordered sequence to start at position $0$, end at position $k^*$, and satisfy the original edge condition at every step.
An arbitrary bijection of the depth indices need not be
legal.
\end{definition}

\begin{proposition}[Classes are exactly legal reorderings]
For every $\ppath\in\PathsDk{k^*}$,
\[
    [\ppath]_\sim
    =
    \left\{
        \ppath^R:
        R\text{ is a legal reordering for }\ppath
    \right\}.
\]
\end{proposition}

\begin{proof}
First, let $R$ be a legal reordering for $\ppath$. Because $R$ is a bijection of the depth indices, it changes only the order in which the positions are visited and not their multiplicities. Hence $c_k(\ppath^R)=c_k(\ppath), \;\forall k\in[k_-,k_+].$ Therefore $\hat C_{D,k^*}(\ppath^R)=\hat C_{D,k^*}(\ppath),$ and so $\ppath^R\in[\ppath]_\sim$. Conversely, let
$\ppath'\in[\ppath]_\sim$, with
\[
    \pi(\ppath)=(k_0,\dots,k_D),
    \quad
    \pi(\ppath')=(k'_0,\dots,k'_D).
\]
Since the two paths have the same cardinality tuple,
for every position $k$ the sets
\[
    A_k
    \Def
    \{d\in\{0,\dots,D\}:k_d=k\},
    \quad
    A'_k
    \Def
    \{d\in\{0,\dots,D\}:k'_d=k\}
\]
have equal cardinality. For each $k$, choose a bijection $R_k:A_k\longrightarrow A'_k.$ The sets $A_k$ partition $\{0,\dots,D\}$, as do the sets $A'_k$. Hence the maps $R_k$ combine uniquely into a
bijection $R:\{0,\dots,D\}\longrightarrow\{0,\dots,D\}$ such that
$k'_{R(d)}=k_d,\; \forall d\in\{0,\dots,D\} \Leftrightarrow k_{R^{-1}(d)}=k'_d, \;\forall d$. Restoring the depth coordinate therefore gives $\ppath^R=\ppath'.$
Since $\ppath'\in\PathsDk{k^*}$, this $R$ is a legal reordering for $\ppath$. Thus every member of $[\ppath]_\sim$ arises from a legal reordering, proving the
reverse inclusion.
\end{proof}

\begin{theorem}[Finest reordering-invariant partition]
\label{thm:minimality}
A partition $\Pi$ of $\PathsDk{k^*}$ is
\emph{reordering-invariant} if, for every block $B\in\Pi$,
every $\ppath\in B$, and every legal reordering $R$ for
$\ppath$, one has $\ppath^R\in B$. The cardinality-class partition is the finest reordering-invariant partition: every cardinality class is contained in a block of every reordering-invariant partition.
\end{theorem}

\begin{proof}
By the preceding proposition, the legal reorderings of any path $\ppath$ are exactly the paths in its cardinality class $[\ppath]_\sim$. Consequently, the cardinality-class partition is itself reordering-invariant. Now let $\Pi$ be any reordering-invariant partition and fix $\ppath\in\PathsDk{k^*}$. Let $B\in\Pi$ be the unique block containing $\ppath$. For every $\ppath'\in[\ppath]_\sim$, the preceding proposition gives a legal reordering $R$ of $\ppath$ such that $\ppath'=\ppath^R.$ Since $\Pi$ is reordering-invariant and $\ppath\in B$, it follows that $\ppath'\in B$. Hence
\[
    [\ppath]_\sim\subseteq B.
\]
Thus every cardinality class is contained in a block of
every reordering-invariant partition. Therefore the
cardinality-class partition refines every such partition
and is the finest reordering-invariant partition.
\end{proof}

\section{Maximality of the Positive Seed}
\label{app:maximal-positive-seed}

\begin{proof}[Proof of \Cref{prop:maximal-positive-seed}]
Fix $D\in\Z_+$ and $k^*\in\{0,\dots,D\}$.
By \eqref{E:Kbounds},
\[
    k_+=\left\lfloor\frac{D+k^*}{2}\right\rfloor,
    \qquad
    D+k^*-2k_+
    =
    \begin{cases}
        0, & \text{$D-k^*$ even},\\
        1, & \text{$D-k^*$ odd}.
    \end{cases}
\]
Here $D+k^*$ and $D-k^*$ have the same parity because
their difference is $2k^*$. First, construct a path by taking $k_+$ upward steps from $0$ to $k_+$, then $D+k^*-2k_+$ stays at $k_+$, and finally $k_+-k^*$ downward steps to $k^*$.
All positions are nonnegative, and the number of steps is
\[
k_+
+(D+k^*-2k_+)
+(k_+-k^*)
=D.
\]
Thus this is a path in $\PosPathsDk{k^*}$. Each position $0\le k<k^*$ is visited once. Each position $k^*\le k<k_+$ is visited twice: once before reaching the peak and once afterward. The peak is visited $1+D+k^*-2k_+$ times, and no negative position is visited. These are exactly the entries in \eqref{E:starting_tuple}, so the constructed path realizes $\hat{s}^{(0)}$. Its position sequence is the one in \Cref{E:highestpath}. The total mass is explicitly
\[
\begin{aligned}
\sum_{k=k_-}^{k_+}c_k=k^*+2(k_+-k^*)+(1+D+k^*-2k_+)=D+1.
\end{aligned}
\]
To prove maximality, let
$\hat c'=(c'_k)_{k=k_-}^{k_+}\in\mathcal C^+_{D,k^*}$
be arbitrary. If $c'_{k_+}=0$, then
$\hat c'\cardprec\hat{s}^{(0)}$, because the seed has
a positive count at the highest position $k_+$. Otherwise, every path realizing $\hat c'$ visits $k_+$. Such a path must visit every position $0\le k<k^*$ at least once. It must also visit every position $k^*\le k<k_+$ both before and after a visit to $k_+$, because it starts at $0$ and ends at $k^*$.
Consequently,
\[
    c'_k\ge 1
    \quad (0\le k<k^*),
    \qquad
    c'_k\ge 2
    \quad (k^*\le k<k_+).
\]
All negative-index counts vanish because
$\hat c'\in\mathcal C^+_{D,k^*}$.
Using \eqref{E:totalmass}, we therefore obtain
\[
\begin{aligned}
    c'_{k_+}
    &=
    D+1-\sum_{k=0}^{k_+-1}c'_k\\
    &\le
    D+1-k^*-2(k_+-k^*)\\
    &=
    1+D+k^*-2k_+.
\end{aligned}
\]
The final expression is precisely the seed's count
at $k_+$. If this inequality is strict, the seed is larger under \Cref{card_ordering}, since the comparison is decided
at $k_+$. If equality holds, every lower-coordinate bound must
also be attained: any larger lower-coordinate count
would force a smaller count at $k_+$. Hence
\[
    c'_k=1
    \quad (0\le k<k^*),
    \qquad
    c'_k=2
    \quad (k^*\le k<k_+),
\]
and all coordinates of $\hat c'$ agree with
\eqref{E:starting_tuple}. Thus $\hat c'=\hat{s}^{(0)}$. Every other admissible positive tuple is therefore strictly smaller than $\hat{s}^{(0)}$, proving that the seed is the unique largest tuple. Empty sums and index ranges are interpreted as zero and empty, respectively, so the argument also covers
$k_+=k^*$ and $D=k^*=0$.
\end{proof}

\section{Admissibility of Positive Cardinality Tuples}
\label{app:positive-admissibility}

\begin{lemma}[Positive-tuple admissibility]
\label{lem:positive-admissibility}
Fix $D\in\Z_+$ and $k^*\in\{0,\dots,D\}$.
Let $\hat c=(c_k)_{k=k_-}^{k_+}\in\Z_+^{[k_-,k_+]}$
satisfy
\[
    c_k=0 \quad (k_-\le k<0),
    \qquad
    \sum_{k=k_-}^{k_+}c_k=D+1.
\]
Then $\hat c\in\mathcal C^+_{D,k^*}$ if and only if
it satisfies conditions (1) and (2) of
\Cref{rem:mass_constraints}.
\end{lemma}

\begin{proof}
For necessity, suppose $\hat c$ is realized by a nonnegative path from $0$ to $k^*$. Because each step changes position by at most one, the path must visit every integer position from $0$ through $k^*$. Hence $c_k\ge 1$ when $(0\le k\le k^*).$ Now fix $k$ with $k^*\le k<k_+$ and suppose $c_{k+1}>0$. To reach $k+1$ from the starting position $0\le k$, the path must first visit $k$. After visiting $k+1$, it must visit $k$ again to reach its terminal position $k^*\le k$. These visits occur on opposite sides of a visit to $k+1$, so they are distinct. Therefore $c_k\ge 2$. For sufficiency, suppose the two conditions hold. Let
\[
    b\Def
    \max\bigl\{
        k\in[0,k_+]:c_k>0
    \bigr\}
\]
be the highest occupied position. This maximum exists, and $b\ge k^*$, because
$c_{k^*}\ge1$. Starting from $c_b>0$ and applying condition (2) successively downward gives $c_k\ge2$ when $(k^*\le k<b).$ Together with condition (1), we therefore have
\[
    c_k\ge1
    \qquad (0\le k<k^*),
    \qquad
    c_k\ge2
    \qquad (k^*\le k<b),
    \qquad
    c_b\ge1.
\]
All entries above $b$ vanish by its definition, and
all negative-index entries vanish by assumption. Consider the path that moves upward from $0$ to $b$ and then downward from $b$ to $k^*$, with no stays. It has $b+(b-k^*)=2b-k^*$ steps. Its visit counts are $1$ on $0\le k<k^*$, $2$ on $k^*\le k<b$, and $1$ at $b$. Its total number of visits is therefore
$k^*+2(b-k^*)+1=2b-k^*+1.$ Each prescribed count $c_k$ is at least the corresponding count of this path. At the first visit to each position, insert stays to supply the difference:
\[
\begin{cases}
    c_k-1 \text{ stays}, & 0\le k<k^*,\\
    c_k-2 \text{ stays}, & k^*\le k<b,\\
    c_b-1 \text{ stays}, & k=b.
\end{cases}
\]
Each inserted stay adds exactly one visit to its position
and does not alter any other visit count or the endpoint. The total number of inserted stays is
\[
\begin{aligned}
    &\sum_{k=0}^{k^*-1}(c_k-1)
    +\sum_{k=k^*}^{b-1}(c_k-2)
    +(c_b-1)\\
    &\qquad=
    \sum_{k=0}^{b}c_k
    -\bigl(k^*+2(b-k^*)+1\bigr)\\
    &\qquad=
    D+1-(2b-k^*+1)\\
    &\qquad=
    D-2b+k^*.
\end{aligned}
\]
All inserted counts are nonnegative. The resulting path
therefore has exactly $(2b-k^*)+(D-2b+k^*)=D$ steps, remains nonnegative, ends at $k^*$, and has cardinality tuple $\hat c$. Thus $\hat c\in\mathcal C^+_{D,k^*}$. Empty sums and index ranges contribute nothing. In particular, when $b=k^*$ the initial path is simply the ascent from $0$ to $k^*$; when $b=k^*=0$ it consists only of the root before the stays are inserted. The construction therefore also covers $D=0$.
\end{proof}

\section{Positive Mass-Layer Decomposition}
\label{app:positive-mass-layers}

\begin{proof}[Proof of \Cref{prop:positive-mass-layers}]
Fix $D\in\Z_+$ and $k^*\in\{0,\dots,D\}$, and let
$\hat c\in\mathcal C^+_{D,k^*}$ have highest occupied position
$b=\max\{k\in[0,k_+]:c_k>0\}.$ By \Cref{lem:positive-admissibility},
\[
    c_k\ge1
    \quad (0\le k<k^*), \textbf{ and }     c_k\ge2
    \quad (k^*\le k<b).
\]
We distinguish the value of the count at the highest
occupied position.

\paragraph{Case 1: $c_b=1$.}
Reserve one visit at each position
$0,\dots,k^*-1$, two visits at each position
$k^*,\dots,b-1$, and the single visit at $b$.
The number of visits remaining after these mandatory
occupancies are removed is
\[
\begin{aligned}
    i
    &=
    (D+1)
    -
    \left[
        k^*
        +2(b-k^*)
        +1
    \right]\\
    &=
    D+k^*-2b.
\end{aligned}
\]
These residual visits may occur only at positions
$0,\dots,b-1$ if the tuple is to remain in this case;
placing residual mass at $b$ would give $c_b\ge2$.
Hence there are $\ell=b$ available receiving slots. The residual counts therefore form a weak composition of $i$ into $\ell=b$ slots. Conversely, adding any such weak composition to the mandatory occupancies above produces an admissible tuple with highest occupied position $b$ and $c_b=1$, by
\Cref{lem:positive-admissibility}.

\paragraph{Case 2: $c_b\ge2$.}
Now reserve two visits at $b$ in addition to the
mandatory visits below it. The remaining mass is
\[
\begin{aligned}
    i
    &=
    (D+1)
    -
    \left[
        k^*
        +2(b-k^*)
        +2
    \right]\\
    &=
    D+k^*-2b-1.
\end{aligned}
\]
The residual mass may be distributed over positions
$0,\dots,b$, so in this case $\ell=b+1.$ Again the residual counts form a weak composition of $i$ into these $\ell$ slots, and every such composition produces an admissible tuple. In either case the residual counts are nonnegative, so
$i\ge0$. Since $b\ge k^*$, $i\le D-k^*.$ The highest occupied position $b$ is uniquely determined by the tuple, and the alternatives $c_b=1$ and $c_b\ge2$ are disjoint. Hence every admissible positive tuple belongs to exactly one mass layer. We can now express the number of available slots directly in terms of the layer index. From \eqref{eq:positive-layer-index},
\[
    D+k^*-i
    =
    \begin{cases}
        2b,   & c_b=1,\\
        2b+1, & c_b\ge2.
    \end{cases}
\]
Therefore
\[
    \left\lceil
        \frac{D+k^*-i}{2}
    \right\rceil
    =
    \begin{cases}
        b,   & c_b=1,\\
        b+1, & c_b\ge2,
    \end{cases}
\]
which is precisely the number of available receiving
slots in the two cases. Thus, for every positive-mass
layer,
\[
    \ell
    =
    \ell(i)
    =
    \left\lceil
        \frac{D+k^*-i}{2}
    \right\rceil.
\]
Consequently, whenever $\ell\ge1$, the number of tuples
in layer $i$ is the number of weak compositions of $i$ into $\ell$ slots: $\binom{i+\ell-1}{\ell-1}.$ If $i>0$ and $\ell=0$, no such composition exists and the layer is empty.
For $i=0$, the residual mass vanishes. Equation \eqref{eq:positive-layer-index} gives $b=k_+=\left\lfloor
\frac{D+k^*}{2}\right\rfloor,$ and the resulting mandatory occupancies are exactly those of the maximal seed $\hat{s}^{(0)}$ in \eqref{E:starting_tuple}. Hence layer $0$ contains exactly that tuple. Finally, consider the tuple ordering. Increasing $i$ either decreases the count at the current highest occupied position or causes that highest occupied position to move one level to the left. In either case, every tuple in the earlier layer is larger than every tuple in the later layer under \Cref{card_ordering}. Within a fixed layer, all tuples have the same mandatory baseline, so comparing two tuples is equivalent to comparing their residual weak compositions from the highest available position downward. Thus increasing mass layers, with residual compositions listed in descending coordinate order within each layer, recover the descending cardinality-tuple order. The recovery of $b$, the case $c_b=1$ versus $c_b\ge2$, and the residual weak composition is unique. Hence the layers are disjoint and together exhaust
$\mathcal C^+_{D,k^*}$.
\end{proof}

\section{Parity Form of the Positive Count}
\label{app:positive-parity-count}

\begin{proof}[Proof of the parity form
\eqref{eq:final-closed-form}]
Fix $D\in\Z_+$ and $k^*\in\{0,\dots,D\}$. From \eqref{switching_term}, $D+k^*=2k_++\beta-1.$ For $1\le i\le D-k^*$, \Cref{prop:positive-mass-layers} gives $\ell(i)=
\left\lceil\frac{D+k^*-i}{2}\right\rceil.$ If $\beta=1$, then $D+k^*=2k_+$ and
\[
\begin{aligned}
\ell(i)=\left\lceil k_+-\frac{i}{2}\right\rceil=k_+-\left\lfloor\frac{i}{2}\right\rfloor=k_++1-\left\lceil\frac{i+1}{2}\right\rceil.
\end{aligned}
\]
If $\beta=2$, then $D+k^*=2k_++1$ and
\[
\begin{aligned}
    \ell(i)
    &=
    \left\lceil
        k_++\frac{1-i}{2}
    \right\rceil
    =
    k_++1-
    \left\lfloor\frac{i+1}{2}\right\rfloor\\
    &=
    k_++1-
    \left\lceil\frac{i+1}{2}\right\rceil
    +
    \delta(i).
\end{aligned}
\]
Together these identities prove \eqref{eq:uniform-delta}. Now set $q=\left\lceil\frac{i+1}{2}\right\rceil$ temporarily within this proof. If $(\beta-1)\delta(i)=0$, then $\ell(i)=k_++1-q,$ and therefore
\[
    \binom{i+\ell(i)-1}{\ell(i)-1}
    =
    \binom{i+k_+-q}{k_+-q}.
\]
This is exactly the first term in \eqref{eq:single-summand}, while the second vanishes. If $(\beta-1)\delta(i)=1$, then necessarily $\beta=2$ and $i$ is even. In this case $\ell(i)=k_++2-q.$ Hence
\[
    \binom{i+\ell(i)-1}{\ell(i)-1}
    =
    \binom{i+k_++1-q}{k_++1-q}.
\]
Pascal's identity gives
\[
\begin{aligned}
\binom{i+k_++1-q}{k_++1-q}=\binom{i+k_+-q}{k_+-q}+
\binom{i+k_+-q}{k_++1-q},
\end{aligned}
\]
which is precisely \eqref{eq:single-summand}. Thus the two-binomial expression equals the weak-composition count $\binom{i+\ell(i)-1}{\ell(i)-1}$ for every positive-mass layer. Finally, \Cref{prop:positive-mass-layers} partitions $\mathcal C^+_{D,k^*}$ into the unique seed and the positive-mass layers. Summing their cardinalities gives \eqref{eq:positive-count-layers}, and substituting \eqref{eq:single-summand} into that sum gives
\eqref{eq:final-closed-form}. When $D=k^*$, the sum is empty and the count is $1$. When $k^*=0<D$, the final nominal index $i=D$ has $\ell(D)=0$ and contributes no tuple, consistently
with the zero-slot convention.
\end{proof}

\section{Reflection and the First Mixed Case}
\label{app:reflection-first-mixed}

\subsection{Reflection symmetry}
\label{app:reflection-symmetry}

\begin{proof}[Proof of \Cref{thm:reflection-symmetry}]
Let $\ppath\in\PathsDk{k^*}$ have position sequence
$\pi(\ppath)=(k_0,\dots,k_D)$.
The reflected sequence starts at $-k_0=0$, ends at
$-k_D=-k^*$, and each increment satisfies
\[
    (-k_d)-(-k_{d-1})
    =-(k_d-k_{d-1})\in\{-1,0,1\}.
\]
Also $|-k_d|=|k_d|\le d$, so each reflected vertex belongs to $\VerticesD$. Thus $\ppath'\in\PathsDk{-k^*}$. Reflecting twice recovers the original path, proving bijectivity. By \eqref{E:Kbounds}, the feasible bounds for the reflected endpoint are
\[
\begin{aligned}
    \left\lceil\frac{-k^*-D}{2}\right\rceil
    &=-\left\lfloor\frac{D+k^*}{2}\right\rfloor=-k_+,\\
    \left\lfloor\frac{D-k^*}{2}\right\rfloor
    &=-\left\lceil\frac{k^*-D}{2}\right\rceil=-k_-.
\end{aligned}
\]
For every integer $k\in[-k_+,-k_-]$,
\[
\begin{aligned}
    c_k(\ppath')
    &=\left|\{d\in\{0,\dots,D\}:-k_d=k\}\right|\\
    &=\left|\{d\in\{0,\dots,D\}:k_d=-k\}\right|\\
    &=c_{-k}(\ppath).
\end{aligned}
\]
Therefore the reflected tuple depends only on the original
tuple, not on the chosen realizing path.
It belongs to $\mathcal C_{D,-k^*}$ and has the form
\eqref{eq:tuple-reflection}. Reflecting its indices again
recovers the original tuple, proving that the induced
map is also bijective. The two cardinality equalities follow.
Finally, $k_d\ge0$ for every depth if and only if
$-k_d\le0$ for every depth. This proves the restriction
to nonnegative and nonpositive paths.
The argument also applies to $D=0$, when the root path
and its tuple are fixed by reflection.
\end{proof}

\subsection{The endpoint $k^*=D-2$}
\label{app:case-three-count}

\begin{proof}[Proof of \eqref{eq:case-three-positive}
and \eqref{eq:case-three-count}]
Fix $D\in\Z_+$ with $D\ge2$ and $k^*=D-2$.
Then $k_+=D-1$ and the positive-mass layer indices are
$i=1,2$, with slot counts $D-1,D-2$, respectively.
For $D\ge3$, their counts are
\[
    \binom{D-1}{D-2}=D-1,
    \qquad
    \binom{D-1}{D-3}=\binom{D-1}{2}.
\]
For $D=2$, layer $i=1$ contributes one tuple and
layer $i=2$ is empty, as its receiving-slot count is zero.
Thus, in every case, the positive count is
\[
\begin{aligned}
    1+(D-1)+\binom{D-1}{2}
    &=1+(D-1)+\frac{(D-1)(D-2)}{2}\\
    &=1+\frac{D(D-1)}{2}\\
    &=1+\binom{D}{2}.
\end{aligned}
\]
Here the leading one counts the seed. This proves
\eqref{eq:case-three-positive}. For the mixed contribution, let $j_+,j_-,j_0$ be the numbers of upward steps, downward steps, and stays, respectively, as in \eqref{E:algebra}. They satisfy
\[
    j_++j_-+j_0=D,
    \quad
    j_+-j_-=D-2.
\]
Subtracting the second equation from the first gives $2j_-+j_0=2.$ As the counts are nonnegative integers, the only
possibilities are $(j_-,j_0)=(0,2)$ or  $(j_-,j_0)=(1,0).$
In the first case, there are no downward steps, so no position is negative. In the second case, all steps except the unique downward step are upward. If the downward step occurs at step $d\in\{1,\dots,D\}$, the position immediately after it is $d-2$. Positions before it are nonnegative, and
positions afterward increase. Thus a negative position
is visited if and only if $d=1$. It follows that there is exactly one mixed path: the downward step is first, and all remaining steps are upward. This path visits $-1$ once, $0$ twice, and each position $1,\dots,D-2$ once. The last range
is empty when $D=2$. Its tuple is distinct from every positive tuple because its count at $-1$ is nonzero. Therefore
$|\mathcal C^-_{D,D-2}|=1.$ The disjoint split in \Cref{set_splitting} now gives $|\mathcal C_{D,D-2}| =\left(1+\binom{D}{2}\right)+1=\binom{D}{2}+2,$ proving \eqref{eq:case-three-count}.
\end{proof}

\section{Proof of the Exact Degeneracy Formula}
\label{app:trinomial-degeneracy}

We first record the fixed-transition counting argument, with the initial and terminal positions distinguished. It is the last-exit argument underlying Euler-trail enumeration and is closely related to the enumeration of sequences with prescribed transition totals \cite{vanAardenneDeBruijn1951,Whittle1955}. We give the proof to make all normalization factors explicit, including when the path returns to its starting position.

\begin{lemma}[Fixed-transition Euler-trail count]
\label{lem:fixed-transition-count}
Let $V$ be a finite nonempty vertex set, let $s,t\in V$, and let $N=(N_{ij})_{i,j\in V}$ specify nonnegative integer transition counts. An Euler trail is a directed walk that uses every edge exactly once. For each $i\in V$, write
\[
    n_i\Def\sum_{j\in V}N_{ij}
\]
for the outdegree of $i$; its indegree is $\sum_{j\in V}N_{ji}.$ A loop contributes one to both the outdegree and the indegree. Suppose
\begin{equation}
\label{eq:fixed-transition-balance}
    n_i-\sum_{j\in V}N_{ji}
    =\mathbf 1_{\{i=s\}}-\mathbf 1_{\{i=t\}},
    \quad i\in V.
\end{equation}

In the directed multigraph having $N_{ij}$ edges from $i$ to $j$, let $\operatorname{arb}_t(N)>0$ be the number of spanning
arborescences directed toward $t$. Such an arborescence chooses one outgoing edge at each $i\ne t$ and none at $t$, with every chosen route ending at $t$. Parallel edges are counted separately in $\operatorname{arb}_t(N)$. Then the number $\mathcal E_{s,t}(N)$ of vertex sequences starting
at $s$, ending at $t$, and using exactly these transition counts is
\begin{equation}
\label{eq:fixed-transition-best}
    \mathcal E_{s,t}(N)
    =\operatorname{arb}_t(N)
      \frac{\displaystyle\prod_{i\in V}
                     (n_i-1+\mathbf 1_{\{i=t\}})!}
           {\displaystyle\prod_{i,j\in V}N_{ij}!}.
\end{equation}
For $V=\{t\}$, the empty arborescence is counted once. A sequence is indexed from its prescribed starting position; cyclic rotations are not identified when $s=t$.
\end{lemma}

\begin{proof}
Label every edge separately, including parallel edges and loops. Fix an arborescence toward $t$. At each $i\ne t$, put its chosen arborescence edge last in a list of the outgoing edges, and order all other outgoing edges arbitrarily. At $t$, order all outgoing edges arbitrarily. Starting at $s$, repeatedly use the next unused edge in the list at the current vertex. This walk uses each edge at most once and eventually stops, say at $v$. The used transitions have outgoing-minus-incoming imbalance
$\mathbf 1_{\{i=s\}}-\mathbf 1_{\{i=v\}}$ at vertex $i$.
Subtracting from \eqref{eq:fixed-transition-balance}, the unused transitions therefore have imbalance
\[
    \bigl(\mathbf 1_{\{i=s\}}-\mathbf 1_{\{i=t\}}\bigr)
    -\bigl(\mathbf 1_{\{i=s\}}-\mathbf 1_{\{i=v\}}\bigr)
    =\mathbf 1_{\{i=v\}}-\mathbf 1_{\{i=t\}}.
\]
If $v\ne t$, this is $1$ at $v$, which is impossible because the walk stopped there with no unused outgoing edge. Thus $v=t$, and the unused transitions are balanced at every vertex. Suppose an unused edge remains. Its source has unused outgoing edges and is not $t$. Its arborescence edge is also unused because that edge was placed last. The next vertex on the arborescence route therefore has an unused incoming edge and, by balance, an unused outgoing edge. Repeating this argument along the route reaches $t$, contradicting that no outgoing edge remains there. Hence every edge is used: each choice of arborescence and local lists constructs a labeled Euler trail. Conversely, in any labeled Euler trail, mark the last outgoing edge used at each $i\ne t$. Such an edge exists because $\operatorname{arb}_t(N)>0$ implies $n_i\ge1$ for $i\ne t$. The marked edges cannot form a directed cycle. Indeed, the marked edge traversed last in that cycle would arrive at a vertex whose last exit had already occurred; the trail could neither continue nor finish at $t$. The same argument excludes a marked loop. Following the marked edges must therefore lead to $t$, so they form an arborescence. The outgoing-edge orders are recovered uniquely from the trail. The two constructions are inverse. For a fixed arborescence there are
\[
    \left(\prod_{i\ne t}(n_i-1)!\right)n_t!
    =\prod_{i\in V}(n_i-1+\mathbf 1_{\{i=t\}})!
\]
choices of local lists. Multiplication by $\operatorname{arb}_t(N)$ counts the
labeled trails. Each vertex sequence has exactly $\prod_{i,j\in V}N_{ij}!$ labelings: assign the $N_{ij}$ distinct labels to the occurrences of transition $i\to j$ in any order, independently for each pair. This includes the $N_{ii}!$ labelings of stays at $i$. Dividing gives \eqref{eq:fixed-transition-best}. In particular, when $s=t$ the terminal list still contributes $n_t!$; no additional
factor for a choice of starting time or cyclic rotation is present.
\end{proof}

\begin{proof}[Proof of \Cref{thm:exact-trinomial-degeneracy}]
Fix $\hat c\in\mathcal C_{D,k^*}$, with occupied interval $[a,b]$. All position ranges below contain their integer points.

\paragraph{Crossing counts and feasibility}
For a realizing path with $\pi(\ppath)=(k_0,\dots,k_D)$, define
\begin{equation}
\label{eq:transition-counts-up-down}
\begin{aligned}
    u_j&\Def\bigl|\{d\in\{1,\dots,D\}:
                         (k_{d-1},k_d)=(j,j+1)\}\bigr|,\\
    d_j&\Def\bigl|\{d\in\{1,\dots,D\}:
                         (k_{d-1},k_d)=(j+1,j)\}\bigr|,
                         \qquad a\le j<b,
\end{aligned}
\end{equation}
and
\begin{equation}
\label{eq:transition-counts-horizontal}
    h_j\Def\bigl|\{d\in\{1,\dots,D\}:k_{d-1}=k_d=j\}\bigr|,
    \qquad a\le j\le b.
\end{equation}
Use the zero boundary conventions from the main text. Across each interface, telescoping gives
\[
\begin{aligned}
    u_j-d_j
    &=\sum_{d=1}^{D}\left(
        \mathbf 1_{\{k_{d-1}\le j\}}-\mathbf 1_{\{k_d\le j\}}
      \right)\\
    &=\mathbf 1_{\{0\le j\}}-\mathbf 1_{\{k^*\le j\}}
      =\eta_j(k^*).
\end{aligned}
\]
Consequently $m_j=\min\{u_j,d_j\}$ gives the unique representation \eqref{eq:ud-m-parameterization}. Every visit to $j$, except the terminal visit when $j=k^*$, is followed by exactly one outgoing step. Thus
\begin{equation}
\label{eq:outdegree-from-occupation}
    h_j+u_j+d_{j-1}=c_j-\mathbf 1_{\{j=k^*\}},
\end{equation}
which gives \eqref{eq:h-from-c-and-m}. After visiting a position $j<k^*$ the path must cross $j\to j+1$
to finish at $k^*$. After visiting $j+1>k^*$ it must cross
$j+1\to j$. Hence the strict positivity conditions in
\eqref{eq:M-feasible} are necessary. Every realizing path therefore determines one vector in $\mathcal M(\hat c;k^*)$. For any vector in that set, nonnegative stays give
\[
    0\le m_j\le u_j(m)\le c_j,\qquad
    0\le m_j\le d_j(m)\le c_{j+1},
    \qquad a\le j<b.
\]
Thus $m_j\le\min\{c_j,c_{j+1}\}$ at each interface and the set of crossing vectors is finite. Conversely, fix $m\in\mathcal M(\hat c;k^*)$. Form a directed multigraph on the vertex set $[a,b]$ with transition matrix
\begin{equation}
\label{eq:transition-matrix-m}
\begin{aligned}
    N^{(m)}_{j,j}&=h_j(m) &&(a\le j\le b),\\
    N^{(m)}_{j,j+1}&=u_j(m),\qquad
    N^{(m)}_{j+1,j}=d_j(m) &&(a\le j<b),
\end{aligned}
\end{equation}
with all other entries equal to zero. Thus the only permitted
transitions are stays and crossings between adjacent occupied
positions. In the calculations below, we temporarily suppress the argument $m$ from $u_j(m)$, $d_j(m)$, and $h_j(m)$. Since $\eta_{a-1}(k^*)=\eta_b(k^*)=0$, the boundary conventions give, for every $a\le j\le b$,
\[
\begin{aligned}
    \sum_{q=a}^{b}N^{(m)}_{j,q}
       -\sum_{q=a}^{b}N^{(m)}_{q,j}
    &=(h_j+u_j+d_{j-1})-(h_j+u_{j-1}+d_j)\\
    &=(u_j-d_j)-(u_{j-1}-d_{j-1})\\
    &=\eta_j(k^*)-\eta_{j-1}(k^*)\\
    &=\mathbf 1_{\{j=0\}}-\mathbf 1_{\{j=k^*\}}.
\end{aligned}
\]
The outgoing and incoming totals are therefore, respectively,
\[
    \sum_{q=a}^{b}N^{(m)}_{j,q}=c_j-\mathbf 1_{\{j=k^*\}},
    \qquad
    \sum_{q=a}^{b}N^{(m)}_{q,j}=c_j-\mathbf 1_{\{j=0\}}.
\]
In particular, the total number of transitions is exactly
\[
    \sum_{j=a}^{b}\sum_{q=a}^{b}N^{(m)}_{j,q}
    =\sum_{j=a}^{b}(c_j-\mathbf 1_{\{j=k^*\}})
    =(D+1)-1=D.
\]
Choosing one upward edge from every position below $k^*$ and one downward edge from every position above $k^*$ gives an arborescence toward $k^*$. All required edges exist by \eqref{eq:M-feasible}. The fixed-transition lemma therefore supplies Euler trails from $0$ to $k^*$ using the prescribed transitions. Each trail has $D$ steps, all in $\{-1,0,1\}$. Restoring depth coordinates gives a path in $\PathsDk{k^*}$, since $|k_d|\le d$. Its visit count at $j$ is its outgoing count plus its possible terminal visit, namely
$(c_j-\mathbf 1_{\{j=k^*\}})+\mathbf 1_{\{j=k^*\}}=c_j.$ This proves sufficiency. Different crossing vectors specify
different transition counts, so the corresponding path families are disjoint and exhaust the fiber of $\hat c$.

\paragraph{Collecting the counting factors}
For fixed $m$, put $n_j=\sum_{q=a}^{b}N^{(m)}_{j,q}$.
The factorial numerator of \eqref{eq:fixed-transition-best} becomes
\[
\begin{aligned}
    \prod_{j=a}^{b}(n_j-1+\mathbf 1_{\{j=k^*\}})!
    &=\prod_{j=a}^{b}
       (c_j-\mathbf 1_{\{j=k^*\}}-1+\mathbf 1_{\{j=k^*\}})!\\
    &=\prod_{j=a}^{b}(c_j-1)!.
\end{aligned}
\]
The denominator factors by transition type as
\[
    \prod_{i=a}^{b}\prod_{j=a}^{b}N^{(m)}_{ij}!
    =\left(\prod_{j=a}^{b}h_j(m)!\right)
     \left(\prod_{j=a}^{b-1}u_j(m)!\,d_j(m)!\right),
\]
because all other transition counts are zero and $0!=1$.
Hence
\begin{equation}
\label{eq:fixed-m-degeneracy}
\mathcal E_{0,k^*}(N^{(m)})
=\operatorname{arb}_{k^*}(N^{(m)})
 \frac{\displaystyle\prod_{j=a}^{b}(c_j-1)!}
      {\displaystyle
       \left(\prod_{j=a}^{b}h_j(m)!\right)
       \left(\prod_{j=a}^{b-1}u_j(m)!\,d_j(m)!\right)}.
\end{equation}
In the underlying line, the only simple route from a position to $k^*$ points right below $k^*$ and left above $k^*$.
Every arborescence must therefore have this shape. The independent parallel-edge choices give
\[
\begin{aligned}
    \operatorname{arb}_{k^*}(N^{(m)})
    &=\left(\prod_{j=a}^{k^*-1}u_j(m)\right)
      \left(\prod_{j=k^*+1}^{b}d_{j-1}(m)\right)\\
    &=\left(\prod_{j=a}^{k^*-1}u_j(m)\right)
      \left(\prod_{j=k^*}^{b-1}d_j(m)\right).
\end{aligned}
\]
Substituting into \eqref{eq:fixed-m-degeneracy} and summing over the disjoint crossing families proves
\eqref{eq:degeneracy-explicit-finite-sum}. The factor
$\prod_{j=a}^{b}(c_j-1)!$ is independent of $m$, which justifies placing it outside the sum.

\paragraph{Determinant representation and its collapse}
Define the outgoing directed Laplacian by
\begin{equation}
\label{eq:directed-laplacian-m}
    (L_m)_{ij}=
    \begin{cases}
        -N^{(m)}_{ij},&i\ne j,\\[2pt]
        \displaystyle\sum_{\substack{q=a\\q\ne i}}^{b}
                     N^{(m)}_{iq},&i=j,
    \end{cases}
    \qquad a\le i,j\le b.
\end{equation}
Thus stays do not enter $L_m$. Let $L_m^{(k^*)}$ be obtained by deleting the row and column indexed by $k^*$. Its determinant is $1$ when the resulting matrix is empty. Deleting $k^*$ separates the matrix into blocks indexed by
$[a,k^*-1]$ and $[k^*+1,b]$. Before deletion, its nonzero
entries, apart from boundary zeros, are
\[
    (L_m)_{jj}=u_j+d_{j-1},\qquad
    (L_m)_{j,j+1}=-u_j,\qquad
    (L_m)_{j+1,j}=-d_j,
\]
where the argument $m$ is again suppressed.
For a nonempty left block, elimination begins with pivot $u_a$.
After pivot $u_{j-1}$ has been obtained, adding
$d_{j-1}/u_{j-1}$ times the preceding row to row $j$ removes
its subdiagonal entry and changes its diagonal to
\[
    u_j+d_{j-1}
       -\frac{d_{j-1}}{u_{j-1}}u_{j-1}=u_j.
\]
All pivots are positive by \eqref{eq:M-feasible}, so the left-block
determinant is $\prod_{j=a}^{k^*-1}u_j$.
For a nonempty right block, eliminate from the last row upward.
The first pivot is $d_{b-1}$; adding $u_j/d_j$ times row $j+1$
to row $j$ changes its diagonal to
\[
    d_{j-1}+u_j-\frac{u_j}{d_j}d_j=d_{j-1}.
\]
These pivots are also positive, and the right-block determinant is
$\prod_{j=k^*}^{b-1}d_j$. Row additions preserve determinants and
no row interchanges or scalings were used. Therefore
\begin{equation}
\label{eq:arborescence-product-collapse}
    \det L_m^{(k^*)}
    =\left(\prod_{j=a}^{k^*-1}u_j(m)\right)
     \left(\prod_{j=k^*}^{b-1}d_j(m)\right)
    =\operatorname{arb}_{k^*}(N^{(m)}).
\end{equation}
This verifies the required cofactor identity directly for the line.
Consequently the equivalent determinant form is
\begin{equation}
\label{eq:degeneracy-determinant-form}
\deg_1(\hat c;k^*)
=\left(\prod_{j=a}^{b}(c_j-1)!\right)
 \sum_{m\in\mathcal M(\hat c;k^*)}
 \frac{\det L_m^{(k^*)}}
      {\displaystyle
       \left(\prod_{j=a}^{b}h_j(m)!\right)
       \left(\prod_{j=a}^{b-1}u_j(m)!\,d_j(m)!\right)}.
\end{equation}
When $a=b$, the support contains both endpoints only if
$a=b=k^*=0$. Then $c_0=D+1$, $h_0=D$, and the sole crossing vector
is empty. Both formulas reduce to $D!/D!=1$, including $D=0$.
The argument otherwise applies without any sign restriction on
$k^*$.
\end{proof}

\begin{proof}[Proof of \eqref{eq:fiber-sum-raw-count}]
The fibers $\hat C_{D,k^*}^{-1}(\{\hat c\})$ form a disjoint
partition of $\PathsDk{k^*}$ by \eqref{full_card_set} and
\eqref{equivalence_relation}. Thus
\[
    |\PathsDk{k^*}|
    =\sum_{\hat c\in\mathcal C_{D,k^*}}
       \left|\hat C_{D,k^*}^{-1}(\{\hat c\})\right|
    =\sum_{\hat c\in\mathcal C_{D,k^*}}\deg_1(\hat c;k^*).
\]
Expanding one factor for each step gives
\[
    (z^{-1}+1+z)^D
    =\sum_{(s_1,\dots,s_D)\in\{-1,0,1\}^D}
        z^{s_1+\cdots+s_D}.
\]
The coefficient of $z^{k^*}$ counts exactly the increment sequences
ending at $k^*$, which correspond bijectively to the raw paths.
For $D=0$, the polynomial is $1$ and the only admissible endpoint
is $0$. This proves \eqref{eq:fiber-sum-raw-count}.
\end{proof}

\begin{proof}[Proof of \eqref{eq:exact-average-with-degeneracy}
and \eqref{eq:Np-from-degeneracy}]
For $\pi(\ppath)=(k_0,\dots,k_D)$, group the path sums by their
cardinality tuples and apply \eqref{E:efficientaggregation}:
\[
\begin{aligned}
    \sum_{\ppath\in\PathsDk{k^*}}\sum_{d=0}^{D}\vvalue_{k_d}
    &=\sum_{\hat c\in\mathcal C_{D,k^*}}
      \sum_{\ppath\in\hat C_{D,k^*}^{-1}(\{\hat c\})}
      \sum_{j=k_-}^{k_+}c_j\vvalue_j\\
    &=\sum_{\hat c\in\mathcal C_{D,k^*}}
      \deg_1(\hat c;k^*)\sum_{j=k_-}^{k_+}c_j\vvalue_j.
\end{aligned}
\]
Dividing by the positive count
$|\PathsDk{k^*}|=\sum_{\hat c}\deg_1(\hat c;k^*)$ proves
\eqref{eq:exact-average-with-degeneracy}.
For a fixed aggregate $p$, restrict the same partition to the
fibers satisfying $\sum_{j=k_-}^{k_+}c_j\vvalue_j=p$.
Each such fiber contributes its full cardinality
$\deg_1(\hat c;k^*)$, and their sum is exactly $N_p$.
This proves \eqref{eq:Np-from-degeneracy}, including the possibility
that different tuples have the same aggregate.
\end{proof}

\section{Proofs for the Motzkin Specialization and Occupation Profiles}
\label{app:motzkin-integration}

Throughout this appendix, the paths, their increments, and their
cardinality tuples retain the definitions in \Cref{setup} and
\eqref{card_map}. The statements below specialize the classical
first-return and weighted-path constructions
\cite{BanderierFlajolet2002,Flajolet1980} to these conventions.

\begin{proof}[Proof of \Cref{thm:motzkin-generation}]
For $\ppath\in\PosPathsDk{0}$, the increments
$s_d=k_d-k_{d-1}$ lie in $\{-1,0,1\}$, every partial sum is
nonnegative, and their total is zero. These are precisely the
conditions for the corresponding word in $\mathrm U,\mathrm H,
\mathrm D$ to be a Motzkin path. Conversely, the partial sums $k_d=\sum_{j=1}^{d}s_j$, with $k_0=0$ recover a nonnegative position sequence ending at zero. Restoring the depth coordinates gives its unique vertex path. Taking differences and partial sums are inverse operations, proving the bijection. Writing $\varepsilon$ for the empty word, $\mathfrak M_0=\{\varepsilon\}$ and $M_0=1$, with vertex-path counterpart $((0,0))$. For $D\ge1$, a Motzkin word either starts with $\mathrm H$ or with $\mathrm U$. In the first case, removing that horizontal step leaves an arbitrary word in $\mathfrak M_{D-1}$. In the second, consider its first subsequent return to height zero. That return step is $\mathrm D$. If the intervening word has length $r$, then $0\le r\le D-2$; translating its positions down by one gives a word in $\mathfrak M_r$. The remaining suffix is a word in $\mathfrak M_{D-2-r}$. Conversely, each such pair reconstructs a unique Motzkin word whose first return occurs at depth $r+2$. Hence
\begin{equation}
\label{eq:motzkin-disjoint-generation}
    \mathfrak M_D
    =\mathrm H\mathfrak M_{D-1}
      \;\sqcup\!\bigsqcup_{r=0}^{D-2}
        \mathrm U\mathfrak M_r\mathrm D\mathfrak M_{D-2-r}.
\end{equation}
Products here mean word concatenation; the interior path after
$\mathrm U$ is interpreted one level above the baseline.
The first letter and first return make this union disjoint, and the two cases exhaust all nonempty words. Taking cardinalities proves \eqref{eq:motzkin-recurrence}, including $D=1$ when the second union is empty. Multiplying the recurrence by $t^D$ and summing gives
\[
\begin{aligned}
    M(t)-1
    &=\sum_{D\ge1}M_{D-1}t^D
      +\sum_{D\ge2}\sum_{r=0}^{D-2}M_rM_{D-2-r}t^D\\
    &=tM(t)+t^2\sum_{r,j\ge0}M_rM_jt^{r+j}\\
    &=tM(t)+t^2M(t)^2.
\end{aligned}
\]
Each coefficient involves finitely many terms, so the rearrangement is valid for formal power series. The recurrence determines every coefficient from earlier ones, proving uniqueness given $M_0=1$. Solving $t^2M(t)^2+(t-1)M(t)+1=0$ gives
\[
    M(t)=\frac{1-t\pm\sqrt{1-2t-3t^2}}{2t^2}.
\]
The square root with constant term $1$ begins
$1-t-2t^2+\cdots$. The minus sign therefore gives a power series with constant term $1$; the plus sign gives a nonzero $t^{-2}$ term. This proves \eqref{eq:motzkin-generating-function}. To derive the coefficient formula, fix $r$ upward steps. The endpoint condition requires $r$ downward steps and $D-2r$ stays. Deleting the stays leaves a nonnegative up--down word of length $2r$ ending at zero. Among all $\binom{2r}{r}$ balanced up--down words, those that visit $-1$ are in bijection with words having $r+1$ upward and $r-1$ downward steps: interchange up and down in the prefix through the first visit to $-1$. The inverse interchanges the prefix through the first visit to $+1$ in the resulting word, which ends at $2$. Thus, for $r\ge1$, the number of nonnegative balanced words is
\[
\begin{aligned}
    \binom{2r}{r}-\binom{2r}{r-1}
    &=\binom{2r}{r}\left(1-\frac{r}{r+1}\right)\\
    &=\frac{1}{r+1}\binom{2r}{r}.
\end{aligned}
\]
For $r=0$, the same value is $1$, counting the empty word. Choose which $2r$ of the $D$ positions hold the non-horizontal steps, in $\binom{D}{2r}$ ways, and fill the others with stays. Reinsertion preserves nonnegativity, and deletion recovers both choices uniquely. Summing these disjoint cases proves \eqref{eq:motzkin-explicit-count}.
\end{proof}

\begin{proof}[Proof of \Cref{cor:dyck-specialization}]
Give each horizontal letter the additional factor $q$. In the
first-return decomposition, the initial horizontal step then
contributes $qt$, whereas the initial up-step and its matching
down-step together contribute $t^2$. The two remaining words have independent weights. Consequently
\[
    M(t,q)-1=qtM(t,q)+t^2M(t,q)^2.
\]
The constant term is $1$. Comparing coefficients determines a unique
series in $\mathbb Q[q][[t]]$. Solving the quadratic and
choosing the square root with constant term $1$ proves
\eqref{eq:motzkin-bivariate-solution}. The same insertion argument as above also gives
\[
    [t^D]M(t,q)=\sum_{r=0}^{\lfloor D/2\rfloor}
        \binom{D}{2r}\frac{1}{r+1}\binom{2r}{r}q^{D-2r}.
\]
At $q=0$, only zero-stay paths contribute. For $D=2r$ their count is $\binom{2r}{r}/(r+1)$, and for odd $D$ it is zero. Substitution in the quadratic solution gives \eqref{eq:dyck-generating-function}. At $q=1$ all horizontal weights are discarded, recovering $M(t)$.
\end{proof}

\begin{proof}[Proof of \Cref{thm:occupation-profile-j-fraction}]
Fix $h\ge0$. At each degree $D$, the definition of $E_h$ sums over finitely many words. Such a translated excursion reaches no higher than $h+\lfloor D/2\rfloor$, since every upward displacement must be undone before its return. Each coefficient is therefore a polynomial in finitely many height variables. The empty word contributes $1$. If an excursion starts with a horizontal step at height $h$, the new visit contributes $ty_h$. The remaining excursion contributes $E_h$ and excludes its own initial visit, which has already been counted. This case contributes $ty_hE_h$. Otherwise, take the first return to $h$ after the initial up-step. The up-step contributes $ty_{h+1}$, the interior excursion contributes
$E_{h+1}$, the matching down-step contributes $ty_h$, and the final excursion contributes $E_h$. Their product is exactly
\[
    (ty_{h+1})E_{h+1}(ty_h)E_h
    =t^2y_hy_{h+1}E_{h+1}E_h.
\]
The initial visit of each of the two subexcursions is excluded, so
neither the arrival at $h+1$ nor the return to $h$ is counted twice.
These are the same disjoint first-return cases as before. Adding
them proves \eqref{eq:Eh-functional-equation}.
Rearranging gives
\[
    \bigl(1-ty_h-t^2y_hy_{h+1}E_{h+1}\bigr)E_h=1.
\]
The factor in parentheses has constant term $1$ and is therefore
invertible as a formal series, proving \eqref{eq:Eh-resolved}.

The equations determine the entire family uniquely. Indeed,
$[t^0]E_h=1$ for every $h$, and, for $D\ge1$,
\[
\begin{aligned}
    [t^D]E_h
    &=y_h[t^{D-1}]E_h\\
    &\quad+y_hy_{h+1}\sum_{r=0}^{D-2}
          ([t^r]E_{h+1})([t^{D-2-r}]E_h).
\end{aligned}
\]
Every coefficient on the right has degree smaller than $D$.
Induction on $D$, simultaneously for all $h$, proves uniqueness. To justify the infinite fraction, temporarily impose a height ceiling $b\ge0$. At height $b$, the restricted excursion can only consist
of stays, so its series is $1/(1-ty_b)$. Applying
\eqref{eq:Eh-resolved} successively at heights $b-1,\dots,0$, and
multiplying by $y_0$, gives the finite continued fraction ending in
$1-ty_b$. It counts exactly the root-based excursions whose highest
position is at most $b$. An excursion reaching $b+1$ and returning
to zero needs at least $2b+2$ steps. Thus this truncation agrees with
$F^+(\mathbf y;t)$ through degree $2b+1$; in particular, every
coefficient of degree $D$ is exact once $b\ge\lfloor D/2\rfloor$.
The finite fractions therefore stabilize coefficientwise and yield
\eqref{eq:occupation-profile-j-fraction}. No analytic convergence
assumption is needed. By \eqref{eq:Fplus-definition}, restoring the initial root visit gives
\[
\begin{aligned}
    F^+(\mathbf y;t)
    &=\sum_{D\ge0}t^D
      \sum_{\ppath\in\PosPathsDk{0}}
            y_0\prod_{d=1}^{D}y_{k_d}\\
    &=\sum_{D\ge0}t^D
      \sum_{\ppath\in\PosPathsDk{0}}
            \prod_{h\ge0}y_h^{c_h(\ppath)}.
\end{aligned}
\]
The exponents count all $D+1$ visits. In particular, the degree-zero
term is $y_0$, corresponding to the sole root visit, not $1$.
By \Cref{set_splitting}, the image of $\PosPathsDk{0}$ under
$\hat C_{D,0}$ is exactly $\mathcal C^+_{D,0}$. Moreover, every
realizing path of a tuple in $\mathcal C^+_{D,0}$ is nonnegative:
its visit counts at all negative positions vanish. Hence its entire
fiber in $\PathsDk{0}$ lies in $\PosPathsDk{0}$, and the fiber size is
precisely $\deg_1(\hat c;0)$ from \Cref{def:path-degeneracy}.
Grouping the finite depth-$D$ sum by those fibers gives
\[
\begin{aligned}
    \sum_{\ppath\in\PosPathsDk{0}}\prod_{h\ge0}y_h^{c_h(\ppath)}
    &=\sum_{\hat c\in\mathcal C^+_{D,0}}
      \left|\hat C_{D,0}^{-1}(\{\hat c\})\right|
      \prod_{h\ge0}y_h^{c_h}\\
    &=\sum_{\hat c\in\mathcal C^+_{D,0}}
      \deg_1(\hat c;0)\prod_{h\ge0}y_h^{c_h}.
\end{aligned}
\]
This proves \eqref{eq:profile-series-degeneracy-expansion}.
Distinct tuples give distinct monomials in the algebraically
independent variables $y_h$. Extracting the indicated monomial at
degree $D$ proves \eqref{eq:profile-coefficient-degeneracy}.
\end{proof}

\begin{proof}[Proof of \Cref{cor:motzkin-from-degeneracies}]
Each coefficient is a polynomial, so setting all height variables
to $1$ is well defined coefficient by coefficient. Every raw path
then contributes $1$, giving $F^+(\mathbf 1;t)=M(t)$.
Applying the same specialization to
\eqref{eq:profile-series-degeneracy-expansion} yields
\[
    M(t)=\sum_{D\ge0}t^D
         \sum_{\hat c\in\mathcal C^+_{D,0}}\deg_1(\hat c;0).
\]
Comparison of coefficients proves \eqref{eq:motzkin-fiber-sum}.
At $D=0$, both sides equal $1$: the only tuple has $c_0=1$ and
its degeneracy is $1$.
\end{proof}

\section*{Acknowledgments}
I thank Manav Vora, Ryan Roach, and Pingbang Hu of the University of Illinois Urbana-Champaign for helpful feedback on the manuscript and several mathematical proofs. William Cosley of the University of Northern Iowa helped me prepare Appendix A and develop an initial visualization of the traditional path-enumeration approach. I am also grateful to Dāniels Ponamarjovs of the College of Alberta, Latvia, for guidance on C++ coding practices during the early stages of this work.

\end{document}

%% file: treefigure.tex
\begin{figure}[H]
\label{general_tree}
\centering
\begin{minipage}{0.3\textwidth}
\centering
\begin{tikzpicture}[scale=0.2, transform shape, >=Stealth,
  every node/.style={circle, draw, minimum size=1cm, font=\small},
  node distance=2cm and 1.2cm]
    \node (n0) at (0,0) {};
    
    \node (n1u) at (3,  2) {};
    \node (n1m) at (3,  0) {};
    \node (n1d) at (3, -2) {};
    
    \node (n2uu) at (6,  4) {};
    \node (n2um) at (6,  2) {};
    \node (n2mm) at (6,  0) {};
    \node (n2dm) at (6, -2) {};
    \node (n2dd) at (6, -4) {};
    
    \node (n3uuu) at (9,  6) {};
    \node (n3uum) at (9,  4) {};
    \node (n3umm) at (9,  2) {};
    \node (n3mmm) at (9,  0) {};
    \node (n3dmm) at (9, -2) {};
    \node (n3ddm) at (9, -4) {};
    \node (n3ddd) at (9, -6) {};
    
    \node (n4uuuu) at (12,  8) {};
    \node (n4uuum) at (12,  6) {};
    \node (n4uumm) at (12,  4) {};
    \node (n4ummm) at (12,  2) {};
    \node (n4mmmm) at (12,  0) {};
    \node (n4dmmm) at (12, -2) {};
    \node (n4ddmm) at (12, -4) {};
    \node (n4dddm) at (12, -6) {};
    \node (n4dddd) at (12, -8) {};
    
    \draw[red, thick] (n0) -- (n1u);
    \draw (n0) -- (n1m);
    \draw (n0) -- (n1d);
    
    \draw[red, thick] (n1u) -- (n2uu);
    \draw (n1u) -- (n2um);
    \draw (n1u) -- (n2mm);
    
    \draw (n1m) -- (n2um);
    \draw (n1m) -- (n2mm);
    \draw (n1m) -- (n2dm);
    
    \draw (n1d) -- (n2mm);
    \draw (n1d) -- (n2dm);
    \draw (n1d) -- (n2dd);
    
    \draw[red, thick] (n2uu) -- (n3uuu);
    \draw (n2uu) -- (n3uum);
    \draw (n2uu) -- (n3umm);
    
    \draw (n2um) -- (n3uum);
    \draw (n2um) -- (n3umm);
    \draw (n2um) -- (n3mmm);
    
    \draw (n2mm) -- (n3umm);
    \draw (n2mm) -- (n3mmm);
    \draw (n2mm) -- (n3dmm);
    
    \draw (n2dm) -- (n3mmm);
    \draw (n2dm) -- (n3dmm);
    \draw (n2dm) -- (n3ddm);
    
    \draw (n2dd) -- (n3dmm);
    \draw (n2dd) -- (n3ddm);
    \draw (n2dd) -- (n3ddd);
    
    \draw[red, thick] (n3uuu) -- (n4uuuu);
    \draw (n3uuu) -- (n4uuum);
    \draw (n3uuu) -- (n4uumm);
    
    \draw (n3uum) -- (n4uuum);
    \draw (n3uum) -- (n4uumm);
    \draw (n3uum) -- (n4ummm);
    
    \draw (n3umm) -- (n4uumm);
    \draw (n3umm) -- (n4ummm);
    \draw (n3umm) -- (n4mmmm);
    
    \draw (n3mmm) -- (n4ummm);
    \draw (n3mmm) -- (n4mmmm);
    \draw (n3mmm) -- (n4dmmm);
    
    \draw (n3dmm) -- (n4mmmm);
    \draw (n3dmm) -- (n4dmmm);
    \draw (n3dmm) -- (n4ddmm);
    
    \draw (n3ddm) -- (n4dmmm);
    \draw (n3ddm) -- (n4ddmm);
    \draw (n3ddm) -- (n4dddm);
    
    \draw (n3ddd) -- (n4ddmm);
    \draw (n3ddd) -- (n4dddm);
    \draw (n3ddd) -- (n4dddd);
\end{tikzpicture}
\caption*{$k^{*} = D, \, \, D \geq 0$}
\end{minipage}
\hspace{0.1em}
\begin{minipage}{0.3\textwidth}
\centering
\begin{tikzpicture}[scale=0.2, transform shape, >=Stealth,
  every node/.style={circle, draw, minimum size=1cm, font=\small},
  node distance=2cm and 1.2cm]
    \node (n0) at (0,0) {};
    
    \node (n1u) at (3,  2) {};
    \node (n1m) at (3,  0) {};
    \node (n1d) at (3, -2) {};
    
    \node (n2uu) at (6,  4) {};
    \node (n2um) at (6,  2) {};
    \node (n2mm) at (6,  0) {};
    \node (n2dm) at (6, -2) {};
    \node (n2dd) at (6, -4) {};
    
    \node (n3uuu) at (9,  6) {};
    \node (n3uum) at (9,  4) {};
    \node (n3umm) at (9,  2) {};
    \node (n3mmm) at (9,  0) {};
    \node (n3dmm) at (9, -2) {};
    \node (n3ddm) at (9, -4) {};
    \node (n3ddd) at (9, -6) {};
    
    \node (n4uuuu) at (12,  8) {};
    \node (n4uuum) at (12,  6) {};
    \node (n4uumm) at (12,  4) {};
    \node (n4ummm) at (12,  2) {};
    \node (n4mmmm) at (12,  0) {};
    \node (n4dmmm) at (12, -2) {};
    \node (n4ddmm) at (12, -4) {};
    \node (n4dddm) at (12, -6) {};
    \node (n4dddd) at (12, -8) {};
    
    \draw[cyan, thick] (n0) -- (n1u);
    \draw[cyan, thick] (n0) -- (n1m);
    \draw (n0) -- (n1d);
    
    \draw[cyan, thick] (n1u) -- (n2uu);
    \draw (n1u) -- (n2um);
    \draw (n1u) -- (n2mm);
    
    \draw[cyan, thick] (n1m) -- (n2um);
    \draw (n1m) -- (n2mm);
    \draw (n1m) -- (n2dm);
    
    \draw (n1d) -- (n2mm);
    \draw (n1d) -- (n2dm);
    \draw (n1d) -- (n2dd);
    
    \draw[cyan, thick] (n2uu) -- (n3uuu);
    \draw (n2uu) -- (n3uum);
    \draw (n2uu) -- (n3umm);
    
    \draw[cyan, thick] (n2um) -- (n3uum);
    \draw (n2um) -- (n3umm);
    \draw (n2um) -- (n3mmm);
    
    \draw (n2mm) -- (n3umm);
    \draw (n2mm) -- (n3mmm);
    \draw (n2mm) -- (n3dmm);
    
    \draw (n2dm) -- (n3mmm);
    \draw (n2dm) -- (n3dmm);
    \draw (n2dm) -- (n3ddm);
    
    \draw (n2dd) -- (n3dmm);
    \draw (n2dd) -- (n3ddm);
    \draw (n2dd) -- (n3ddd);
    
    \draw (n3uuu) -- (n4uuuu);
    \draw[cyan, thick] (n3uuu) -- (n4uuum);
    \draw (n3uuu) -- (n4uumm);
    
    \draw[cyan, thick] (n3uum) -- (n4uuum);
    \draw (n3uum) -- (n4uumm);
    \draw (n3uum) -- (n4ummm);
    
    \draw (n3umm) -- (n4uumm);
    \draw (n3umm) -- (n4ummm);
    \draw (n3umm) -- (n4mmmm);
    
    \draw (n3mmm) -- (n4ummm);
    \draw (n3mmm) -- (n4mmmm);
    \draw (n3mmm) -- (n4dmmm);
    
    \draw (n3dmm) -- (n4mmmm);
    \draw (n3dmm) -- (n4dmmm);
    \draw (n3dmm) -- (n4ddmm);
    
    \draw (n3ddm) -- (n4dmmm);
    \draw (n3ddm) -- (n4ddmm);
    \draw (n3ddm) -- (n4dddm);
    
    \draw (n3ddd) -- (n4ddmm);
    \draw (n3ddd) -- (n4dddm);
    \draw (n3ddd) -- (n4dddd);
\end{tikzpicture}
\caption*{$k^{*} = D - 1, \, \, D - 1 \geq 0$}
\end{minipage}
\begin{minipage}{0.3\textwidth}
\centering
\begin{tikzpicture}[scale=0.2, transform shape, >=Stealth,
  every node/.style={circle, draw, minimum size=1cm, font=\small},
  node distance=2cm and 1.2cm]
    \node (n0) at (0,0) {};
    
    \node (n1u) at (3,  2) {};
    \node (n1m) at (3,  0) {};
    \node (n1d) at (3, -2) {};
    
    \node (n2uu) at (6,  4) {};
    \node (n2um) at (6,  2) {};
    \node (n2mm) at (6,  0) {};
    \node (n2dm) at (6, -2) {};
    \node (n2dd) at (6, -4) {};
    
    \node (n3uuu) at (9,  6) {};
    \node (n3uum) at (9,  4) {};
    \node (n3umm) at (9,  2) {};
    \node (n3mmm) at (9,  0) {};
    \node (n3dmm) at (9, -2) {};
    \node (n3ddm) at (9, -4) {};
    \node (n3ddd) at (9, -6) {};
    
    \node (n4uuuu) at (12,  8) {};
    \node (n4uuum) at (12,  6) {};
    \node (n4uumm) at (12,  4) {};
    \node (n4ummm) at (12,  2) {};
    \node (n4mmmm) at (12,  0) {};
    \node (n4dmmm) at (12, -2) {};
    \node (n4ddmm) at (12, -4) {};
    \node (n4dddm) at (12, -6) {};
    \node (n4dddd) at (12, -8) {};
    
    \draw[green, thick] (n0) -- (n1u);
    \draw (n0) -- (n1m);
    \draw[green, thick] (n0) -- (n1d);
    
    \draw[green, thick] (n1u) -- (n2uu);
    \draw (n1u) -- (n2um);
    \draw (n1u) -- (n2mm);
    
    \draw (n1m) -- (n2um);
    \draw (n1m) -- (n2mm);
    \draw (n1m) -- (n2dm);
    
    \draw[green, thick] (n1d) -- (n2mm);
    \draw (n1d) -- (n2dm);
    \draw (n1d) -- (n2dd);
    
    \draw[green, thick] (n2uu) -- (n3uuu);
    \draw (n2uu) -- (n3uum);
    \draw (n2uu) -- (n3umm);
    
    \draw (n2um) -- (n3uum);
    \draw (n2um) -- (n3umm);
    \draw (n2um) -- (n3mmm);
    
    \draw[green, thick] (n2mm) -- (n3umm);
    \draw (n2mm) -- (n3mmm);
    \draw (n2mm) -- (n3dmm);
    
    \draw (n2dm) -- (n3mmm);
    \draw (n2dm) -- (n3dmm);
    \draw (n2dm) -- (n3ddm);
    
    \draw (n2dd) -- (n3dmm);
    \draw (n2dd) -- (n3ddm);
    \draw (n2dd) -- (n3ddd);
    
    \draw (n3uuu) -- (n4uuuu);
    \draw (n3uuu) -- (n4uuum);
    \draw[green, thick] (n3uuu) -- (n4uumm);
    
    \draw (n3uum) -- (n4uuum);
    \draw (n3uum) -- (n4uumm);
    \draw (n3uum) -- (n4ummm);
    
    \draw[green, thick] (n3umm) -- (n4uumm);
    \draw (n3umm) -- (n4ummm);
    \draw (n3umm) -- (n4mmmm);
    
    \draw (n3mmm) -- (n4ummm);
    \draw (n3mmm) -- (n4mmmm);
    \draw (n3mmm) -- (n4dmmm);
    
    \draw (n3dmm) -- (n4mmmm);
    \draw (n3dmm) -- (n4dmmm);
    \draw (n3dmm) -- (n4ddmm);
    
    \draw (n3ddm) -- (n4dmmm);
    \draw (n3ddm) -- (n4ddmm);
    \draw (n3ddm) -- (n4dddm);
    
    \draw (n3ddd) -- (n4ddmm);
    \draw (n3ddd) -- (n4dddm);
    \draw (n3ddd) -- (n4dddd);
\end{tikzpicture}
\caption*{$k^{*} = D - 2, \, \, D - 2 \geq 0$}
\end{minipage}
\hspace{0.1em}
\begin{minipage}{0.3\textwidth}
\centering
\begin{tikzpicture}[scale=0.2, transform shape, >=Stealth,
  every node/.style={circle, draw, minimum size=1cm, font=\small},
  node distance=2cm and 1.2cm]
    \node (n0) at (0,0) {};
    
    \node (n1u) at (3,  2) {};
    \node (n1m) at (3,  0) {};
    \node (n1d) at (3, -2) {};
    
    \node (n2uu) at (6,  4) {};
    \node (n2um) at (6,  2) {};
    \node (n2mm) at (6,  0) {};
    \node (n2dm) at (6, -2) {};
    \node (n2dd) at (6, -4) {};
    
    \node (n3uuu) at (9,  6) {};
    \node (n3uum) at (9,  4) {};
    \node (n3umm) at (9,  2) {};
    \node (n3mmm) at (9,  0) {};
    \node (n3dmm) at (9, -2) {};
    \node (n3ddm) at (9, -4) {};
    \node (n3ddd) at (9, -6) {};
    
    \node (n4uuuu) at (12,  8) {};
    \node (n4uuum) at (12,  6) {};
    \node (n4uumm) at (12,  4) {};
    \node (n4ummm) at (12,  2) {};
    \node (n4mmmm) at (12,  0) {};
    \node (n4dmmm) at (12, -2) {};
    \node (n4ddmm) at (12, -4) {};
    \node (n4dddm) at (12, -6) {};
    \node (n4dddd) at (12, -8) {};
    
    \draw[magenta, thick] (n0) -- (n1u);
    \draw (n0) -- (n1m);
    \draw[magenta, thick] (n0) -- (n1d);
    
    \draw[magenta, thick] (n1u) -- (n2uu);
    \draw (n1u) -- (n2um);
    \draw (n1u) -- (n2mm);
    
    \draw (n1m) -- (n2um);
    \draw (n1m) -- (n2mm);
    \draw (n1m) -- (n2dm);
    
    \draw (n1d) -- (n2mm);
    \draw (n1d) -- (n2dm);
    \draw[magenta, thick] (n1d) -- (n2dd);
    
    \draw (n2uu) -- (n3uuu);
    \draw (n2uu) -- (n3uum);
    \draw[magenta, thick] (n2uu) -- (n3umm);
    
    \draw (n2um) -- (n3uum);
    \draw (n2um) -- (n3umm);
    \draw (n2um) -- (n3mmm);
    
    \draw (n2mm) -- (n3umm);
    \draw (n2mm) -- (n3mmm);
    \draw (n2mm) -- (n3dmm);
    
    \draw (n2dm) -- (n3mmm);
    \draw (n2dm) -- (n3dmm);
    \draw (n2dm) -- (n3ddm);
    
    \draw[magenta, thick] (n2dd) -- (n3dmm);
    \draw (n2dd) -- (n3ddm);
    \draw (n2dd) -- (n3ddd);
    
    \draw (n3uuu) -- (n4uuuu);
    \draw (n3uuu) -- (n4uuum);
    \draw (n3uuu) -- (n4uumm);
    
    \draw (n3uum) -- (n4uuum);
    \draw (n3uum) -- (n4uumm);
    \draw (n3uum) -- (n4ummm);
    
    \draw (n3umm) -- (n4uumm);
    \draw (n3umm) -- (n4ummm);
    \draw[magenta, thick] (n3umm) -- (n4mmmm);
    
    \draw (n3mmm) -- (n4ummm);
    \draw (n3mmm) -- (n4mmmm);
    \draw (n3mmm) -- (n4dmmm);
    
    \draw[magenta, thick] (n3dmm) -- (n4mmmm);
    \draw (n3dmm) -- (n4dmmm);
    \draw (n3dmm) -- (n4ddmm);
    
    \draw (n3ddm) -- (n4dmmm);
    \draw (n3ddm) -- (n4ddmm);
    \draw (n3ddm) -- (n4dddm);
    
    \draw (n3ddd) -- (n4ddmm);
    \draw (n3ddd) -- (n4dddm);
    \draw (n3ddd) -- (n4dddd);
\end{tikzpicture}
\caption*{$k^{*} = 0$}
\end{minipage}
\caption{Recombining Tree Paths}
\label{fig:tree-four-panel}
\end{figure}